%% file: multiFactor.tex
\documentclass[table,11pt]{article}
\usepackage{amsmath}
\usepackage{graphicx}
\usepackage{psfrag}
\usepackage{epsf}
\usepackage{enumerate}

\usepackage{wrapfig,booktabs}

\usepackage[round,longnamesfirst]{natbib}
\usepackage{etoolbox}
\usepackage{mleftright}

\usepackage{xr}
\makeatletter
\newcommand\bibstyle@comma{\bibpunct(),a,,}
\newcommand\bibstyle@semicolon{\bibpunct();a,,}
\makeatother

\pretocmd\cite{\citestyle{comma}}\relax\relax
\pretocmd\citep{\citestyle{semicolon}}\relax\relax

\usepackage{url} 
\usepackage{amsthm}
\usepackage{xpatch}
\usepackage{subcaption}
\usepackage{bm}

\usepackage{longtable}
\usepackage{pdflscape}
\usepackage{mathrsfs}
\usepackage{tikz}
\usepackage{color} 
\usepackage{amssymb}
\usepackage{latexsym}
\usepackage{float}
\usepackage{multirow}
\usepackage{dsfont}
\usepackage{enumitem}
\usepackage{booktabs}
\usepackage{titling}

\usepackage[ruled]{algorithm2e}
\usepackage{lscape}
\usepackage{rotating}
\numberwithin{equation}{section}

\usepackage{breakcites}
\usepackage{xcolor}
\definecolor{lightblue}{HTML}{044E9E}
\usepackage[breaklinks=true, hyperfootnotes = true,colorlinks,citecolor=lightblue,urlcolor=lightblue, linkcolor=lightblue]{hyperref}

\definecolor{lightpurple}{HTML}{9C8FC3}

\def\RR{\mathbb R}
\def\ZZ{\mathbb Z}

\newcommand{\Var}{\operatorname{Var}} 
\newcommand{\rank}{\operatorname{rk}} 
\newcommand{\vecop}{\operatorname{vec}} 
\newcommand{\vechop}{\operatorname{vech}}

\newcommand{\diag}{\operatorname{diag}} 
 
\newcommand{\tr}{\operatorname{tr}} 
 
\newcommand{\E}{\operatorname{E}} 

\newcommand{\col}{\operatorname{col}}

\newcommand{\distr}{\operatorname{d}} 
\newcommand{\prob}{\operatorname{p}} 
\newcommand{\Dim}{d} 

\newcommand{\npf}{\operatorname{r}_p}

\usepackage{mleftright}
\usepackage{xparse}

\NewDocumentCommand{\evaluat}{sO{\big}mm}{%
  \IfBooleanTF{#1}
   {\mleft. #3 \mright|_{#4}}
   {#3#2|_{#4}}%
}

\usepackage{hhline} 
\usepackage{highlight}
\renewcommand{\opacity}{50}

\newcommand{\md}[1]{{\color{black} #1}}

\newtheorem{proposition}{Proposition}[section]
\newtheorem{lemma}{Lemma}[section]
\newtheorem{corollary}{Corollary}[section]

\theoremstyle{definition}
\newtheorem{remark}{Remark}[section]
\newtheorem{example}{Example}[section]

\newtheorem{Massump}{Assumption}

\newtheorem{CRassump}{Assumption}

\newtheorem{Wassump}{Assumption}

\newtheorem{Passump}{Assumption}

\makeatletter
\xpatchcmd{\proof}{\@addpunct{.}}{\@addpunct{:}}{}{}
\makeatother

\DeclareFontFamily{U}{mathx}{\hyphenchar\font45}
\DeclareFontShape{U}{mathx}{m}{n}{<-> mathx10}{}
\DeclareSymbolFont{mathx}{U}{mathx}{m}{n}
\DeclareMathAccent{\widebar}{0}{mathx}{"73}

\newcommand{\mockalph}[1]{}

\allowdisplaybreaks

\def\spacingset#1{\renewcommand{\baselinestretch}%
{#1}\small\normalsize} \spacingset{1}

\newtheorem*{assumptionBIC*}{\assumptionnumber}
\providecommand{\assumptionnumber}{}


\begin{document}

\title{Testing common structure in high-dimensional factor models: change-point and two-sample procedures%
\
\
\thanks{The authors would like to thank the participants of the NBER-NSF Time Series Conference 2023 in Montreal for their helpful comments and suggestions as well as the two anonymous reviewers and the associate editor. Marie D\"uker's research was supported by NSF grant DMS-1934985, and Vladas Pipiras acknowledges NSF grants DMS-2113662 and DMS-2134107.}
}
\author{
Marie-Christine D\"uker \\ Technical University of Munich                             \and
Vladas Pipiras               \\ University of North Carolina}
\date{\today}

\maketitle

\bigskip

\begin{abstract}
\noindent
This work proposes a novel procedure to test for common structures across two high-dimensional factor models.
The introduced test allows to uncover whether two factor models are driven by the same loading matrix up to some linear transformation. 
The test can be used to discover inter-individual relationships between two datasets. 
In addition, it can be applied to test for structural changes over time in the loading matrix of an individual factor model. The test aims to reduce the set of possible alternatives in a classical change-point setting.
The theoretical results establish the asymptotic behavior of the introduced test statistic. The theory is supported by a simulation study showing promising results in empirical test size and power. 
Two real data applications are considered: the first investigates changes in the loadings of the celebrated US macroeconomic dataset of Stock and Watson, and the second examines similarities of the loadings of macroeconomic indicators for the US and South Korea.
\end{abstract}

\section{Introduction}
High-dimensional factor models are used to describe data which are driven by a relatively small number of latent time series. It is assumed that the temporal and cross-sectional dependence in the data can be attributed to these so-called common factors. The co-movements of the observed component time series are assumed to be driven by these factors and their dynamic structures.
This kind of behavior is reasonable to assume in many applications \citep{sargent1977business} and provides an effective tool for dimension reduction in high-dimensional time series. Originally introduced in finance \citep{Markowitz,cox1976survey,ross1977capital,feng2020taming}, factor models have found their way into a wide range of disciplines including economics \citep{stock2002forecasting,stock2009forecasting,han2018estimation}, psychometrics \citep{timmerman2003four,song2014analyzing} and genomics \citep{carvalho2008high}.
\par
This work focuses on discovering common structures in two high-dimensional factor models from their respective time series data. In many of the disciplines mentioned above, one is not only interested in studying intra-individual differences for one factor model but also inter-individual similarities across factor models. Examples include resemblances between macroeconomic indices of different countries or behavioral/neurological measurements for different individuals.
While it is unlikely for two high-dimensional series to share the exact same loading matrix, one can ask whether they share a common structure in the form of linear combinations of loading matrices. We propose a novel procedure allowing one to test whether the loadings in one series can be expressed as a linear combination of the other.
\par
The questions of interest in high-dimensional factor model analysis are manyfold and have led to several lines of research in the statistics/econometrics literature. 
In particular, the estimation of the loading matrix and the latent factors is well understood and their estimators' theoretical properties have been studied in several works; see \cite{bai2003inferential, bai2008large, lam2011estimation, doz2011two, doz2012quasi}. Other questions of interest include estimating the number of factors \citep{bai2002determining,onatski2009testing,lam2012factor,li2017determining}, high-dimensional covariance estimation and sparsity \citep{fan2008high, fan2023bridging}, using a factor model approach to describe multivariate count data \citep{wedel2003factor,jung2011dynamic,wang2018modelling} and extensions to networks \citep{brauning2016dynamic,brauning2020dynamic}.
\par
The questions considered in this work are closely related to change-point analysis in high-dimensional factor models. Recent literature has proposed ways to test for changes in the loadings over time. Changes in the loading matrix of a factor model are observationally equivalent to changes in the second-order structure of their factors. 
This way, a high-dimensional change-point problem can be reduced to a low-dimensional one.
Han and Inoue \citep{han2015tests,han2015tests_WP} proposed a Wald-type test, and \cite{bai2022likelihood} proposed a likelihood ratio test.
Similar approaches have been pursued in \cite{chen2014detecting,ma2018estimation,duan2023quasi}. More recently, \cite{barigozzi2018simultaneous} studied the problem of multiple change-points and proposed a wavelet-based approach.
\cite{baltagi2021estimating} generalized the method in \cite{han2015tests} to test for multiple breaks in the loading matrix. (We use the terms ``change" and ``break" throughout the paper interchangeably.) \cite{su2017time} established a method to estimate the latent factors and time-varying factor loadings simultaneously.
\par
Much less attention has been given to multi-subject analysis. In many applications, the objective is to estimate several unknowns corresponding to related models; this problem is often referred to as multi-task learning in the machine learning literature. When multivariate repeated measurements are collected from multiple subjects, there is substantial interest in learning the same model parameters across all subjects. For factor models, this has led to different methods of estimating loading matrices across subjects; see \cite{timmerman2003four,de2012clusterwise,song2014analyzing}.
Multi-subject learning is slightly better understood for multiple regression problems 
\citep{jalali2013dirty,ollier2014joint,gross2016data} and multiple vector autoregression models \citep{fisher2022penalized}, where one assumes the presence of common effects in the regression vectors across individuals.
\par
Even though there is substantial interest in multi-subject models in the applied sciences, the statistical literature on multi-subject factor models is scarce. We put forward the framework of discovering interconnections for two-subject factor models. A modification of our test can also address situations when the two factor models are driven by different numbers of factors.
Our method can also be applied in change-point settings, assuming that structural changes occur over time. The existing literature tests no change in the loadings against the alternative hypothesis that a non-negligible portion of the cross sections have a break in their loadings. Our test allows to discover what kind of change occurs and effectively reduces the number of alternatives.
\par
The rest of this paper is organized as follows. Section \ref{se:HTP} formally introduces the hypothesis testing problem for two-subject factor models and for changes in the loading matrix over time. Section \ref{sec:test_statistic} motivates and introduces the test statistic including its theoretical properties. Section \ref{sec:alternative} studies the behavior under the alternative.
Section \ref{sec:simulations_study} provides simulation results giving some insights into the numerical performance of our test. A first  application to real data can be found in Section \ref{sec:data_application} and we conclude with Section \ref{sec:conclusion}. 
The supplementary document \cite{Supplement} contains the appendices for this article. Appendix \ref{sec:data_application_2} provides a second data application. Some variations of our approach and a discussion are presented in Appendix \ref{se:extensionsapplications}.
Appendix \ref{app:alternative_derivations} provides a detailed discussion on the test's behavior under the alternative.
Appendices \ref{se:appB}--\ref{se:matrixnorminequalities} contain the proofs.
\par
\textit{Notation:} For the reader's convenience, notation used throughout the paper is collected here. Let $A \in \RR^{\Dim \times \Dim}$. For symmetric $A$, write $\lambda_{1}(A)\geq\ldots\geq\lambda_{\Dim}(A)$ for its eigenvalues. The maximum and minimum eigenvalues of  $A$ are then denoted by $\lambda_{\max}(A)$ and $\lambda_{\min}(A)$, respectively. A range of different norms is used, including the maximum, the spectral and the Frobenius norm, defined respectively as $\| A \|_{\max}=\max_{1 \leq i,j \leq \Dim} |A_{ij}|$, $\| A \|=\sqrt{\lambda_{\max}(A'A)}$ and $\| A \|_{F}^2 = \sum_{i,j=1}^d |A_{ij}|^2$ for a matrix $A =(A_{ij})_{i,j=1,\dots,\Dim} \in \RR^{\Dim \times \Dim}$.  
For the vectorized version of a matrix $A$, we write $\vecop(A)$. The vec operator transforms a
matrix into a vector by stacking its columns one underneath the other. Similarly, $\vechop(A)$ denotes the vector that is obtained from $\vecop(A)$ by eliminating all superdiagonal elements of $A$.

\section{Hypothesis testing problem} \label{se:HTP}
The hypothesis testing problem studied in this paper is motivated by two different subject areas.
On the one hand, we study multi-task learning which is concerned with discovering structures across multiple models; see Section \ref{sec:multi_factor}. On the other hand, we are interested in structural changes in the loadings of factor models over time; see Section \ref{se:applicationCPs}. Finally, we formulate the hypothesis testing problem in Section \ref{sec:unify}.

\subsection{Two-subject factor model} \label{sec:multi_factor}
Multi-subject learning studies similarities and differences in model structure across individuals (subjects). The concept is formalized here for two subjects. 

Let $X_{t}^{1}= (X_{1,t}^{1}, \dots, X_{d,t}^{1})'$, $t \in \ZZ$, and $X_{t}^{2}= (X_{1,t}^{2}, \dots, X_{d,t}^{2})'$, $t \in \ZZ$, denote two $d$-dimensional time series. Here, the prime denotes transpose. Both series are assumed to follow a factor model representation such that
\begin{equation} \label{eq:two-subject_factor_model}
\begin{aligned}
X_{t}^{1} = \Lambda_{1} F_{t}^{1} + \varepsilon_{t}^{1}, \hspace{0.2cm} t =1, \dots, T_{1}, \\
X_{t}^{2} = \Lambda_{2} F_{t}^{2} + \varepsilon_{t}^{2}, \hspace{0.2cm} t =1, \dots, T_{2}, 
\end{aligned}
\end{equation}
where $\Lambda_{1},\Lambda_{2}$ are $d \times r$ loading matrices with $r<d$, $F^{k}_t = (F^{k}_{1,t}, \dots, F^{k}_{r,t})'$, $t \in \ZZ$, $k=1,2$, are $r$-vector time series of latent factors and $\varepsilon^{k}_t = (\varepsilon^{k}_{1,t}, \dots, \varepsilon^{k}_{d,t})'$, $t \in \ZZ$, $k=1,2$, are idiosyncratic errors which are assumed to be independent of the factors $F^{k}_{t}$. We will refer to \eqref{eq:two-subject_factor_model} as a two-subject factor model.
A typical assumption on the factors is to suppose that they follow a stable VAR$(p)$ model 
\begin{equation*}\label{e:dfm-basics-var} 
	F_t  = \Psi_1 F_{t-1} + \ldots + \Psi_p F_{t-p} + e_t,\quad t\in\ZZ,
\end{equation*}
where, for notational simplicity, we suppress dependence on $k$. 

There has been a lot of research on how to estimate the loading matrices $\Lambda_1$ and $\Lambda_2$ jointly by imposing some common structure. A typical objective function is
\begin{equation} \label{eq:psy_problem}
\sum_{k=1}^2 \sum_{t=1}^{T_k} \| X_{t}^{k} - \Lambda F_{t}^{k} \|^2_{F} 
\hspace{0.2cm}
\text{ subject to }
\hspace{0.2cm}
\frac{1}{T_k} \sum_{t=1}^{T_k} F^k_t F^{k'}_{t} = D_{k} \Sigma D_{k}
\hspace{0.2cm}
\text{ for }
\hspace{0.2cm}
k =1,2,
\end{equation}
and diagonal matrices $D_k$, $k=1,2$, with positive diagonal entries, which may be different for each subject and some nonsingular matrix $\Sigma$ which is assumed to be shared across subjects; see \cite{timmerman2003four,de2012clusterwise}.
Note that the constraint in \eqref{eq:psy_problem} is equivalent to having 
\begin{equation} \label{eq:psy_problem2}
\Lambda_k = \Lambda D_{k} 
\hspace{0.2cm}
\text{ and } 
\hspace{0.2cm}
\frac{1}{T_k} \sum_{t=1}^{T_k} F^k_t F^{k'}_{t} = \Sigma
\hspace{0.2cm}
\text{ for }
\hspace{0.2cm}
k =1,2.
\end{equation}
While the problem \eqref{eq:psy_problem} has drawn some interest in terms of optimization, there is no literature verifying that the constraint of having some common structure is a reasonable assumption. Our goal here is to develop a test for whether the data admit this kind of structure in the first place.
The hypothesis we are interested in is whether the two loading matrices $\Lambda_{1}, \Lambda_{2}$ in \eqref{eq:two-subject_factor_model} share common structure similar to \eqref{eq:psy_problem2}. More specifically, can one of the matrices be represented as a linear combination of the other? 

\subsection{Structural changes over time} \label{se:applicationCPs}
The second motivation for our testing problem comes from the literature on testing for structural changes in loading matrices of a single factor model over time. Since \cite{han2015tests} proposed a classical change-point test to test the null hypothesis of no change against the alternative of structural changes in the factor loadings by applying a CUSUM type statistic, the related literature has grown significantly.

Consider the following factor model that allows for a structural break in the factor loadings:
\begin{equation} \label{eq:cp_model}
X_{t} = 
\begin{cases}
\Lambda_{1} F_{t} + \varepsilon_{t} \hspace{0.2cm} 1 \leq t \leq \lfloor \pi T \rfloor, \\
\Lambda_{2} F_{t} + \varepsilon_{t} \hspace{0.2cm} \lfloor \pi T \rfloor + 1 \leq t \leq T,
\end{cases}
\end{equation}
where, for $\pi \in (0,1)$, $\lfloor \pi T \rfloor$ is a possibly unknown break date. The change-point literature typically distinguishes between three alternatives, expressed in terms of column spaces. Recall that the column space of a matrix $A \in \RR^{d\times r}$ is the set of all linear combinations of the columns in $A$.  More formally, for $A = (a_1, \dots, a_r)$, $\col(A) = \left\{ \sum_{i=1}^{r} \mu_i a_i ~|~ \mu_i \in \RR \right\}$. The three alternatives are as follows:
\begin{enumerate}[label=\textit{Type \arabic*:}, ref=\arabic*, align=left]
\item The columns of $\Lambda_{1}$ and $\Lambda_{2}$ are linearly independent, that is, $\col(\Lambda_1) \cap \col(\Lambda_2) = \{0\}$. Note that this requires $2r \leq d$.
\label{item:Type1}
\item Loadings undergo an invertible linear transformation, i.e. the columns of $\Lambda_{1}$ and $\Lambda_{2}$ are linearly dependent or, equivalently, $\col(\Lambda_1) = \col(\Lambda_2)$.
\label{item:Type2}
\item Some (but not all) columns of $\Lambda_{1}$ and $\Lambda_{2}$ are linearly independent, that is,  $\col(\Lambda_1) \cap \col(\Lambda_2) \neq \{0\}$ and  $\col(\Lambda_1) \neq \col(\Lambda_2)$.
\label{item:Type3}
\end{enumerate}
An invertible linear transformation as described by Type \ref{item:Type2} can be written as $\Lambda_2 = \Lambda_1 \Phi$ for some nonsingular matrix $\Phi$ (this assumes that $\Lambda_1$ and $\Lambda_2$ have full column rank). Suppose further that $\E F_t F'_t = \Sigma_{F} $ for all $t=1,\dots,T$.
Then, a linear transformation by the nonsingular matrix $\Phi$ is observationally equivalent to no change in the loadings and a change in the second order structure of the factors:
\begin{equation} \label{eq:type3}
\E F_t F'_t = 
\begin{cases}
\Sigma_{F} 			&\hspace{0.2cm} 1 \leq t \leq \lfloor \pi T \rfloor, \\
\Phi \Sigma_{F} \Phi'  	&\hspace{0.2cm} \lfloor \pi T \rfloor + 1 \leq t \leq T.
\end{cases}
\end{equation}

Note that the different types of breaks are not fully characterized by Types \ref{item:Type1}--\ref{item:Type3}. For instance, given $\Sigma_F = I_r$ and $\Phi$ being an orthogonal matrix, we have  $\Sigma_{F} = \Phi \Sigma_{F} \Phi' $. Therefore, this is not considered a change. Whenever a Type \ref{item:Type2} change is meant, we need to assume $\Sigma_{F} \neq \Phi \Sigma_{F} \Phi' $. This is formalized in Assumption 10 in \cite{han2015tests}. 

While the change-point literature tends to formalize these three alternatives, there is currently no way of detecting them separately. We propose a test for changes up to an invertible linear transformation in the loadings.
To be more precise, suppose a regular change-point test as that of \cite{han2015tests} rejects the hypothesis of no change and one estimates a break point $\pi^{*}$. Then, the occurring change can follow any of the three types above.  
We aim to test for the hypothesis of a change of Type \ref{item:Type2} against the alternative of any other kind of changes. 
In other words, can one test whether the change in the loadings follows a Type \ref{item:Type2} break?

\subsection{Hypothesis testing problem} \label{sec:unify}
The formulation of a Type \ref{item:Type2} break in \eqref{eq:type3} resembles the constraint of interest in the multi-task learning literature as formalized in \eqref{eq:psy_problem}. While multi-task learning aims to find structural similarities, the change-point problem seeks for structural differences. Both problems reduce to the same hypothesis. 

For the remainder of the paper, we focus on the two-subject model \eqref{eq:two-subject_factor_model} but have the change-point problem \eqref{eq:cp_model} in mind as another application. We consider
\begin{equation} \label{eq:model_unified}
\begin{aligned}
X_{t}^{1} = \Lambda_{1} F^{1}_{t} + \varepsilon_{t}^{1}, \hspace{0.2cm} t =1, \dots, T_{1}, \\
X_{t}^{2} = \Lambda_{2} F^{2}_{t} + \varepsilon_{t}^{2}, \hspace{0.2cm} t =1, \dots, T_{2}.
\end{aligned}
\end{equation}
Then, the null hypothesis can be formalized as
\begin{equation} \label{eq:H01}
H_{0} : \Lambda_{k}  = \Lambda \Phi_{k} \text{ for } k =1, 2.
\end{equation}
Here, the matrix $\Lambda \in \RR^{d \times r}$ is a matrix of full column rank and the matrices $\Phi_{k} \in \RR^{r\times r}$, $k=1,2$, are nonsingular. Set $\Lambda'_{k} = (\lambda_{k,1}, \dots, \lambda_{k,d})$ with $\lambda_{k,i} \in \RR^{r}$ for $k=1,2$.
We want to test against the alternative hypothesis that there at least one cross section that has linearly independent loadings, i.e.,
\begin{equation*} \label{eq:H1}
H_{1} : \lambda_{1,i}'  \neq \lambda_{2,i}' \Phi \hspace{0.2cm} \text{ for some } i =1, \dots, d,
\end{equation*}
and all nonsingular matrices $\Phi \in \RR^{r\times r}$. 

Before moving on to proposing a test for the stated hypothesis problem, we pause here to get some intuition for what \eqref{eq:H01} means for data which might follow this model. Consider the following simple example.
\begin{example} \label{example:1}
Let $\Phi_{1} = I_{r}$ and $\Phi_{2} = \diag(\phi_{1} , \dots, \phi_{r} )$ such that
\begin{equation} \label{eq:simple_example}
\Lambda_{1}  = \Lambda 
\hspace{0.2cm}
\text{ and }
\hspace{0.2cm}
\Lambda_{2}  = \Lambda \diag(\phi_{1} , \dots, \phi_{r} ).
\end{equation}
In addition, we assume that $\E F^k_{t} F^{k'}_{t} = I_{r}$, $k=1,2$.
Recall that \eqref{eq:simple_example} is observationally equivalent to differences in the second-order structure of the factors, i.e. 
\begin{equation*}
\E F^1_{t} F^{1'}_{t} = I_{r}
\hspace{0.2cm}
\text{ and }
\hspace{0.2cm}
\E F^2_{t} F^{2'}_{t} = \diag(\phi^2_{1} , \dots, \phi^2_{r} ).
\end{equation*}
This setting means that only the volatility of the individual factors changes over time. We suspect that this type of change explains why our test does not reject the null in our data application; see Section \ref{sec:data_application} for more details.
\end{example}

\section{Testing procedure} \label{sec:test_statistic}
In this section, we detail our method to test the null hypothesis \eqref{eq:H01}. In particular, we give a motivation for our test statistic before formulating it.

We first recall estimation of loadings and factors in individual factor models based on principal component analysis (PCA).
In order to estimate the loadings and factors in the two-subject factor model \eqref{eq:model_unified}, we concatenate the series across time and apply PCA following \cite{bai2003inferential}. For $k=1,2$, set
\begin{equation} \label{eq:concatenated_each_series}
\begin{pmatrix}
X_{1}^{k'} \\
\vdots \\
X_{T_k}^{k'} 
\end{pmatrix}
=
\begin{pmatrix}
F_{1}^{k'}  \\
\vdots \\
F_{T_k}^{k'} 
\end{pmatrix}
\Lambda'_{k}
+
\begin{pmatrix}
\varepsilon_{1}^{k'}  \\
\vdots \\
\varepsilon_{T_k}^{k'} 
\end{pmatrix}
\end{equation}
and for short $X^{k'} = F^{k} \Lambda'_{k} + \varepsilon^{k}$. Write $\widehat{\Sigma}^{k}_{T_{k}} = X^{k'} X^{k}$ and let $\widehat{Q}^{k}_{r}$ be the $r$ eigenvectors corresponding to the $r$ largest eigenvalues of $\widehat{\Sigma}^{k}_{T_{k}}$. Then, the PCA estimators are
\begin{equation} \label{eq:PCAestimates_individual_series}
\widehat{\Lambda}'_{k} = \frac{1}{T_k} \widehat{F}^{k'} X^{k'}, \hspace{0.2cm}
\widehat{F}^{k} = \sqrt{T_k} \widehat{Q}_{r}^{k}.
\end{equation}
These estimators are consistent under appropriate assumptions outlined in Section \ref{se:assumptions} below; see \cite{bai2003inferential}.

\subsection{Motivation} \label{se:motivation}
The null \eqref{eq:H01} states that the columns of $\Lambda_{k}$ lie in the subspace generated by the columns of $\Lambda$. Therefore, the null hypothesis can be reformulated as
\begin{equation*} \label{eq:H02}
H_{0} : P_{0} \Lambda_{k} = \Lambda_{k}, \hspace{0.2cm} k =1, 2,
\end{equation*}
where $P_{0}$ is the projection matrix onto the subspace generated by the columns of $\Lambda$. Since $\Lambda$ has linearly independent columns, the projection matrix is given by $P_{0} = \Lambda ( \Lambda' \Lambda )^{-1} \Lambda'$.

This observation inspires the following test. Define the projection matrices of the individual loadings as
\begin{equation} \label{eq:true_projections_individual_series}
P_{k} = \Lambda_{k} ( \Lambda'_{k} \Lambda_{k} )^{-1} \Lambda'_{k}, \hspace{0.2cm} k =1,2.
\end{equation}
The projection matrices $P_{k}$ in \eqref{eq:true_projections_individual_series} can be estimated through the PCA estimator \eqref{eq:PCAestimates_individual_series} of the loading matrices as
\begin{equation} \label{eq:def_Phat12}
\widehat{P}_{k} = \widehat{\Lambda}_{k} ( \widehat{\Lambda}'_{k} \widehat{\Lambda}_{k})^{-1} \widehat{\Lambda}'_{k}, \hspace{0.2cm} k =1, 2.
\end{equation}

Note that projection matrices are invariant under any transformation with a nonsingular matrix $\Phi \in \RR^{r \times r}$, i.e.
\begin{equation} \label{eq:projection_invariance}
\Lambda_{k}\Phi ( \Phi' \Lambda'_{k} \Lambda_{k} \Phi )^{-1} \Phi' \Lambda'_{k} = \Lambda_{k} ( \Lambda'_{k} \Lambda_{k} )^{-1} \Lambda'_{k}, \hspace{0.2cm} k =1,2.
\end{equation}
In particular, under the hypothesis \eqref{eq:H01}, 
\begin{equation} \label{eq:projection_invariance_for_each_P}
P_{k} 
= \Lambda_{k} ( \Lambda'_{k} \Lambda_{k} )^{-1} \Lambda'_{k} 
= \Lambda \Phi_{k} ( \Phi'_{k} \Lambda' \Lambda \Phi_{k} )^{-1} \Phi'_{k} \Lambda'
= \Lambda ( \Lambda' \Lambda )^{-1} \Lambda'
, \hspace{0.2cm} k =1,2.
\end{equation}
Introduce the notation
\begin{equation*}
X^{k}_{1:T_{k}} = (X_{1}^{k}, \dots, X_{T_{k}}^{k} ) \in \RR^{d \times T_{k} }, \hspace{0.2cm} k =1,2.
\end{equation*}
The general idea is to split the series $X^{k}_{1:T_{k}}$ into two halves and to apply an appropriate transformation such that there is no change in the loadings under the null, but there is a change under the alternative. For notational simplicity we suppose throughout this section that $T_1,T_2$ are even. Otherwise, one simply considers $\lfloor T_1 /2 \rfloor, \lfloor T_2/2 \rfloor$. We propose the following transformation
\begin{equation} \label{eq:YT2def}
Y_{1:T_{2}} = \mleft( \widehat{P}_{1} X^{2}_{1:T_{2}/2}, \widehat{P}_{2} X^{2}_{ (T_{2}/2 +1) :T_{2}} \mright),
\end{equation}
where $\widehat{P}_{1}, \widehat{P}_{2}$ are as in \eqref{eq:def_Phat12}.
Note that the roles of the first and the second series are interchangeable. 
To motivate our choice of transforming the series $X^{2}_{1:T_{2}}$ by $\widehat{P}_{1}$ and $\widehat{P}_{2}$, we consider the transformation at the population level. Then, by \eqref{eq:projection_invariance_for_each_P}, 
\begin{equation*}
\begin{aligned}
P_{1} X_{t}^2 
&\approx P_{1} \Lambda_{2} F_{t}^2 = P_{1} \Lambda \Phi_{2} F_{t}^2 
= \Lambda ( \Lambda' \Lambda )^{-1} \Lambda' \Lambda \Phi_{2} F_{t}^2 = \Lambda \Phi_{2} F_{t}^2 = \Lambda_{2} F_{t}^2, 
\hspace{0.2cm} \text{ for } t =1,\dots, T_{2}/2, \\
P_{2} X_{t}^2 
&\approx P_{2} \Lambda_{2} F_{t}^2 = P_{2} \Lambda \Phi_{2} F_{t}^2 
= \Lambda ( \Lambda' \Lambda )^{-1} \Lambda' \Lambda \Phi_{2} F_{t}^2 = \Lambda_{2} F_{t}^2,
\hspace{0.2cm} \text{ for } t =T_{2}/2+1,\dots, T_{2}
\end{aligned}
\end{equation*}
and therefore
\begin{equation} \label{eq:motivation_null_population}
\begin{aligned}
P_{1} X_{t}^2 
&\approx  \Lambda_{2} F_{t}^2, 
\hspace{0.2cm} \text{ for } t =1,\dots, T_{2}/2, \\
P_{2} X_{t}^2 
&\approx \Lambda_{2} F_{t}^2,
\hspace{0.2cm} \text{ for } t =T_{2}/2+1,\dots, T_{2}.
\end{aligned}
\end{equation}
Since $P_{1} = P_{2}$ under the null hypothesis, the idiosyncratic errors of the transformed series $Y_{1:T_2}$ are still stationary. 
\begin{equation} \label{eq:motivation_null_population_errors}
\begin{aligned}
P_{1} \E (\varepsilon_{t}^2 \varepsilon_{t}^{2'}) P_{1}' = P_{0} \Sigma_{\varepsilon} P_{0}', \\
P_{2} \E (\varepsilon_{t}^2 \varepsilon_{t}^{2'}) P_{2}' = P_{0} \Sigma_{\varepsilon} P_{0}'.
\end{aligned}
\end{equation}

Due to \eqref{eq:motivation_null_population}, the series $Y_{1:T_{2}} $ in \eqref{eq:YT2def} is expected to have the same loadings and factors under the null as $X^{2}_{1:T_{2}} $. However, given \eqref{eq:motivation_null_population_errors}, the errors have a different covariance matrix than the original series.

In contrast, when the null hypothesis is not satisfied, we get
\begin{equation} \label{eq:motivation_alternative_population}
\begin{aligned}
P_{1} X_{t}^2 
&\approx P_{1} \Lambda_{2} F_{t}^2
= \Lambda_{1} ( \Lambda'_{1} \Lambda_{1} )^{-1} \Lambda'_{1} \Lambda_{2} F_{t}^2 \neq \Lambda_{2} F_{t}^2, 
\hspace{0.2cm} \text{ for } t =1,\dots, T_{2}/2, \\
P_{2} X_{t}^2 
&\approx P_{2} \Lambda_{2} F_{t}^2
= \Lambda_{2} ( \Lambda'_{2} \Lambda_{2} )^{-1} \Lambda'_{2} \Lambda_{2} F_{t}^2 = \Lambda_{2} F_{t}^2,
\hspace{0.2cm} \text{ for } t =T_{2}/2+1,\dots, T_{2}.
\end{aligned}
\end{equation}
The illustrated behavior of $Y_{1:T_{2}}$ at the population level in \eqref{eq:motivation_null_population} and \eqref{eq:motivation_alternative_population} suggests that for the transformed series $Y_{1:T_{2}}$ our hypothesis testing problem is equivalent to testing for a change at time $T_{2}/2$ in the loadings. The roles of $X^{1}_{1:T_{1}}$ and $X^{2}_{1:T_{2}}$ are interchangeable. Depending on which series we base our test on, we can consider either 
\begin{equation} \label{eq:transfomedseriesnotationchange1}
\begin{aligned}
Y_{t}=
\begin{cases}
\widehat{P}_{1} X_{t}^2, \hspace{0.2cm} \text{ for } t =1,\dots, T_{2}/2,
\\
\widehat{P}_{2} X_{t}^2, \hspace{0.2cm} \text{ for } t =T_{2}/2+1,\dots, T_{2}
\end{cases}
\end{aligned}
\end{equation}
or
\begin{equation}  \label{eq:transfomedseriesnotationchange2}
\begin{aligned}
Y_{t}=
\begin{cases}
\widehat{P}_{1} X_{t}^1, \hspace{0.2cm} \text{ for } t =1,\dots, T_{1}/2,
\\
\widehat{P}_{2} X_{t}^1, \hspace{0.2cm} \text{ for } t =T_{1}/2+1,\dots, T_{1}.
\end{cases}
\end{aligned}
\end{equation}
Under the hypothesis, i.e., when there is no change in the loadings, the true loadings and factors of the transformed series  \eqref{eq:transfomedseriesnotationchange1} and \eqref{eq:transfomedseriesnotationchange2} should respectively coincide with the true loadings and factors of the series $X_{t}^2$ and $X_{t}^1$.

To avoid complicated notation, we base our test on the second series $X_{t}^2$ and write from now on 
\begin{equation} \label{eq:F=F^2}
F_t 
\hspace{0.2cm}
\mbox{instead of} 
\hspace{0.2cm}
F_t^2. 
\end{equation}
We keep $\Lambda_2$ to denote the corresponding loading matrix and to avoid any confusion with $\Lambda$ as used to characterize the null hypothesis in \eqref{eq:H01}.
We also assume from now on that 
$T=T_1 = T_2$.
Several quantities are defined above before and after the ``change-point" $T/2$. We shall sometimes distinguish between them below by using superscript $b$ for ``before" and $a$ for ``after". For example, in view of \eqref{eq:F=F^2}, we shall write 
\begin{equation} \label{eq:FaFb}
F^b = (F_{1}, \dots, F_{T/2})', \hspace{0.2cm}
F^a = (F_{T/2+1}, \dots, F_{T})'.
\end{equation}

In order to test for changes in the loadings of the transformed series $\{Y_t\}_{t=1,\dots,T}$ we can now employ a change-point test as introduced in \cite{han2015tests}. This is outlined in more detail in the next section. See also \cite{BAEK2021107067} for an application of the test in \cite{han2015tests} to testing for changes in the autocovariances of the latent factors.

\subsection{The test statistic} \label{sec3.2}
Recall from \eqref{eq:YT2def} the transformed series $Y_{1:T} = \mleft( \widehat{P}_{1} X^{2}_{1:T/2}, \widehat{P}_{2} X^{2}_{ (T/2 +1) :T} \mright)$.
To employ the change-point test of \cite{han2015tests}, we first concatenate \eqref{eq:transfomedseriesnotationchange1} across time to get for short 
\begin{equation} \label{eq:defofY}
Y := Y'_{1:T} = 
\begin{pmatrix}
F_{1:T/2}' \Lambda'_2 \widehat{P}_{1} \\
F_{(T/2+1):T}' \Lambda'_2 \widehat{P}_{2}
\end{pmatrix}
+
\begin{pmatrix}
\varepsilon'_{1:T/2} \widehat{P}_{1} \\
\varepsilon'_{(T/2+1):T} \widehat{P}_{2} 
\end{pmatrix}
=:
\begin{pmatrix}
F^b \Lambda'_2 \widehat{P}_{1} \\
F^a \Lambda'_2 \widehat{P}_{2}
\end{pmatrix}
+
\begin{pmatrix}
\varepsilon^b \widehat{P}_{1} \\
\varepsilon^a \widehat{P}_{2} 
\end{pmatrix}
\end{equation}
with $F' := F_{1:T} = (F_{1}, \dots, F_{T})$, $F_{t} \in \RR^r$. The use of superscripts $a$ and $b$ was discussed before in \eqref{eq:FaFb}.
We follow the same procedure as in \eqref{eq:concatenated_each_series} but now for the transformed series \eqref{eq:defofY}.
The PCA estimators are then based on $\widehat{\Sigma}_{T} = Y Y'$ such that
\begin{equation} \label{eq:PCAestimates_individual_G}
\widehat{F} = \sqrt{T} \widehat{Q}_{r}, \hspace{0.2cm}
\widehat{\Lambda}' = \frac{1}{T} \widehat{F}' Y,
\end{equation}
where $\widehat{Q}_{r}$ are the $r$ eigenvectors corresponding to the $r$ largest eigenvalues of $\widehat{\Sigma}_{T}$.

We base the proposed test statistic on the pre- and post-sample means of $\widehat{F}_{t} \widehat{F}_{t}'$, where $\widehat{F}_{t}$ is the PCA estimator of the factors of our transformed series as stated in \eqref{eq:PCAestimates_individual_G}. Define
\begin{equation} \label{eq:def:V(G)}
V(\widehat{F}) = \vechop \mleft( \frac{1}{\sqrt{T}} \sum_{t=1}^{T/2} \widehat{F}_{t} \widehat{F}_{t}' - \frac{1}{\sqrt{T}} \sum_{t=T/2 + 1}^{T} \widehat{F}_{t} \widehat{F}_{t}' \mright).
\end{equation}
To normalize appropriately, we use the following long-run variance estimator as suggested in Section 2.4 of \cite{han2015tests},
\begin{equation} \label{eq:def:Omega(G)}
\Omega(\widehat{F}) = \widehat{\Gamma}_{0}(\widehat{F}) + \sum_{j=1}^{T-1} \kappa \mleft( \frac{j}{b_{T}} \mright) \Big( \widehat{\Gamma}_{j}(\widehat{F}) + \widehat{\Gamma}_{j}'(\widehat{F}) \Big),
\end{equation}
where $\kappa$ is a kernel function, $b_{T}$ the corresponding bandwidth and 
\begin{equation} \label{eq:def:Gamma(G)}
\widehat{\Gamma}_{j}(\widehat{F}) = \frac{1}{T} \sum_{t=1+j}^{T} \vechop (\widehat{F}_{t} \widehat{F}_{t}' - I_{r}) \vechop (\widehat{F}_{t-j} \widehat{F}_{t-j}' - I_{r})'.
\end{equation}
Then, the Wald type test statistic can be defined as
\begin{equation} \label{eq:def:Wald}
W(\widehat{F}) = V(\widehat{F})' \Omega^{-1}(\widehat{F}) V(\widehat{F}).
\end{equation}
Our main theoretical contribution is to prove that the PCA estimators of the transformed series $Y$ are still consistent estimators for the loadings and factors. 

We show in Appendix \ref{se:appB} that $\widehat{F}$ behaves as $FH$ or $FH_{0}$. Here, 
\begin{equation} \label{eq:HconvH0}
H := (\Lambda'_2 \Lambda_2/d) (F' \widehat{F} /T) \widehat{V}^{-1}_{r} \overset{\prob}{\to} H_{0}, \text{ as } d,T\to\infty,
\end{equation}
where $\widehat{V}_{r}$ is an $r \times r$ diagonal matrix with the $r$ largest eigenvalues of $\frac{1}{Td} \widehat{\Sigma}_{T}$ on the diagonal. The convergence \eqref{eq:HconvH0} is shown in Appendix \ref{se:appB}.
We also use $\Omega(FH_0)$ which is $\Omega(\cdot)$ as in \eqref{eq:def:Omega(G)} but evaluated at $FH_0$. Both quantities are shown to be consistent estimators of the true long-run variance 
\begin{equation} \label{eq:Omega}
\Omega = \lim_{ T \to \infty } \Var \mleft( \vechop \mleft( \frac{1}{\sqrt{T}} \mleft( \sum_{t=1}^{T} H'_{0} F_{t} F_{t}' H_{0} - I_{r} \mright) \mright) \mright).
\end{equation}

\begin{remark} \label{re:othertests}
We base our test on a CUSUM like statistic as considered in \cite{han2015tests}.
However, it is certainly conceivable to use our transformation in combination with other change-point tests. For example, \cite{bai2022likelihood} considered a likelihood ratio based test. \md{See also Section \ref{se:roadmap} for a general strategy on how to infer the asymptotic behavior for a test statistic based on a transformed series.}
\end{remark}

\subsection{Assumptions} \label{se:assumptions}

To study the limiting distribution of our Wald type test statistic, we state some assumptions. Those assumptions are standard in the factor model literature and ensure that the PCA estimators used in our test statistic are consistent. 

In the factor model literature, there are different approaches of proving consistency of the PCA estimators. Those different approaches involve different assumptions. We follow here the approach in \cite{doz2011two}. The assumptions therein are slightly different than those introduced by \cite{stock2002forecasting,bai2002determining} and \cite{bai2003inferential} but have a similar role. 

Note that all assumptions need to be satisfied for both series in \eqref{eq:model_unified}. To avoid unnecessarily complicated notation we state the assumptions for a universal factor model
\begin{equation*}
X_{t} = \Upsilon F_{t} + \varepsilon_{t}.
\end{equation*}

\begin{Massump} \label{ass:1}
Suppose $\{ X_t\}_{t=1,\dots,T}$ is stationary with $\E X_t =0$ and $\Var(X_{j,t}) \leq M$ for all $j=1,\dots, d$ and $t=1,\dots,T$, where the constant $M$ does not depend on $d$.
\end{Massump}

\begin{Massump} \label{ass:2}
$\{F_t\}$ and $\{\varepsilon_t\}$ are independent and admit Wold representations:
\begin{enumerate}[label=(\roman*)]
\item
$F_t = \sum_{k=0}^{\infty}C_k a_{t-k}$, with $\sum_{k=0}^{\infty}\|C_k\|<\infty$ and a white noise series $\{a_t\}$ with finite fourth moments, and $\{F_t\}$ has positive definite covariance matrix $\E F_t F_t' = \Sigma_F$.
\item
$\varepsilon_t = \sum_{k=0}^{\infty} D_k b_{t-k}$, $\sum_{k=0}^{\infty}\|D_k\|<\infty$, and $\{b_t = (b_{1,t},\dots, b_{d,t})' \}$ is a white noise series such that $ \E b_{j,t}^4 \leq M$ for all $j=1,\dots,\Dim$ and $t=1,\dots,T$.
\end{enumerate}
\end{Massump}
Assumption \ref{ass:2}(i) implies that $ \frac{1}{T} \sum_{t=1}^{T} F_t F_t' \overset{\prob}{\to} \Sigma_{F}$ as $T \to \infty$; see Theorem 2 in \cite{Hannan1970:Multiple}, p. 203. The convergence aligns with Assumption A in \cite{bai2003inferential}. 
A convenient way to parameterize the dynamics would be to further assume that the common factors follow a vector autoregression model. 

The following Assumption \ref{ass:3} imposes more restrictions on the correlation structure of idiosyncratic errors of the factor models. In particular, it collects the assumptions which go beyond the assumptions necessary to prove consistent estimation of the factors. Instead, they are needed to infer convergence of the long-run variance matrix in Proposition \ref{prop3}. These assumptions could also be inferred from more restrictive assumptions on the matrices $D_k$ in Assumption \ref{ass:2}.
\begin{Massump} \label{ass:3}
 There exists a positive constant $M$ such that for all $d,T$, 
\begin{enumerate}[label=(\roman*)]
\item
$\frac{1}{d} \sum_{i,j=1}^{d} | \E(\varepsilon_{i,t_1} \varepsilon_{j,t_2}) | \leq M$ for all $t_{1},t_{2} =1,\dots, T$;
\item
$\frac{1}{d^2} \E \| e'_t \varepsilon \Upsilon \|^4 \leq M$ for each $t = 1,\dots, T$ with $e_{t}$ denoting the $t$th unit vector in $\RR^{T}$ and $\varepsilon = (\varepsilon_1, \dots, \varepsilon_T)'$;
\item
$\E ( \varepsilon^4_{i,t} ) \leq M$ for all $i=1,\dots,d$ and $t = 1,\dots, T$.
\end{enumerate}
\end{Massump}

For the following assumptions, let $\Pi_{\Upsilon}$ denote the diagonal matrix whose diagonal entries are the eigenvalues of $\Upsilon' \Upsilon$ in decreasing order.

\begin{CRassump} \label{ass:C1}
There is a positive definite diagonal matrix $\widetilde{\Pi}$ such that $\| \Pi_{\Upsilon}/d - \widetilde{\Pi} \| \to 0$ with $\| \Pi_{\Upsilon}/d - \widetilde{\Pi} \| = \mathcal{O}\mleft(\frac{1}{\sqrt{d}}\mright)$. In particular, one can infer that
\begin{enumerate}[label=(\roman*)]
\item $\liminf_{d\to\infty} \frac{\lambda_{\min}(\Upsilon'\Upsilon) }{d} > 0$;
\item $\limsup_{d\to\infty}  \frac{\lambda_{\max}(\Upsilon'\Upsilon)}{d}   = \limsup_{d\to\infty} \frac{\|\Upsilon\|^2}{d} <\infty$.
\end{enumerate}
\end{CRassump}

\begin{CRassump} \label{ass:C2}
The eigenvalues of $\Pi_{\Upsilon}$ are distinct.
\end{CRassump}

\begin{CRassump} \label{ass:C4}
Suppose $\limsup_{d\to\infty} \sum_{h \in \ZZ}\|\Gamma_{\varepsilon}(h) \|<\infty$ (with $\Gamma_{Z}(h)=\E Z_{t+h}Z_t'$ denoting the autocovariance function of a stationary mean-zero series $\{Z_t\}$).
\end{CRassump}

\begin{CRassump} \label{ass:C3}
Suppose there is a $\widebar{\lambda}$ such that $\| \Upsilon \|_{\max} \leq \widebar{\lambda} < \infty$.
\end{CRassump}

%

\subsection{Limiting distribution under the null} \label{sec:Mainresultshypothesis}
In this section, we present our theoretical results, in particular, the asymptotic behavior of our test statistic \eqref{eq:def:Wald} under the null hypothesis. For that we need a few more assumptions.
 
\begin{Wassump} \label{ass:9} (Asymptotics)
\begin{enumerate}[label=(\roman*)]
\item
Suppose $\Omega$ in \eqref{eq:Omega} is positive definite and
\begin{equation*}
\| \Omega(FH_{0}) - \Omega \| = o_{\prob}(1).
\end{equation*}
\item For $W(\cdot)$ in \eqref{eq:def:Wald}, suppose
\begin{equation*}
W(FH_{0}) \overset{\distr}{\to} \chi^{2}(r(r+1)/2),
\end{equation*}
where $\chi^{2}(r(r+1)/2)$ denotes the chi-squared distribution with $r(r+1)/2$ degrees of freedom.
\end{enumerate}
\end{Wassump}

Theorem 3 in \cite{andrews1993tests} provides conditions under which Assumption \ref{ass:9}(ii) is satisfied.

\begin{Wassump} \label{ass:10}
Suppose $\kappa$ in \eqref{eq:def:Omega(G)} is the Bartlett kernel $(\kappa(x) = (1-|x|)\mathds{1}_{\{|x|\leq 1\}})$ and the bandwidth $b_{T}$ satisfies $b_{T} = O\mleft( T^{1/3} \mright)$ and $\frac{T^{\frac{2}{3}}}{d} \to 0$, as $d,T \to \infty$.
\end{Wassump}
The following proposition states the asymptotic behavior of our test statistic. 
\begin{proposition} \label{prop1}
Recall the definition of $W(\cdot)$ in \eqref{eq:def:Wald} and suppose Assumptions \ref{ass:1}--\ref{ass:3}, \ref{ass:C1}--\ref{ass:C3}, \ref{ass:9},\ref{ass:10}, and $\frac{\sqrt{T}}{d} \to 0$, as $d,T \to \infty$. Then, under the null hypothesis \eqref{eq:H01},
\begin{equation*}
W(\widehat{F}) \overset{\distr}{\to} \chi^{2}(r(r+1)/2),
\end{equation*}
where $\chi^{2}(r(r+1)/2)$ denotes the chi-squared distribution with $r(r+1)/2$ degrees of freedom.
\end{proposition}

Note that Proposition \ref{prop1} requires $\frac{\sqrt{T}}{d} \to 0$, as $d,T \to \infty$. This assumption on the relationship between dimension and sample size is typical in the setting of testing for structural changes in high-dimensional factor models. In particular, it is imposed in \cite{han2015tests}. 
The assumption is also covered by the assumptions studied in random matrix theory. Random matrix theory typically distinguishes between $\frac{T}{d} \to c>0$ and $\frac{T}{d} \to 0$ (also known as the ultra-high dimensional regime), both imply our assumption. 

The proof of Proposition \ref{prop1} requires convergence of $V(\cdot)$ in  \eqref{eq:def:V(G)} and of the long-run variance $\widehat{\Omega}(\cdot)$ in \eqref{eq:def:Omega(G)}. The two propositions below formalize that.

 \begin{proposition} \label{prop2}
Recall the definition of $V(\cdot)$ in  \eqref{eq:def:V(G)} and suppose Assumptions \ref{ass:1}--\ref{ass:3}(i), \ref{ass:C1}--\ref{ass:C3} and $\frac{\sqrt{T}}{d} \to 0$, as $d,T \to \infty$. Then, under the null hypothesis \eqref{eq:H01},
 \begin{equation*}
\| V(\widehat{F}) - V(FH_{0}) \| \overset{\prob}{\to} 0.
 \end{equation*}
 \end{proposition}

\begin{proposition} \label{prop3}
Recall the definition of $\Omega(\cdot)$ in \eqref{eq:def:Omega(G)} and suppose Assumptions \ref{ass:1}--\ref{ass:3}, \ref{ass:C1}--\ref{ass:C3}, \ref{ass:10} and $\frac{\sqrt{T}}{d} \to 0$, as $d,T \to \infty$. Then, under the null hypothesis \eqref{eq:H01},
\begin{equation*}
\| \Omega(\widehat{F}) - \Omega(FH_{0}) \| \overset{\prob}{\to} 0.
\end{equation*}
\end{proposition}

The proofs of the results in this section can be found in the appendix.

\subsection{Proof idea} \label{se:roadmap}

Recall from Section \ref{se:motivation} that our testing procedure transforms the data such that under the null hypothesis, there is no change in the loadings, while there is one under the alternative. The transformation was motivated by the goal of applying an available change-point test as the one introduced in \cite{han2015tests} to the transformed data. This is also reflected in the proofs which are adapted from those in \cite{han2015tests}. It turns out to be sufficient to prove that the PCA estimators of the transformed series are consistent and subsequently apply the results in \cite{han2015tests}.  

As noted in Remark \ref{re:othertests}, instead of applying a CUSUM based statistic, one could potentially use other change-point tests, for instance, a likelihood ratio based test. The general strategy of how to derive the asymptotic behavior of the test statistic remains the same, that is, one proves that the PCA estimators of the transformed series are consistent. Subsequently, one can infer the asymptotic distribution of the respective test statistic.
The last step depends on the test statistic and is a priori not immediate even if the proof turns out to be similar; in the case considered here, this is moved to the supplementary document \cite{Supplement}.

We outline the proof idea here and refer to Appendices \ref{se:appB}--\ref{se:matrixnorminequalities} for more details.
Set $\delta_{dT} = \min\{ \sqrt{d}, \sqrt{T}\}$ and recall the notation $F^b, F^a, \widehat{F},H$ and $F' = (F^{b'}, F^{a'})$ from Section \ref{sec3.2}. We further introduce $\widetilde{F} \widetilde{H} = (H_{1}' F^{b'},  H_{2}' F^{a'} )'$ for matrices $H_1,H_2$ specified in \eqref{def:H1H2} in Appendix \ref{se:appB}. 
Then, in order to prove Propositions \ref{prop1}--\ref{prop3}, it suffices to establish the following to be stated more formally in Appendix \ref{se:appB}: 
\begin{itemize}
\item Proposition \ref{prop2.1} states that
\begin{equation} \label{eq:key1}
\frac{1}{T} \| \widehat{F} - \widetilde{F} \widetilde{H} \|_{F}^2 = \mathcal{O}_{\prob}\mleft(\frac{1}{\delta_{dT}^2}\mright), \text{ as } d,T\to\infty,
\end{equation}
which is analogous to Lemma 1 in \cite{han2015tests}.
\item Proposition \ref{prop2.1.0} states that 
\begin{equation} \label{eq:key2}
\frac{1}{T} \| ( \widehat{F} - \widetilde{F} \widetilde{H})'F \|_{F} = \mathcal{O}_{\prob}\mleft(\frac{1}{\delta_{dT}^2}\mright), \text{ as } d,T\to\infty,
\end{equation}
which is analogous to Lemma 2 in \cite{han2015tests}.
\item Lemma \ref{le:convergenceH12toH} states 
\begin{equation} \label{eq:key3}
H_1 - H_0 = \mathcal{O}_{\prob}\mleft(\frac{1}{\delta_{dT}}\mright)
\text{ and }
H_2 - H_0 = \mathcal{O}_{\prob}\mleft(\frac{1}{\delta_{dT}}\mright),
\text{ as } d,T\to\infty,
\end{equation}
which is analogous to Lemma 6 in \cite{han2015tests}.
\end{itemize}
Note that $H_1, H_2$ have the same limit as $H$ which was introduced in \eqref{eq:HconvH0}. In particular, that means  asymptotically $\widetilde{F} \widetilde{H} \approx FH \approx FH_0$ with high probability.

The asymptotics \eqref{eq:key1} and \eqref{eq:key2} are the key results to proving the convergence of our test statistic under the null hypothesis. The bottom line is that our main contribution is to show that \eqref{eq:key1} and \eqref{eq:key2} are satisfied for the PCA estimator based on the transformed series $Y$.

Inferring Propositions \ref{prop1}--\ref{prop3} from Propositions \ref{prop2.1} and \ref{prop2.1.0} follows arguments in \cite{han2015tests}. Therefore, we omit those in the article. For completeness, we provide a supplementary document \cite{Supplement} with the detailed proofs.
Recall also that our estimation procedure and our proofs for consistent estimation of the underlying factor model follow the representation in \cite{doz2011two}. In contrast, \cite{han2015tests} follow \cite{bai2003inferential}. While both lead to the same results, they differ in their assumptions. The supplementary document \cite{Supplement} also has the purpose of clarifying that the arguments in \cite{han2015tests} go through by using the assumptions in \cite{doz2011two}.

\section{Power against the alternative hypothesis}\label{sec:alternative}
We consider the alternative that a portion of the cross sections have linearly independent loadings. In order to study the behavior of our test statistic under the alternative, we follow \cite{han2015tests} and \cite{duan2023quasi} and rewrite the model in terms of so-called pseudo-factors.

Specifically, let $Z_t$ denote the transformed series under the alternative, which admits the representation
\begin{align} \label{eq:rep_with_pseudo_factors_main}
\begin{pmatrix}
Z^{b} \\
Z^{a}
\end{pmatrix}
&=  
\begin{pmatrix}
G^{b} B' \\
G^{a} C'
\end{pmatrix}
\Theta' + \eta
=: G \Theta' + \eta,
\end{align}
where $Z^b,Z^a$ are the pre- and post-break samples, $G^b,G^a$ denote pseudo-factors, and $\Theta$ is a common loading matrix. This representation is observationally equivalent to a factor model without breaks in $\Theta$, but the matrices $B,C$ carry information about the structural break. 

Depending on the relationship between the pre- and post-break loading spaces, two representative alternatives arise:
\begin{itemize}
    \item \textit{TA 1 (Type~\ref{item:Type1} break).} The column spaces before and after the break are disjoint. In this case, the pseudo-factor representation implies that $(B,C)$ has column rank $2r$.
    \item \textit{TA 2 (Type~\ref{item:Type3} break).} The loading spaces partially overlap. Then $(B,C)$ has column rank strictly between $r$ and $2r$.
\end{itemize}
The explicit derivations of these two cases are provided in Appendix~\ref{app:alternative_derivations}. 

\medskip
Based on \eqref{eq:rep_with_pseudo_factors_main}, we define the PCA estimator $\widehat{G}$ of the pseudo-factors analogously to Section~\ref{sec3.2}. We further define
\begin{equation} \label{eq:Jdef_main}
J = (\Theta' \Theta /d) (G' \widehat{G} /T) \widehat{V}^{-1}_{\npf},
\end{equation}
where $\widehat{V}_{\npf}$ is the diagonal matrix of the $\npf$ largest eigenvalues of $\frac{1}{Td}\widehat{\Sigma}_T$. Then, one can show that $J \overset{\prob}{\to} J_0$ for some nonrandom matrix $J_0$.

For the two types of alternatives, we obtain population covariance matrices
\begin{equation} \label{eq:Di_main}
D_1 := \frac{1}{T}\E (BG^{b'}G^bB'), 
\qquad 
D_2 := \frac{1}{T}\E (CG^{a'}G^aC'),
\end{equation}
which differ in structure between TA~1 and TA~2 (see Appendix~\ref{app:alternative_derivations} for explicit forms). We then set
\begin{equation} \label{defJDDJ_main}
\mathcal{C} := J_0'(D_1 - D_2)J_0.
\end{equation}

Let $\Pi_{\Theta}$ denote the diagonal matrix of eigenvalues of $\Theta'\Theta$ in decreasing order. We impose the following assumptions.

\begin{Passump} \label{ass:P1_main} (Conditions on the break.)
\begin{enumerate}[label=(\roman*)]
\item There is a positive definite diagonal matrix $\widetilde{\Pi}$ such that $\| \Pi_{\Theta}/d - \widetilde{\Pi} \| = \mathcal{O}(d^{-1/2})$.
\item The eigenvalues of $\Pi_{\Theta}$ are distinct.
\end{enumerate}
\end{Passump}

\begin{Passump} \label{ass:P2_main} (Conditions on the long-run variance estimator.)
Recall $\mathcal{C}$ from \eqref{defJDDJ_main}.
\begin{enumerate}[label=(\roman*)]
\item Suppose
\[
\liminf_{T \to \infty} \Big\{ \vechop(\mathcal{C})' \Big( b_{T/2}\Omega(GJ_{0})^{-1} \Big) \vechop(\mathcal{C}) \Big\} \overset{\prob}{\to} c > 0,
\]
where $b_T$ is the bandwidth parameter.
\item Assumption \ref{ass:9}(i) holds for $\Omega(GJ_{0})$.
\end{enumerate}
\end{Passump}

\begin{proposition} \label{propAlternative0}
Under Assumptions \ref{ass:1}--\ref{ass:3}, \ref{ass:C1}--\ref{ass:C3}, and \ref{ass:P1_main}--\ref{ass:P2_main}:
\begin{enumerate}[label=(\roman*)]
\item For $\mathcal{C} \neq 0$ in \eqref{defJDDJ_main},
\[
\frac{1}{T} \sum_{t=1}^{T/2} \widehat{G}_{t} \widehat{G}_{t}' - \frac{1}{T} \sum_{t=T/2 + 1}^{T} \widehat{G}_{t} \widehat{G}_{t}' 
\overset{\prob}{\to} \mathcal{C}.
\]
\item The test statistic $W(\widehat{G})$ defined in \eqref{eq:def:Wald} is consistent under the alternative; i.e. $W(\widehat{G}) \overset{\prob}{\to} \infty$.
\end{enumerate}
\end{proposition}

The derivations of \eqref{eq:rep_with_pseudo_factors_main} and the explicit forms of $D_1,D_2$ are provided in Appendix~\ref{app:alternative_derivations}.
We refer to Appendix \ref{app:C} 
in the supplementary material for the proof of Proposition \ref{propAlternative0}.

\section{Simulation study} \label{sec:simulations_study}
We examine our testing procedure in a simulation study considering the two-subject setting (Section \ref{se:sim:two-subjects}) and the change-point analysis (Section \ref{se:sim:change}). 

\subsection{Two subjects} \label{se:sim:two-subjects}
We generate data for two subjects with $r$ factors. Recall the notation 
\begin{equation} \label{eq:two-subject_factor_mdoel_for_DGPs}
\begin{aligned}
X_{t}^{1} = \Lambda_{1} F_{t}^{1} + \varepsilon_{t}^{1}, \hspace{0.2cm} t =1, \dots, T, \\
X_{t}^{2} = \Lambda_{2} F_{t}^{2} + \varepsilon_{t}^{2}, \hspace{0.2cm} t =1, \dots, T, 
\end{aligned}
\end{equation}
where $\Lambda_{k} = (\lambda_{k,ij})_{i=1,\dots,\Dim;j=1,\dots,r}$, $k=1,2$. Under the null hypothesis, there are nonsingular matrices $\Phi_k \in \RR^{r \times r}$ such that $\Lambda_k = \Lambda \Phi_k$, $k=1,2$, and some common loading matrix $\Lambda = (\lambda_{ij})_{i=1,\dots,\Dim;j=1,\dots,r}$ with full column rank.
We consider the following data generating models for the null hypothesis (NDGP):
\begin{enumerate}[label=NDGP\arabic*:, align=left]
\item 
Independent factors and errors ($r=3$):
Choose $\lambda_{ij} \sim \mathcal{N}(0,1)$. Set 
$\Phi_1 = I_r$ and $\Phi_2$ is randomly generated with eigenvalues in the interval $(0.75,1.25)$.
The factors are $F_{i,t}^{k} \sim \mathcal{N}(0,1)$.
For the errors, suppose $\varepsilon_{i,t}^{k} \sim \mathcal{N}(0,1)$, $k=1,2$.
\item 
Temporal correlation in the errors ($r=3$):
Choose $\lambda_{ij} \sim \mathcal{N}(0,1)$. The matrices $\Phi_k$, $k=1,2$ are randomly generated with eigenvalues in the interval $(0.75,1.25)$.
The factors are $F_{i,t}^{k} \sim \mathcal{N}(0,1)$.
For the errors, suppose $\varepsilon_{i,t}^{k} = \sigma_i v_{i,t}$, $\sigma_i \sim \mathcal{U}(0.5,1.5)$, $v_{i,t} = 0.5 v_{i,t-1} + e_{i,t}$, 
$e_{i,t}\sim \mathcal{N}(0,1)$.
\item 
Temporal correlation in the factors ($r=3$):
Choose $\lambda_{ij} \sim \mathcal{N}(0,1)$. The matrices $\Phi_k$, $k=1,2$ are randomly generated with eigenvalues in the interval $(0.75,1.25)$.
The factors are generated as $F_{i,t}^{k} = 0.5 F_{i,t-1}^{k} + e_{i,t}$, $e_{i,t} \sim \mathcal{N}(0,1-0.5^2)$.
For the errors, suppose $\varepsilon_{i,t}^{k} \sim \mathcal{N}(0,1)$, $k=1,2$.
\item 
Different DGPs for each subject ($r=3$):
Choose $\lambda_{i,j} \sim \mathcal{N}(0,1)$.
The matrices $\Phi_k$, $k=1,2$, are randomly generated with eigenvalues in the interval $(0.75,1.25)$. 
The factors are $F_{i,t}^{k} \sim \mathcal{N}(0,1)$, $k=1,2$.
For the errors, suppose $\varepsilon_{i,t}^{1} \sim \mathcal{N}(0,1)$
and $\varepsilon_{i,t}^{2} = \sigma_i v_{i,t}$, $\sigma_i \sim \mathcal{U}(0.5,1.5)$, $v_{i,t} = 0.5 v_{i,t-1} + e_{i,t}$, 
$e_{i,t}\sim \mathcal{N}(0,1)$.
\end{enumerate}
The corresponding empirical sizes over $1000$ replications can be found in Table \ref{tab:twosubjectssize}. NDGP1 does not allow for any serial correlation of factors or idiosyncratic errors, while NDGP2--NDGP4 allow for cross-sectional heteroskedasticity. NDGP3 assumes that the factors follow an AR$(1)$ model. All simulations show reasonable results for the test sizes.

\begin{table}[ht]
    \centering
    \begin{tabular}{|c||c|c|c|c|c|}
    \hline
         $(d,T)$ & NDGP1 & NDGP2 & NDGP3 & NDGP4 \\ \hline \hline
        (200,200) & \gradient{0.0275} & \gradient{0.0295} & \gradient{0.0505} & \gradient{0.031}   \\ \hline
        (500,500) & \gradient{0.0325} & \gradient{0.052} & \gradient{0.0495} & \gradient{0.051}    \\ \hline
    \end{tabular}
        \caption{Empirical test size for the setting with two individuals sharing the same loading matrix up to a rotation with a nonsingular matrix.}
    \label{tab:twosubjectssize}
\end{table}

The following list provides the data generating models under the alternative (ADGP):
\begin{enumerate}[label=ADGP\arabic*:, align=left]
\item 
Additive change across all loading vectors ($r=3$):
Choose $\lambda_{ij} \sim \mathcal{N}(b/2,1)$ and $\Lambda_{1} = \Lambda$, $\Lambda_{2} = \Lambda - b$.
The factors are $F_{i,t}^{k} \sim \mathcal{N}(0,1)$.
For the errors, suppose $\varepsilon_{i,t}^{k} \sim \mathcal{N}(0, 1+b^2/4)$, $k=1,2$.
\item 
Additive change across a partial set of loading vectors ($r=3$):
Choose $\lambda_{ij} \sim \mathcal{N}(b/2,1)$ and $\Lambda_{1} = \Lambda$, 
$\lambda_{2,ij} = \lambda_{1,ij} -b$ for $i=1, \dots, a \cdot \Dim$ and $j=1,\dots,r$ and $\lambda_{2,ij} = \lambda_{1,ij}$ for $i=a \cdot \Dim+1,\dots,\Dim$ and $j=1,\dots,r$.
The factors are $F_{i,t}^{k} \sim \mathcal{N}(0,1)$.
For the errors, suppose $\varepsilon_{i,t}^{k} \sim \mathcal{N}(0, 1+b^2/4)$, $k=1,2$.
\item 
Partial change in the loadings plus an invertible linear transformation ($r=4$):
Set 
$ \Lambda_{1} = \Lambda$ and $ \Lambda_{2} = \mleft( \Pi_{2} \ \Lambda \Phi \mright)$ with 
$\lambda_{ij} \sim \mathcal{N}(0,1)$. Let $\Pi_2 = (\pi_{2,ij})_{i=1,\dots,\Dim;j=1,\dots,(r-c)} $, $\pi_{2,ij} \sim \mathcal{N}(0,2)$ and $\Phi \in \RR^{r \times c}$ randomly generated. 
The factors are $F_{i,t}^{k} \sim \mathcal{N}(0,1)$.
For the errors, suppose $\varepsilon_{i,t}^{k} \sim \mathcal{N}(0,1)$, $k=1,2$.
\item Loading matrices drawn from different distributions ($r=4$): Same as ADGP3 but with $\lambda_{2,ij} \sim \mathcal{C}auchy(0,1)$.
\end{enumerate}

\begin{table}[ht]
    \centering
    \begin{tabular}{|c||c|c||c|c||c|c||c|c|}
    \hline
         $(d,T)$ & $b$ & ADGP1 		& $a$ & ADGP2 		& $c$ & ADGP3 		& $c$ & ADGP4  \\ \hline \hline
        (200,200) & $1/3$ &\gradient{0.251}  & 0.2 & \gradient{0.3775} & 1	&\gradient{1}   		& 1 & \gradient{1}	\\ \hline
        (500,500) & $1/3$ &\gradient{0.6815}& 0.2 & \gradient{0.8465} & 1 	&\gradient{1}  		& 1 & \gradient{1}    \\ \hline \hline
        (200,200) & $2/3$ &\gradient{1}  & 0.4 & \gradient{0.9885} & 2	&\gradient{1}   	& 2 & \gradient{1} \\ \hline
        (500,500) & $2/3$ &\gradient{1} 		& 0.4 & \gradient{1}  & 2	&\gradient{1}   	& 2 & \gradient{1}  \\ \hline \hline
        (200,200) & $1$    &\gradient{1} 	& 0.6 & \gradient{1} & 3	&\gradient{1}    	& 3 & \gradient{1} \\ \hline
        (500,500) & $1$    &\gradient{1} 		& 0.6 & \gradient{1} & 3	&\gradient{1}    	& 3 & \gradient{1}  \\ \hline \hline       
    \end{tabular}
        \caption{Empirical test power for the setting with two individuals under different alternatives.}
    \label{tab:twosubjectspower}
\end{table}

The corresponding empirical powers can be found in Table \ref{tab:twosubjectspower}. 
ADGP1 focuses on an additive change in the loadings, for varying magnitudes $b \in \{ 1/3, 2/3, 1\}$. The model is very similar to DGP A1 in \cite{han2015tests} for changes in the loadings over time.
Under ADGP1, our test does not have good power when the magnitude of the break is small. With larger magnitude of the break, the test becomes more powerful, though the power is not monotonically increasing in $b$. This observation aligns with observations made for DGP A1 in \cite{han2015tests}.

Similarly, ADGP2 aligns with DGP A2 in \cite{han2015tests} and considers additive differences in a fraction of the loadings. The proportion of loadings that are different increases with $a \in \{ 0.2, 0.4, 0.6\}$ and $b$ is set to $1$. Our test struggles to achieve good power for small fraction $a =0.2$. However, it increases with sample size and $a$. 

In ADGP3, we consider partial changes up to an invertible linear transformation in the loadings concatenated with linearly independent column vectors. The proportion of the linearly independent vectors decreases with $c \in \{1,2,3\}$.

The last model, ADGP4, is the simplest to detect since both loading matrices are drawn from different distributions. Our test does not have any issues to detect those differences.

\subsection{Change-point} \label{se:sim:change}
We generate data with a change at $\lfloor \pi T \rfloor$ and  $r_1$ pre- and $r_2$ post-break factors:
\begin{equation} \label{eq:cp_model_simulations}
\begin{aligned}
X_{t} = 
\begin{cases}
\Lambda_{1} F_{t} + \varepsilon_{t}, \hspace{0.2cm} t =1, \dots, \lfloor \pi T \rfloor, \\
\Lambda_{2} F_{t} + \varepsilon_{t}, \hspace{0.2cm} t =\lfloor \pi T \rfloor + 1, \dots, T.
\end{cases}
\end{aligned}
\end{equation}
We consider the following settings for the factors and loadings $\Lambda_{1}, \Lambda_{2}$ under the null hypothesis with $r=r_{1}=r_{2}=3$:
\begin{enumerate}[label=NDGPcp\arabic*:, align=left]
\item Independent factors and errors:
Choose $\lambda_{ij} \sim \mathcal{N}(0,1)$. Set 
$\Phi_1 = I_r$ and $\Phi_2$ is randomly generated with eigenvalues in the interval $(0.75,1.25)$.
The factors are $F_{i,t} \sim \mathcal{N}(0,1)$.
For the error, suppose $\varepsilon_{i,t} \sim \mathcal{N}(0,1)$.
\item Temporal correlation in the errors: 
Choose $\lambda_{ij} \sim \mathcal{N}(0,1)$. The matrices $\Phi_k$, $k=1,2$, are randomly generated with eigenvalues in the interval $(0.75,1.25)$.
The factors are $F_{i,t} \sim \mathcal{N}(0,1)$.
For the errors, suppose $\varepsilon_{i,t} = \sigma_i v_{i,t}$, $\sigma_i \sim \mathcal{U}(0.5,1.5)$, $v_{i,t} = 0.5 v_{i,t-1} + e_{i,t}$, 
$e_{i,t}\sim \mathcal{N}(0,1)$.
\end{enumerate}
Under the alternative with $r=r_{1}=r_{2}=3$:
\begin{enumerate}[label=ADGPcp\arabic*:, align=left]
\item 
Additive change across all loading vectors:
Choose $\lambda_{ij} \sim \mathcal{N}(b/2,1)$ and $\Lambda_{1} = \Lambda$, $\Lambda_{2} = \Lambda - b$.
The factors are $F_{i,t} \sim \mathcal{N}(0,1)$.
For the errors, suppose $\varepsilon_{i,t} \sim \mathcal{N}(0, 1+b^2/4)$, with $b=1$.
\item 
Additive change across a partial set of loading vectors: 
Choose $\lambda_{ij} \sim \mathcal{N}(b/2,1)$ and $\Lambda_{1} = \Lambda$, 
$\lambda_{2,ij} = \lambda_{1,ij} -b$ for $i=1, \dots, a \cdot \Dim$ and $j=1,\dots,r$ and $\lambda_{2,ij} = \lambda_{1,ij}$ for $i=a \cdot \Dim+1,\dots,\Dim$ and $j=1,\dots,r$.
The factors are $F_{i,t} \sim \mathcal{N}(0,1)$.
For the errors, suppose $\varepsilon_{i,t} \sim \mathcal{N}(0,1+b^2/4)$, with $b=1$ and $a=0.4$.
\item 
Partial change in the loadings plus a change up to an invertible linear transformation: 
We suppose here that $r=4$
Set 
$ \Lambda_{1} = \Lambda$ and $ \Lambda_{2} = \mleft( \Pi_{2} \ \Lambda \Phi \mright)$ with 
$\lambda_{ij} \sim \mathcal{N}(0,1)$. Let $\Pi_2 = (\pi_{2,ij})_{i=1,\dots,\Dim; j=1,\dots,(r-c)} $, $\pi_{ij} \sim \mathcal{C}auchy(0,1)$ and $\Phi \in \RR^{r \times c}$ randomly generated. 
The factors are $F_{i,t} \sim \mathcal{N}(0,1)$.
For the errors, suppose $\varepsilon_{i,t} \sim \mathcal{N}(0,1)$.
\end{enumerate}

Table \ref{tab:cpsize} shows the rejection rates for the different DGPs under the null hypothesis. 
As observed in a comprehensive comparison of different change-point tests in \cite{su2017time}, the Wald type test of \cite{han2015tests} tends to under-reject the null hypothesis. This can also be observed in the simulation results for our test.

Note that the effective sample size in the change-point setting is smaller than in the two-subject model setting since the pre- and post-break samples are of length $\lfloor T\pi \rfloor$ and $T-\lfloor T\pi \rfloor$. The change-point test is then applied to a transformed series of length $\lfloor T\pi \rfloor$.
In Table \ref{tab:cpsize}, one can see that the test size increases with the sample size and with changes closer to the boundary. A smaller $\pi$ has the effect that we have more data available to create our artificial change-point by transforming the longer sequence of data $1-\lfloor T\pi \rfloor$. The estimation of the projection matrices happens based on $\lfloor T\pi \rfloor$th data points and does not seem to be impacted by fewer data points.
\begin{table}[ht]
    \centering
    \begin{tabular}{|c|c||c|c|}
    \hline
         $(d,T)$ & $\pi$ & NDGPcp1 & NDGPcp2 \\ \hline \hline
        (200,200) & 	0.5	& \gradient{0.018} & \gradient{0.0285}    \\ \hline
        (500,500) & 	0.5	& \gradient{0.035} & \gradient{0.0425}   \\ \hline \hline
        (200,200) & 	0.4	& \gradient{0.0215} & \gradient{0.034}    \\ \hline
        (500,500) & 	0.4	& \gradient{0.025} & \gradient{0.044}     \\ \hline \hline
        (200,200) & 	0.3	& \gradient{0.016} & \gradient{0.044}   \\ \hline
        (500,500) & 	0.3	& \gradient{0.037} & \gradient{0.0495}    \\ \hline \hline
        (200,200) & 	0.2	& \gradient{0.023} & \gradient{0.0585}   \\ \hline
        (500,500) & 	0.2	& \gradient{0.038} & \gradient{0.049}    \\ \hline \hline
    \end{tabular}
        \caption{Empirical test size for different change-point locations and different DGPs under the hypothesis.}
    \label{tab:cpsize}
\end{table}

For comparison, we applied the change-point test of \cite{han2015tests} to NDGPcp1 with $\pi=0.5$. The resulting values for $d,T=200$ and $d,T=500$ are respectively $0.8225$ and $0.9995$. This confirms that while our test does not reject due to a change in form of a linear transformation in the loadings, a regular test for structural changes in the loadings does.

\begin{table}[ht]
    \centering
    \begin{tabular}{|c|c||c|c|c|}
    \hline
         $(d,T)$ & $\pi$ & ADGPcp1 & ADGPcp2 			& ADGPcp3 \\ \hline \hline
        (200,200) & 	0.5	& \gradient{0.975} & \gradient{0.662}	 & \gradient{0.712}    \\ \hline
        (500,500) & 	0.5	& \gradient{1} & \gradient{0.9905}	    & \gradient{1} \\ \hline \hline
        (200,200) & 	0.4	& \gradient{0.999} & \gradient{0.799}& \gradient{0.9985}    \\ \hline
        (500,500) & 	0.4	& \gradient{1} & \gradient{0.998}     	 & \gradient{1} \\ \hline \hline
        (200,200) & 	0.3	& \gradient{1} & \gradient{0.8685}	 & \gradient{1} \\ \hline
        (500,500) & 	0.3	& \gradient{1} & \gradient{0.9995}	 & \gradient{1} \\ \hline \hline
        (200,200) & 	0.2	& \gradient{1} & \gradient{0.9275}	& \gradient{1}  \\ \hline
        (500,500) & 	0.2	& \gradient{1} & \gradient{0.9995}	& \gradient{1}  \\ \hline \hline
    \end{tabular}
        \caption{Empirical test power for different change-point locations.}
    \label{tab:cppower}
\end{table}
In Table \ref{tab:cppower}, one can find the rejections rates under the alternative.
For ADGPcp2 the test suffers under the relatively small sample size $T=200$. However, we achieve good power for the other two DGPs.

\section{Data application: US macroeconomic dataset of Stock and Watson} \label{sec:data_application}
To illustrate our procedure, we study the US macroeconomic dataset of \cite{stock2009forecasting}. The original dataset contains 108 monthly and 79 quarterly time series of US nominal and real variables, including prices, interest rates, money and credit aggregates, stock prices, exchange rates. We use here preprocessed data following the suggestions in Table A.1 in \cite{stock2009forecasting}. 
The US quarterly data are taken from the DRI/McGraw-Hill Basic Economics database of 1999.\footnote{https://dataverse.unc.edu/dataset.xhtml?persistentId=hdl:1902.29/D-17267} Following \cite{stock2009forecasting}, we then removed some high level aggregates related by identities to the lower level sub-aggregates and ended up with $d=109$ time series, spanning from 1959:Q3--2006:Q4 and having a length of $T=190$.

This dataset has been studied in many methodological papers concerning factor models. We refer to Section 6 in \cite{baltagi2021estimating} for a detailed analysis of the dataset, including the estimation of multiple change-points in the loading matrices and the number of factors in each sub-interval. 
\cite{baltagi2021estimating} detected two change-points 1979:Q1 and 1983:Q4.
In particular, \cite{baltagi2021estimating} hypothesized that the first break could be due to the impact of the Iranian revolution on the oil price and US inflation.
The second break could be due to the great moderation, which is also considered in \cite{chen2014detecting} and \cite{ma2018estimation}.
The corresponding estimated number of factors in the three regimes are $r_{1} = 3, r_{2} = 3$ and $r_{3} = 4$, respectively.
Our test allows one now to study what type of breaks occur. 

First, we apply our test modified for the analysis of change-points to the first and second regime. Our test rejects the null hypothesis at a $5\%$ significance level. 
In other words, we reject the hypothesis that the loading matrices in the first and second regimes can be expressed as linear combinations of each other. This gives evidence that a non-negligible portion of the cross sections have linearly independent loadings.

Secondly, we apply our test for different numbers of factors as presented. For this modification of our test, we refer the reader to Appendix \ref{se:Different numbers of factors}. The test is here adapted to the change-point setting to the second and third regimes. In this case, our test does not reject the null hypothesis at a $5\%$ significance level. That is, there is enough evidence to conclude that there is a change in the loading matrices, but not that those loadings are linearly independent.
As mentioned above, the second break has been considered in many works, including \cite{chen2014detecting,ma2018estimation,baltagi2021estimating}. It has been hypothesized to be due to the great moderation. The great moderation refers to a period of decreased macroeconomic volatility experienced in the United States starting in the 1980s. 
For example, during this period, the standard deviation of quarterly real gross domestic product declined by half and the standard deviation of inflation declined by two-thirds, according to figures reported by former US Federal Reserve Chair Ben Bernanke.\footnote{The Federal Reserve Board. ``Remarks by Governor Ben S. Bernanke. https://www.federalreserve.gov/boarddocs/speeches/2004/20040220/".}
Our test might indicate that the changes in volatility happened in a somewhat more ``structured'' way alluding to similarities between the pre- and post-break volatilities of the factors up to transformation with an invertible matrix. We also refer to Example \ref{example:1} for one possibility of what a more ``structured'' change in volatility can look like.


\section{Conclusions}\label{sec:conclusion}
This work advances the analysis of factor models in several directions. On the one hand, the proposed hypothesis test is motivated by learning relationships across multiple subjects, also known as multi-subject learning. In particular, it allows discovering whether structural differences between two models are due to a linear transformation. 
On the other hand, our test pushes forward the framework of testing for structural changes over time in high-dimensional factor models. Our test allows to further characterize what type of change occurs. 

The questions of interest required the development of a new hypothesis test and its theoretical investigation under the null hypothesis and the alternative. Our simulation study indicates a good numerical performance of our test. In addition, we present interesting new findings on the celebrated US macroeconomic data set of Stock and Watson.

Despite our contributions, there is still a lot of room for future work on partial testing, in particular, on finding an efficient way of testing whether a partial set of vectors in the loadings is shared across individuals or between pre- and post-break point. Another direction is the investigation of other change-point tests. Since our procedure does not rely on a specific change-point test, one could use our proposed data transformation to apply other such tests. Furthermore, it would be of interest to test whether the rotational transformation matrix is of a certain form, for example diagonal.

\small
\bibliographystyle{plainnat}
\bibliography{multiFactor_bib}

\newpage

\input{multiFactor_supp}


\flushleft
\begin{tabular}{lp{1 in}l}
Marie-Christine D\"uker & & Vladas Pipiras\\
Dept.\ of Mathematics & & Dept.\ of Statistics and Operations Research\\
Technical University of Munich & & UNC Chapel Hill\\
Parkring 11-13/II & & CB\#3260, Hanes Hall\\
85748 Garching b. München & & Chapel Hill, NC 27599, USA\\
{\it marie.dueker@tum.de} & & {\it pipiras@email.unc.edu}\\
\end{tabular}

\end{document}

%% file: multiFactor_supp.tex
\makeatletter
\gdef\@thanks{}
\makeatother

\title{Supplementary material for ``Testing common structure in high-dimensional factor models:
change-point and two-sample procedures"
\thanks{AMS subject classification. Primary: 62H25, 62M10. Secondary: 62F03.}
\thanks{Keywords: High-dimensional time series, factor models, multi-task learning, change-point analysis.}
}
\author{
Marie-Christine D\"uker \\ Technical University of Munich                            \and
Vladas Pipiras               \\ University of North Carolina}
\date{\today}

\maketitle

\bigskip

This supplement contains additional data analysis, methodological extensions, simulation evidence, and detailed proofs for the article ``Testing common structure in high-dimensional factor models: change-point and two-sample procedures.''

As explained in the main text, our methodology proceeds by transforming a series that is assumed to follow a high-dimensional factor model and then applying an available change-point test to the transformed data. A central theoretical step is therefore to show that the PCA estimators based on the transformed series remain consistent. These results are stated in Propositions \ref{prop2.1} and \ref{prop2.1.0}, and their proofs are provided here.

The supplement is organized as follows. Appendix \ref{sec:data_application_2} presents an additional data application for US and South Korea macroeconomic indices. Appendix \ref{se:extensionsapplications} discusses variations of the testing approach, including the case of different numbers of factors, a supremum-based approach, and alternative test statistics. Appendix \ref{app:CC} reports further simulation results for these extensions and provides additional discussion of the simulation design. Appendix \ref{app:alternative_derivations} supplements Section \ref{sec:alternative} by developing the derivations under the alternative hypothesis in greater detail.

The remaining appendices focus on the theoretical arguments underlying the main results. Appendix \ref{se:appB} establishes auxiliary results under the null hypothesis, and Appendix \ref{se:consistency_PCA_NH} gives the consistency results for the PCA estimators under the null. These results are then used to derive the asymptotic behavior of our test statistic in Propositions \ref{prop1}--\ref{prop3}. Appendix \ref{app:D} contains the detailed proofs of the main results from Appendix \ref{app:D}, while Appendix \ref{app:H} collects additional auxiliary lemmas used in those proofs. The result under the alternative hypothesis, stated in Proposition \ref{propAlternative0}, is treated in Appendix \ref{app:C}, which presents the corresponding analysis, auxiliary results, and consistency statements under the alternative. Finally, Appendix \ref{se:matrixnorminequalities} collects matrix norm inequalities used throughout the supplement.

More broadly, this supplement serves two purposes. First, it provides a complete record of the supplementary empirical results, methodological extensions, and proofs underlying the procedures proposed in the paper. Second, it clarifies how the proof strategy in \cite{han2015tests} can be adapted to our framework when the estimation of the underlying factor model is developed under the assumptions of \cite{doz2011two}. Our estimation procedure and consistency proofs follow \cite{doz2011two}, whereas \cite{han2015tests} build on \cite{bai2003inferential}. Although these approaches lead to the same type of conclusions, they rely on different assumptions, and one aim of this supplement is to show that the arguments in \cite{han2015tests} continue to hold in the setting considered here.

\appendix

\section{Additional data application (US and South Korea macroeconomic indices)} \label{sec:data_application_2}
To illustrate our two-subject procedure, we consider 65 macroeconomic monthly indices for the United States and South Korea from July 2012 to December 2024. 
The US data are taken from FRED \cite{stock2009forecasting}.
The South Korea dataset was introduced in \cite{baek2025kred} and presents a novel dataset with the same indices as the FRED data up to some minor differences. The original data consist of 71 macroeconomic indices, collected monthly for South Korea. 
Since some of the indicators do not exactly match in the way they were collected, we focus on those that are the same for both the US and South Korea. 
In addition, the data after early 2020 are affected by the COVID-19 pandemic. COVID-19 is associated with outliers in the series which account for much of the variability and affect the nature of factors. In order to study the series without this effect, we consider the period July 2012 to December 2019.  
Our resulting dataset consists of two multivariate time series of sample size $T = 90$ and $d=65$.

We also note that we performed transformations according to \cite{stock2009forecasting} to achieve stationarity.
These transformations are taking the logarithm, the differences and removing seasonality. We applied the same transformations as proposed in \cite{stock2009forecasting} to both the US and the South Korea data for the respective indices.

Since the two series are expected to have some dependence, we note that our test does not require independence across subjects. Our procedure rather aims to test whether the dependence across the two panels can be explained by similarities among the corresponding loading matrices.

We first apply an information criterion introduced in \cite{bai2002determining} to estimate the number of factors. The criterion is applied to the US data ($\widehat{r}_1 = 7$), the South Korea data ($\widehat{r}_2=6$) and after concatenating across time ($\widehat{r} = 7$). Based on those estimates, we study two different scenarios in Table \ref{ta:example2} below. We consider the case when both series are explained by $r_1=r_2=7$ factors and when $r_1=7, r_2=6$.

We consider three tests. First, we test the hypothesis $H_0: \Lambda_1 = \Lambda_2$ and then $H_0: \Lambda_1 = \Lambda_2 \Phi$ assuming that the number of factors for both series are $r_1 = r_2 = 7$. For the hypothesis $H_0: \Lambda_1 = \Lambda_2$, we concatenate the data over time and test for changes in the loadings at the concatenation point using the test by \cite{han2015tests}. Subsequently, we test for $H_0: \Lambda_1 = \Lambda_2 \Phi$ using our proposed test. Finally, we consider the test proposed in Appendix \ref{se:Different numbers of factors} for different numbers of factors ($r_1 = 7, r_2 = 6$).

Table \ref{ta:example2} shows the respective values of the test statistics.
The corresponding critical values are given in the last column. The colored cell indicates that the test statistic exceeds the respective critical value and we reject the hypothesis. One can observe that $H_0: \Lambda_1 = \Lambda_2$ is rejected but our test does not reject the hypothesis of the loadings being the same up to rotation. This kind of result could indicate that the factors driving both datasets only differ by changes in volatility. 
This is not unexpected since South Korea is typically associated with extensive price regulation (e.g., energy, utilities), a credible inflation-targeting framework, and financial stability and foreign exchange policies that damp short-run shocks. In contrast, the US’s more market-driven price setting, broader exposure to commodity cycles, and more market-based financial system allow faster pass-through of shocks, expected to lead to higher volatility.

\begin{table}[h!]
\centering
\begin{tabular}{|c||c||c|c||c||c|}
\hline
Time period	&	$H_0$	& $r_1$ & $r_2$	&  $W(\widehat{F})$& $cv_{r_2}$						\\ \hline\hline
\multirow{3}{*}{July 2012--Dec. 2019}	&	$\Lambda_1 = \Lambda_2$		& 7 & 7	& \cellcolor{lightpurple} 41.809 & \multirow{2}{*}{41.337}		\\ 
\cline{2-2} 
                    &	$\Lambda_1 = \Lambda_2 \Phi$		&  7 & 	7	& 30.416 &\\ \cline{2-2}\cline{6-6}
                         &	$\Lambda_1 = \Lambda_2 \widetilde{\Phi}$	& 	7	 & 	6	& 24.929 &	32.671	\\ \hline
\end{tabular}
\caption{The values of the test statistics for $H_0:\Lambda_1 = \Lambda_2$, $H_0:\Lambda_1 = \Lambda_2 \Phi$ and $H_0:\Lambda_1 = \Lambda_2 \widetilde{\Phi}$, $\widetilde{\Phi} \in \RR^{r_2 \times r_1}$, and the corresponding critical values $cv_{r_2} = \chi^2(r_2(r_2+1)/2)$.
The colored cell indicates that the test statistic exceeds the respective critical value and we reject the null hypothesis. The estimated numbers of factors are $\widehat{r}_1 = 7$, $\widehat{r}_2=6$ $\widehat{r} = 7$ based on information criterion introduced in \cite{bai2002determining}.} 
\label{ta:example2}
\end{table}




\section{Variations of the testing approach} \label{se:extensionsapplications}
 
In Appendix \ref{se:Different numbers of factors}, we discuss how our approach can be generalized to a setting where the two factor models are driven by different numbers of factors. In Appendix \ref{se:alternative_approach}, we present an extension of our approach using a different type of change-point test and in Appendix \ref{se:alternativetests} we present alternative ways of transforming the data. We end with a discussion of our test in Appendix \ref{se:Discussion}.
 
 \subsection{Different numbers of factors} \label{se:Different numbers of factors}
 
In this section, we allow the number of factors to be different for the two series under consideration. Suppose $\Lambda_{1} \in \RR^{\Dim \times r_{1}}$ and $\Lambda_{2} \in \RR^{\Dim \times r_{2}}$ with $r_{2} \leq r_{1}$ and 
\begin{equation} \label{eq:new1}
\Lambda_{2} = \Lambda_{1} \Phi_{2},
\end{equation}
where $\Phi_{2} \in \RR^{ r_{1} \times r_{2}} $ is a matrix of full column rank.

Note that our procedure for $r_{1} = r_{2}$ relies heavily on the invariance under invertible transformation of the projection matrices; see \eqref{eq:projection_invariance_for_each_P}. In particular, it ensures that the error terms of the new series \eqref{eq:YT2def} are stationary over time; see \eqref{eq:motivation_null_population_errors}.
This property is not satisfied when $r_{2} \neq r_{1}$, in particular, even under the null hypothesis, 
 \begin{equation*} 
P_{1} 
=
\Lambda_{1} ( \Lambda'_{1} \Lambda_{1} )^{-1} \Lambda'_{1}
\neq 
\Lambda_{1}\Phi_{2} ( \Phi'_{2} \Lambda'_{1} \Lambda_{1} \Phi_{2} )^{-1} \Phi'_{2} \Lambda'_{1} 
=P_{2}.
 \end{equation*}
 Alternatively, we can take advantage of the relationship
 \begin{equation} \label{eq:new3}
 \begin{aligned}
 P_{2} P_{1} 
 &= 
\Lambda_{2} ( \Lambda'_{2} \Lambda_{2} )^{-1} \Lambda'_{2} \Lambda_{1} ( \Lambda'_{1} \Lambda_{1} )^{-1} \Lambda'_{1}
 \\&= 
\Lambda_{2} ( \Lambda'_{2} \Lambda_{2} )^{-1} \Phi_{2}' \Lambda_{1}' \Lambda_{1} ( \Lambda'_{1} \Lambda_{1} )^{-1} \Lambda'_{1}
 = 
\Lambda_{2} ( \Lambda'_{2} \Lambda_{2} )^{-1} \Lambda'_{2}
=
P_{2}.
 \end{aligned}
 \end{equation}
Then, we can follow the same approach as in \eqref{eq:YT2def} and transform the series as
\begin{equation} \label{eq:new4}
Y_{1:T_{2}} = \mleft( (\widehat{P}_{2}\widehat{P}_{1}) X^{2}_{1:T_{2}/2}, \widehat{P}_{2} X^{2}_{ (T_{2}/2 +1) :T_{2}} \mright),
\end{equation}
where $\widehat{P}_{1}, \widehat{P}_{2}$ are as in \eqref{eq:def_Phat12} such that at the population level
\begin{equation} \label{eq:new4.1}
\begin{aligned}
P_{2} P_{1} X_{t}^2 
&\approx P_{2} P_{1} \Lambda_{2} F_{t}^2 
= \Lambda_{2} F_{t}^2
\hspace{0.2cm} \text{ for } t =1,\dots, T_{2}/2, \\
P_{2} X_{t}^2 
&\approx P_{2} \Lambda_{2} F_{t}^2 = \Lambda_{2} F_{t}^2
\hspace{0.2cm} \text{ for } t =T_{2}/2+1,\dots, T_{2}.
\end{aligned}
\end{equation}
Note that the roles of the first and the second series are interchangeable. 

Since $P_{2} P_{1} = P_{2}$ under the null hypothesis, the idiosyncratic errors of the transformed series $Y_{1:T_{2}}$ are still stationary. In addition, the errors have covariances
\begin{equation} \label{eq:new5}
\begin{aligned}
P_{2} P_{1} \E (\varepsilon_{t}^2 \varepsilon_{t}^{2'}) (P_{2} P_{1} )' = P_{2} \Sigma_{\varepsilon} P_{2}', \\
P_{2} \E (\varepsilon_{t}^2 \varepsilon_{t}^{2'}) P_{2}' = P_{2} \Sigma_{\varepsilon} P_{2}'.
\end{aligned}
\end{equation}
Due to \eqref{eq:new3}, the series $Y_{1:T_{2}} $ in \eqref{eq:new4} is expected to have the same loadings and factors under the null \eqref{eq:new1} as $X^{2}_{1:T_{2}} $. However, given \eqref{eq:new5}, the errors have a different covariance matrix than the original series.

Under the alternative, we get
\begin{equation} \label{eq:new6}
\begin{aligned}
P_{2} P_{1} X_{t}^2 
&\approx P_{2} P_{1} \Lambda_{2} F_{t}^2 
\neq \Lambda_{2} F_{t}^2
\hspace{0.2cm} \text{ for } t =1,\dots, T_{2}/2, \\
P_{2} X_{t}^2 
&\approx P_{2} \Lambda_{2} F_{t}^2 = \Lambda_{2} F_{t}^2
\hspace{0.2cm} \text{ for } t =T_{2}/2+1,\dots, T_{2}.
\end{aligned}
\end{equation}
The illustrated behavior of $Y_{1:T_{2}}$ at the population level in \eqref{eq:new4.1} and \eqref{eq:new6} suggests that for the transformed series $Y_{1:T_{2}}$ our hypothesis testing problem is equivalent to testing for a change at time $T_{2}/2$ in the loadings. 
In other words, we can proceed as in Section \ref{sec3.2} and apply an available change-point test as the one of \cite{han2015tests}.

\subsection{Supremum approach} \label{se:alternative_approach}
Our approach suggests to split the data into two halves and applies different transformations to each of them. This way, we artificially create a change in the loadings under the alternative, while there is no change under the null hypothesis. In particular, this means that in the case of a change, we know its location. We therefore apply a test comparing the pre- and post-sample means of the covariances of the factors using the exact change-point.

However, a potential alternative approach is to consider the supremum over a certain window around the artificially created change-point location $T/2$. 
The hope is to improve our test's power while not loosing too much in terms of its size. However, a trade-off between power and size is expected, at least to some degree. See also the corresponding simulation study in Appendix \ref{se:sim:alternative_approach} below.

For the supremum approach, we introduce
\begin{equation*} \label{eq:def:V(G)_pi}
V(\pi, \widehat{F}) = \vechop \mleft( 
\sqrt{T} \Bigg(
\frac{1}{ \lfloor \pi T \rfloor } \sum_{t=1}^{\lfloor \pi T \rfloor} \widehat{F}_{t} \widehat{F}_{t}' - \frac{1}{T-\lfloor \pi T \rfloor} \sum_{t=\lfloor \pi T \rfloor + 1}^{T} \widehat{F}_{t} \widehat{F}_{t}' \Bigg) \mright).
\end{equation*}
To construct a Wald type test statistic, we also need the sample long-run variance of $V(F H_{0})$.
We use the following estimator for the long-run variance \eqref{eq:Omega},
\begin{equation*} \label{eq:def:Omega(G)_pi}
\Omega(\pi, F H_{0}) = \mleft( \frac{1}{\pi} + \frac{1}{1-\pi} \mright) \Omega(F H_{0})
\end{equation*}
with $\Omega(F H_{0})$ as in \eqref{eq:def:Omega(G)}.
Then, the Wald type test statistic across a certain interval can be defined as
\begin{equation} \label{eq:def:Wald_pi}
\sup_{ \pi \in [ \pi_0, 1- \pi_0]} W(\pi,\widehat{F}) 
\hspace{0.2cm}
\text{ with }
\hspace{0.2cm}
W(\pi,\widehat{F}) = V(\pi,\widehat{F})' \Omega^{-1}(\pi,\widehat{F}) V(\pi,\widehat{F})
\end{equation}
for $\pi_0 \in (0,\frac{1}{2}]$. The test statistic is expected to satisfy the following convergence under the null hypothesis,
\begin{equation}\label{eq:con:Wald_pi}
\sup_{ \pi \in [ \pi_0, 1- \pi_0]} W(\pi,\widehat{F}) \overset{\distr}{\to} \sup_{ \pi \in [ \pi_0, 1- \pi_0]} Q_{p}(\pi)
\end{equation}
with $Q_{p}(\pi) = (B_{p}(\pi) - \pi B_{p}(1) )'(B_{p}(\pi) - \pi B_{p}(1) )/ (\pi(1-\pi))$ and $B_{p}$ is a $p$-vector of independent Brownian motions on $[0,1]$ restricted to $[ \pi_0, 1- \pi_0] \in (0,1)$, $p = r(r+1)/2$.
The respective critical values can be found in Table 1 in \cite{andrews1993tests}.

As pointed out in Section \ref{se:roadmap}, the proof of the asymptotic behavior of our test statistic can presumably be applied to any statistic for testing changes in the loadings. That is also expected for the convergence in \eqref{eq:con:Wald_pi}.

\subsection{Other test statistics} \label{se:alternativetests}
The suggested way of applying a projection based transformation is one possibility to test for rotational changes in the loadings. We propose here some other test statistics. A simulation study comparing the different statistics is provided in Appendix \ref{se:otherstatistics_sim}.

\begin{enumerate}[label=\textit{Test \arabic*:}, ref=\arabic*, align=left]
\item Test based on the transformation \eqref{eq:transfomedseriesnotationchange1}. This is the test that has been used throughout the paper and is only mentioned here for completeness.
\item The simplest alternative is to base the analysis on the first series, i.e., we use the transformation \eqref{eq:transfomedseriesnotationchange2} instead of the \eqref{eq:transfomedseriesnotationchange1} and apply then the regular change-point test.
\item A symmetrical approach combines Tests 1 and 2 by considering the sum of the two resulting test statistics, that is,
\begin{equation*}
W(\widehat{F}_1) + W(\widehat{F}_2)
\end{equation*}
with
\begin{equation*}
W(\widehat{F}) = V(\widehat{F})' \Omega^{-1}(\widehat{F}) V(\widehat{F}).
\end{equation*}
The expected asymptotic behavior of the test statistic is 
\begin{equation*}
W(\widehat{F}_1) + W(\widehat{F}_2) \overset{\distr}{\to} \chi^{2}(r(r+1)).
\end{equation*}
\end{enumerate}
Tests 1--3 are all based on sample splits. That is, given two datasets, we (i) split one of the two series into two halves and (ii) use the other series to estimate a projection matrix that is used to transform the first dataset and to create an artificial break at the splitting point. 
Our simulation study in Appendix \ref{se:otherstatistics_sim} seems to suggest that tests based on this procedure seem to perform best when the longer series is used for step (i), i.e., the sample split and the shorter series is used for (ii), i.e., estimating the projection matrix for transformation. The simulation study in Appendix \ref{se:otherstatistics_sim} considers two datasets of different length. Test 1 uses the longer series for (i), the sample split, and Test 2 uses the shorter series. Test 1 seems to perform better than Test 2 in size and power. Test 3 performs similar to Test 1 and does not seem to have any obvious advantage.

\subsection{Discussion on alternative approaches} \label{se:Discussion}
Recall the three different types of breaks in Section \ref{se:applicationCPs}. In order to characterize the behavior of our test statistic under the alternatives, we wrote Types \ref{item:Type1} and \ref{item:Type2} in terms of so-called pseudo-factors in Section \ref{sec:alternative}. Similarly, one can rewrite what happens under Type \ref{item:Type2} (our null hypothesis). Recall from Section \ref{sec:multi_factor} that Type \ref{item:Type2} describes how the loadings undergo a rotational change.

\textit{Type \ref{item:Type2}:} Suppose $\Theta_2 = \Theta_{1}C$ with $C$ being nonsingular. Then the model can be put into the form \eqref{eq:rep_with_pseudo_factors} with the matrix $(B, C)$ in \eqref{eq:rep_with_pseudo_factors} having column rank $r$:
\begin{align*}
\begin{pmatrix}
Z^{b} \\
Z^{a}
 \end{pmatrix}
&=
 \begin{pmatrix}
F^{b} \Theta_{1}' \\
F^{a}  \Theta_{2}'
 \end{pmatrix}
 + \eta
=  
 \begin{pmatrix}
F^{b} \Theta_{1}' \\
F^{a} C' \Theta_{1}'
 \end{pmatrix}
 + \eta
  \\& =  
 \begin{pmatrix}
F^{b} \\
F^{a} C'
 \end{pmatrix}
\Theta_{1}'
 + \eta
= G \Theta' + \eta
 \end{align*}
 with $\Theta = \Theta_{1}$, $B = I_{r}$, and $G^{b} = F^{b}  $, $G^{a} = F^{a}$.
Given that the different types of breaks are characterized through the column rank of the matrix $(B,C)$, a more natural way of distinguishing seems to be a rank test on the covariances of the latent factors. This would also allow us to test whether only a partial set of vectors is shared across individuals. There are several works discussing issues of testing for the rank of a matrix; e.g., \cite{donald2010rank}.
A reliable estimator for the rank, however, would still only answer the question of how many vectors are shared across individuals but not which ones. 

In particular, a rank-type test approach can be applied directly to the data instead of using the transformed series \eqref{eq:YT2def}. The idea is to concatenate the two series $\{ X_t^1 \}$ and $\{ X_t^2 \}$ over time and to do PCA estimation based on $\npf = 2r$ pseudo-factors with $r$ being the original number of factors. Based on the PCA estimator, one can then estimate the rank of the covariances of the latent factors. The rank takes values from $r$ to $2r$. The rank equal to $r$ corresponds to a rotational change in the loadings as under our null hypothesis. With increasing rank the number of column vectors that are not shared across loadings increases, where a rank of $2r$ corresponds to no column vectors being shared.
More formally, let  $\Lambda_{k} = (\Lambda \Phi_k, \Pi_k)$, $k=1,2$ with $\Lambda \in \RR^{d \times r}$, $\Phi_k \in \RR^{r \times r_k}$, $\Pi_k \in \RR^{d \times (r-r_k)}$ and $F_{1:T}^{k} = (F_{1}^{k'}, \dots, F_{T}^{k'})'$. Then, 
 \begin{align}
 Y &=  
 \begin{pmatrix}
F_{1:T}^{1}  & 0\\
0 & F_{1:T}^{2}
\end{pmatrix}
\begin{pmatrix}
\Lambda_{1}' \\
\Lambda_{2}'
\end{pmatrix}
+
\begin{pmatrix}
\varepsilon^1 \\
\varepsilon^2
\end{pmatrix}
=  
\begin{pmatrix}
F_{1:T}^{1}  & 0\\
0 & F_{1:T}^{2}
\end{pmatrix}
\begin{pmatrix}
\Phi_{1}' \Lambda' \\
\Pi_{1}' \\
\Phi_{2}' \Lambda' \\
\Pi_{2}' \\
\end{pmatrix}
+
\begin{pmatrix}
\varepsilon^1 \\
\varepsilon^2
\end{pmatrix}
\nonumber
\\& = 
\begin{pmatrix}
F_{1:s_{1},1}^{1'}\Phi_{1}'  & F_{s_{1}+1:r_{1},1}^{1'}  & 0\\
\vdots 	& \vdots & \vdots \\
F_{1:s_{1},T}^{1'}\Phi_{1}'  & F_{s_{1}+1:r_{1},T}^{1'}  & 0\\
F_{1:s_{2},1}^{2'}\Phi_{2}' & 0  & F_{s_{2}+1:r_{2},1}^{2'} \\
\vdots 	& \vdots & \vdots \\
F_{1:s_{2},T}^{2'}\Phi_{2}' & 0  & F_{s_{2}+1:r_{2},T}^{2'}
\end{pmatrix}
\begin{pmatrix}
\Lambda' \\
\Pi_{1}' \\
\Pi_{2}' \\
\end{pmatrix}
+
\begin{pmatrix}
\varepsilon^1 \\
\varepsilon^2
\end{pmatrix}
=: \mathcal{F} \widetilde{\Lambda}' + \widetilde{\varepsilon},
\label{eq:conatenated_PCA_pseudo}
\end{align}
where we used $F^{k'}_{a:b,t} = (F^{k}_{a,t}, \dots, F^{k}_{b,t})$ in \eqref{eq:conatenated_PCA_pseudo}. Then, the matrix $\E (\mathcal{F}' \mathcal{F})$ has rank $r+r_1+r_2$.

Another natural approach is to check directly whether the projection matrices associated with $\Lambda_1$ and $\Lambda_2$ are equal. 
We refer to a recent work by \cite{liao2023changes} who studied a similar problem for continuous factor models and pursued this approach. They suggested a test statistic based on the Frobenius norm of the projection discrepancy between the pre- and post-break loadings.

The work by \cite{silin2020hypothesis} is similar in nature and also studies a certain distance between two projection matrices calculated based on the eigenvectors of two covariance matrices. The suggested test is based on finite sample properties, introducing a novel resampling method as an alternative to using bootstrap. The authors do not suggest any asymptotic limit of the proposed test statistic.

In comparison to the approaches of \cite{liao2023changes} and \cite{silin2020hypothesis}, our test is low dimensional in nature. While \cite{liao2023changes} and \cite{silin2020hypothesis} base their analyses on the $d \times d$-dimensional projection matrices, we manage to reduce the problem to testing for changes in the low-dimensional factors. In particular, this allows us to derive the asymptotic behavior of our test statistic.

Other works that use similar approaches but impose slightly different assumptions on the model are \cite{pelger2022state} and \cite{koo2023disentangling}.
\cite{pelger2022state} assumes that the loadings are general functions of some recurrent state process and develop a statistical test for a change in the loading structure in different states. Similar to our null hypothesis, they allow the loadings across different states to be the same up to a rotational change.
\cite{koo2023disentangling} assume that given some pre-break loading $\Lambda_1$, the post-break loadings can be written as $\Lambda_2 \Phi + W$, where $\Phi$ is a rotational change and $W$ is a random shift. Similar to our null hypothesis, their test aims to detect a rotational change in the loadings. However, their approach is to directly estimate the pre- and post-break factors to then apply a Wald-type test statistic.

\section{Further simulation results} \label{app:CC} 
In Appendix \ref{se:sim:diffrentnumbers} we present simulation results for the modification for factors with different dimensions as given in Appendix \ref{se:Different numbers of factors}. Appendix \ref{se:sim:alternative_approach} presents some simulation results on the supremum approach discussed in Appendix \ref{se:alternative_approach} and Appendix \ref{se:otherstatistics_sim} gives results on the tests in Appendix \ref{se:alternativetests}.
We conclude this section with a discussion on the simulation design (Appendix \ref{se:sim:discussion}). All simulations are conducted at a $5\%$ significance level and the data were standardized before applying our test.

\subsection{Different numbers of factors} \label{se:sim:diffrentnumbers}
We study here a setting allowing for different numbers of factors using the test in Appendix \ref{se:Different numbers of factors}. This includes both the two-subjects test with each being modeled by a different number of factors and having a different number of factors pre- and post-break in the change-point setting. We assume for both two subject and change-point setting that $\Lambda \in \RR^{\Dim \times r_{1}}$ with $\lambda_{ij} \sim \mathcal{N}(0,1)$, $\Phi_1 \in \RR^{r_{1} \times r_{1}}$, $\Phi_2 \in \RR^{r_{1} \times r_{2}}$ and $r_1=4$, $r_2=3$.
\begin{enumerate}[label=NDGPdnf\arabic*:, align=left]
\item Two subjects, \eqref{eq:two-subject_factor_mdoel_for_DGPs}: 
Set $\Phi_1 = I_r$ and $\Phi_2$ is randomly generated with singular values in the interval $(0.75,1.25)$.
The factors are $F_{i,t}^{k} \sim \mathcal{N}(0,1)$.
For the errors, suppose $\varepsilon_{i,t}^{k} \sim \mathcal{N}(0,1)$, $k=1,2$.
\item Two subjects, \eqref{eq:two-subject_factor_mdoel_for_DGPs}: 
Set $\Phi_1 = I_r$ and $\Phi_2$ is randomly generated with singular values in the interval $(0.75,1.25)$.
The factors for the first subject are generated as $F_{i,t}^{1} = 0.5 F_{i,t-1}^{1} + e_{i,t}$, $e_{i,t} \sim \mathcal{N}(0,1-0.7^2)$.
The factors for the second subject are $F_{i,t}^{2} \sim \mathcal{N}(0,1)$.
For the errors, suppose $\varepsilon_{i,t}^{k} \sim \mathcal{N}(0,1)$, $k=1,2$.
\item Change-point, \eqref{eq:cp_model_simulations}:
$\Phi_1, \Phi_2$ are randomly generated with singular values in the interval $(0.75,1.25)$.
The factors are generated as $F_{i,t} = 0.5 F_{i,t-1} + e_{i,t}$, $e_{i,t} \sim \mathcal{N}(0,1-0.7^2)$.
For the error, suppose $\varepsilon_{i,t} \sim \mathcal{N}(0,1)$.
\item Change-point, \eqref{eq:cp_model_simulations}:
$\Phi_1, \Phi_2$ are randomly generated with singular values in the interval $(0.75,1.25)$.
The factors are generated as  $F_{i,t} \sim \mathcal{N}(0,1)$.
For the error, suppose $\varepsilon_{i,t} \sim \mathcal{N}(0,1)$.
\end{enumerate}

\begin{table}[ht]
    \centering
    \begin{tabular}{|c||c|c|c|c|c|}
    \hline
    		&\multicolumn{2}{c|}{Two-subjects}&		\multicolumn{2}{c|}{Change-point}  \\ \hline
         $(d,T)$ & NDGPdnf1 & NDGPdnf2 & NDGPdnf3 & NDGPdnf4 \\ \hline \hline
        (200,200) & \gradient{0.042} & \gradient{0.0355} & \gradient{0.013} & \gradient{0.009}  \\ \hline
        (500,500) & \gradient{0.038} & \gradient{0.0465} & \gradient{0.029} & \gradient{0.027}     \\ \hline
    \end{tabular}
        \caption{Empirical test sizes for two individuals and change-point setting assuming different numbers of factors in pre- and post-break samples.}
    \label{tab:dnf}
\end{table}

The corresponding empirical sizes can be found in Table \ref{tab:dnf}. The test performs generally well. In particular in the setting of considering two subjects, the empirical size gets close to the nominal $5\%$ significance level. On the other hand, in the change-point setting, we see slightly under- and oversized results for NDGPdnf3 and NDGPdnf4, respectively. As mentioned in Section \ref{se:sim:change}, this could be due to the fact that the effective sample size in the change-point setting is smaller than in the two-subject model setting.
We omit the power analysis for different numbers of factors.

\subsection{Supremum approach} \label{se:sim:alternative_approach}
In Appendix \ref{se:alternative_approach}, we proposed a supremum approach, by applying a change-point test not only at the artificially created change-point $T/2$ but partitioning the data at locations over a certain interval around $T/2$ and taking the supremum over the resulting test statistics. The expected behavior of the test statistic under the null hypothesis is stated in \eqref{eq:con:Wald_pi}. 

The interval for the supremum is determined by $\pi_0 = 0.45$ in  \eqref{eq:def:Wald_pi}. In Table \ref{tab:twosubjectspoweralternative}, we consider some of the DGPs from Section \ref{se:sim:two-subjects} which suffered from low power. As one can see, the power of ADGP3 improved quite a bit. There is no improvement for ADGP2 with proportion $0.2$ of the loading matrix being different from the pre-break loadings.
\begin{table}[ht]
    \centering
    \begin{tabular}{|c||c|c||c|c|}
    \hline
         $(d,T)$ & a & ADGP2 			& $c$ &ADGP3  \\ \hline \hline
        (200,200) & 0.2 & \gradient{0.281}	& 1	& \gradient{1} \\ \hline
        (500,500) & 0.2 & \gradient{0.799}	& 1	& \gradient{1} \\ \hline \hline
        (200,200) & 0.6 & \gradient{1}	  	& 3	& \gradient{1}\\ \hline
        (500,500) & 0.6 & \gradient{1}  		& 3	& \gradient{1}	\\ \hline
    \end{tabular}
        \caption{Empirical test power using the test statistic \eqref{eq:def:Wald_pi} with $\pi_0 = 0.45$.}
    \label{tab:twosubjectspoweralternative}
\end{table}

In addition, we applied the supremum approach to some of the DGPs from Section \ref{se:sim:two-subjects} under the null hypothesis. One can observe the expected trade-off of loosing size while improving power. The size increased for all of the considered NDGPs and exceeds the nominal level of $5\%$.

\begin{table}[ht]
    \centering
    \begin{tabular}{|c||c||c|}
    \hline
         $(d,T)$ & NDGP1			& NDGP2  \\ \hline \hline
        (200,200) & \gradient{0.076}	& \gradient{0.092} \\ \hline
        (500,500) & \gradient{0.105}	& \gradient{0.145} \\ \hline
    \end{tabular}
        \caption{Empirical test size using the test statistic \eqref{eq:def:Wald_pi} with $\pi_0 = 0.45$.}
    \label{tab:twosubjectssizealternative}
\end{table}

\subsection{Other statistics} \label{se:otherstatistics_sim}
We provide here a small simulation study to compare the different test statistics formulated in Appendix \ref{se:alternativetests}. 
Tables \ref{tab:twosubjectssize_othertests} and \ref{tab:twosubjectspower_othertests} show respectively the size and power of the three different tests. We used here synthetic data based on two-subject models and picked different sample sizes for the two datasets. According to the results, the first test performs well and similarly to the symmetric one (Test 3). Test 1 is based on transforming the second series, i.e., the one that is chosen to be longer. This suggests that our test is more sensitive to applying the change-point test to a smaller sample rather than estimating a projection matrix based on a smaller sample.

\begin{table}[ht] 
    \centering
    \begin{tabular}{|c||c|c|c|c|c|}
    \hline
         $(d, T_1, T_2)$ & Test 1 & Test 2 & Test 3 \\ \hline \hline
        (150,100,200) & \gradient{0.0275} & \gradient{0.0095} & \gradient{0.0245}   \\ \hline
        (450,400,500) & \gradient{0.0435} & \gradient{0.0290} & \gradient{0.0360}    \\ \hline
    \end{tabular}
        \caption{Empirical test size for different test statistics applied to NDGP1.}
    \label{tab:twosubjectssize_othertests}
\end{table}

\begin{table}[ht] 
    \centering
    \begin{tabular}{|c||c|c|c|c|c|}
    \hline
         $(d, T_1, T_2)$ & Test 1 & Test 2 & Test 3 \\ \hline \hline
        (150,100,200) & \gradient{0.8435} & \gradient{0.3705} & \gradient{0.883}   \\ \hline
        (450,400,500) & \gradient{1} & \gradient{0.993} & \gradient{1}    \\ \hline
    \end{tabular}
        \caption{Empirical test power for different test statistics applied to ADGP1 with $b=0.5$.}
    \label{tab:twosubjectspower_othertests}
\end{table}

\subsection{Discussion on simulation design} \label{se:sim:discussion}
The theoretical results on consistent estimation of the loadings rely (among others) on Assumption \ref{ass:C1}. The assumptions in there ensure that the rescaled singular values of the loading matrices are nonzero and finite, i.e., $\liminf_{d\to\infty} \frac{\lambda_{\min}(\Lambda'\Lambda) }{d} > 0$ and $\limsup_{d\to\infty}  \frac{\lambda_{\max}(\Lambda'\Lambda)}{d}<\infty$.

Our simulation study includes the assumptions that the matrices $\Phi_k$, $k=1,2$, take only eigenvalues in a certain interval. This interval can certainly be bigger than the one we considered. However, the smallest eigenvalue $\lambda_{\min}(\Phi \Lambda'\Lambda \Phi)$ can not get too close to zero. This should not be surprising since we would move closer to the regime of having a different number of factors. On the population level, a different number of factors results in different loading matrices since their dimensions mismatch. Then, our test statistic is expected to reject the null hypothesis of having the same (or the same up to transformation with a nonsingular matrix) loading matrices across time or individuals.

To illustrate these points, we consider some extreme cases with particularly small and large eigenvalues of $\Phi_{k}$.
We focus here on the two subject case.

\begin{enumerate}[label=DGP:, align=left]
\item $\Lambda_{1} = \Lambda$ and $\Lambda_{2} = \Lambda \Phi$ with $\lambda_{ij} \sim \mathcal{N}(0,1)$.
$\Phi$ is randomly generated but with fixed smallest and largest eigenvalues $\lambda_{\min}(\Phi)=e$ and $\lambda_{\max}(\Phi)=f$.
\end{enumerate}

The DGP satisfies the null hypothesis algebraically but violates the regularity condition (Assumption \ref{ass:C1}/pervasiveness) because the smallest eigenvalue of $\Lambda'_2\Lambda_2$ becomes too small; hence factor estimation is ill-posed and the test over-rejects as one can see from Table \ref{tab:discussion}.

\begin{table}[ht]
    \centering
    \begin{tabular}{|c||c|c|c|c|c|}
    \hline
         $(d,T)$& $e$	& $f$ & $\lambda_{\min}(\Phi \Lambda'\Lambda \Phi)$ & $\lambda_{\max}(\Phi \Lambda'\Lambda \Phi)$ & DGP   \\ \hline \hline
        (200,200)	 & 0.01 	& 1 &  0.140 & 14.117 & \gradient{1}  \\ \hline
        (500,500) 	 & 0.01	& 1 &  0.223 & 22.337 & \gradient{1}  \\ \hline \hline
        (200,200)	 & 0.5 	& 1 &  7.01 & 11.15 & \gradient{0.0275}  \\ \hline
        (500,500) 	 & 0.5	& 1 &  14.18 & 22.40 & \gradient{0.044}  \\ \hline \hline
        (200,200)	 & 0.5 	& 100 &  7.03 & 1411.81 & \gradient{1}  \\ \hline
        (500,500) 	 & 0.5	& 100 &  11.15 & 2234.62 & \gradient{1}  \\ \hline 
    \end{tabular}
        \caption{Rejection rate under different assumptions on the eigenvalues of the rotation matrix.}
    \label{tab:discussion}
\end{table}

\section{Derivations under the alternative (details for Section~\ref{sec:alternative})}\label{app:alternative_derivations}
In this section, we provide the details on the alternative hypothesis. In particular, we characterize different scenarios that may occur for the transformed series \eqref{eq:YT2def}.

 
The possible alternatives are of Types \ref{item:Type1} and \ref{item:Type3} as described in Section \ref{se:applicationCPs}. As in Section \ref{se:motivation}, we base our analysis on $\{ X_{t}^2 \}$, using the transformed series $Y_{1:T_2}$ as introduced in \eqref{eq:YT2def}.
For simplicity, we assume that $T=T_{1}=T_{2}$. Then,
\begin{equation} \label{eq:with_change}
Y_{t} =
\begin{cases}
\widehat{P}_{1} \Lambda_2 F_{t} + \widehat{P}_{1} \varepsilon_{t}
=:
\widetilde{\Theta}_1 F_{t} + \widetilde{\varepsilon}_{t}, \hspace{0.2cm} \text{ for } t =1,\dots, T/2,
\\
\widehat{P}_{2} \Lambda_2 F_{t} + \widehat{P}_{2} \varepsilon_{t}
=:
\widetilde{\Theta}_2 F_{t} + \widetilde{\varepsilon}_{t}, \hspace{0.2cm} \text{ for } t =T/2+1,\dots, T.
\end{cases}
\end{equation}
Recall also the behavior of $Y_{1:T}$ under the alternative at the population level from \eqref{eq:motivation_alternative_population}. Whenever we use the transformed series \eqref{eq:with_change} for estimation, we keep in mind a population model of the form
\begin{equation} \label{eq:with_change_popu}
Z_{t} =
\begin{cases}
P_1 \Lambda_{2} F_{t} + \eta_{t} 
&=:
\Theta_1 F_{t} + \eta_{t}, \hspace{0.2cm} \text{ for } t =1,\dots, T/2,
\\
\Lambda_2 F_{t} + \eta_{t}
&=:
\Theta_2 F_{t} + \eta_{t}, \hspace{0.2cm} \text{ for } t =T/2+1,\dots, T,
\end{cases}
\end{equation}
where $P_1,P_2$ are defined in \eqref{eq:true_projections_individual_series} as the population counterparts of \eqref{eq:def_Phat12} such that the loadings behave according to \eqref{eq:motivation_alternative_population} and with some error series $\eta_{t}$. The pre- and post-break covariances of $\eta_{t}$ may be different. 
We use \eqref{eq:with_change_popu} to illustrate the behavior of the estimated loadings under the alternative when \eqref{eq:with_change} is used for estimation.
As we discuss next, pre- and post-break loadings in \eqref{eq:with_change_popu} admit different types of breaks depending on the behavior of $\Lambda_1$ and $\Lambda_2$.

Recall that $\Lambda_1 = \Lambda_2 \Phi$ with nonsingular $\Phi$ under the null hypothesis or, equivalently, $\col(\Lambda_1) = \col(\Lambda_2)$. In the subsequent analysis, we aim to shed light on how the relationship of $\Lambda_1$ and $\Lambda_2$ under the alternative gets translated to $P_1 \Lambda_2$ and $\Lambda_2$ in \eqref{eq:with_change_popu}.

First, note that our null hypothesis of having a nonsingular matrix $\Phi$ such that $\Lambda_1 = \Lambda_2 \Phi$, is equivalent to $P_1 \Lambda_2 = \Lambda_2$ and $\rank( \Lambda_1' \Lambda_2)=r$. 
The null hypothesis implies the first relation since $P_{1} = P_{2}$, and the second relation since $\rank( \Lambda_1' \Lambda_2) = \rank( \Phi' \Lambda_2' \Lambda_2)$ and both $\Phi$ and $\Lambda_2' \Lambda_2$ have rank $r$. For the converse, we have
\begin{equation*}
P_1 \Lambda_2 = \Lambda_1 \left( (\Lambda_1'\Lambda_1)^{-1} \Lambda_1'\Lambda_2 \right) = \Lambda_2
\end{equation*}
and therefore $\Lambda_1 \Psi = \Lambda_2$ with nonsingular $\Psi = (\Lambda_1'\Lambda_1)^{-1} \Lambda_1'\Lambda_2$.

Then, the alternative hypothesis, i.e., $\Lambda_1 \neq \Lambda_2 \Phi$ for all nonsingular $\Phi$, is equivalent to the condition that $P_1 \Lambda_2 \neq \Lambda_2$ or $\rank( \Lambda_1' \Lambda_2)<r$. 
In fact, we will argue that \textit{(i)} 
\begin{equation} \label{eq:maineqcase1}
P_1 \Lambda_2 = \Lambda_1 \left( (\Lambda_1'\Lambda_1)^{-1} \Lambda_1'\Lambda_2 \right)  \neq \Lambda_2
\end{equation}
and discuss whether \textit{(ii)} $\col( P_1 \Lambda_2) \cap \col(\Lambda_2) = \{0\}$ and \textit{(iii)} $\col( P_1 \Lambda_2) = \col(\Lambda_2)$.
The different behaviors of the respective column spaces of $P_1 \Lambda_2$ and $\Lambda_2$ are then associated with different types of changes in the model \eqref{eq:with_change_popu}. 
That is, if \textit{(i)} and \textit{(ii)} hold, then the change in \eqref{eq:with_change_popu} is of Type \ref{item:Type1}. 
If \textit{(i)} and \textit{(iii)} hold, it is of Type \ref{item:Type2}.
If \textit{(i)} holds and \textit{(ii)} and \textit{(iii)} do not hold,  the change is of Type \ref{item:Type3}.

To investigate \textit{(i)}--\textit{(iii)}, we distinguish between Cases 1 and 2 below depending on whether $\col(\Lambda_1) \cap \col(\Lambda_2) = \{0\}$ (Type \ref{item:Type1}) or 
$\col(\Lambda_1) \cap \col(\Lambda_2) \neq \{0\}$ and  $\col(\Lambda_1) \neq \col(\Lambda_2)$ (Type \ref{item:Type3}) in the original model \eqref{eq:cp_model}.
\vspace{0.2cm}\\
\noindent
\textit{Case 1:} Suppose $\col(\Lambda_1) \cap \col(\Lambda_2) = \{0\}$.
We consider the following subcases: $\rank( \Lambda_1' \Lambda_2)=r$, $=0$ and $=s (\neq 0,r)$.
\begin{enumerate}[label=\textit{1\alph*}.,ref=\textit{1\alph*}]
\item \label{item1a}
$\rank( \Lambda_1' \Lambda_2)=r$: Then,
\begin{enumerate}[label=\textit{(\roman*)}]
\item 
$P_1 \Lambda_2 \neq \Lambda_2$: Arguing by contradiction, assume equality in \eqref{eq:maineqcase1} so that $P_1 \Lambda_2 = \Lambda_2$. This yields 
$\Lambda_1 = \Lambda_2 (\Lambda_1'\Lambda_2)^{-1} \Lambda_1'\Lambda_1$ which contradicts $\col(\Lambda_1) \cap \col(\Lambda_2) = \{0\}$.
\item 
$\col( P_1 \Lambda_2) \cap \col(\Lambda_2) = \{0\}$: By contradiction, suppose there are nonzero vectors $v_1,v_2 \in \RR^{r}$ such that 
$P_1 \Lambda_2 v_1 = \Lambda_2 v_2$. Then, $\Lambda_1 \widetilde{v}_1 = \Lambda_2 v_2$ with $ \widetilde{v}_1 = \left( (\Lambda_1'\Lambda_1)^{-1} \Lambda_1'\Lambda_2 \right)v_1$ which, since $\Lambda_1 \widetilde{v}_1 \neq 0$ because $\Lambda_1$ has full column rank, contradicts $\col(\Lambda_1) \cap \col(\Lambda_2) = \{0\}$.
\end{enumerate}
This corresponds to a change of Type \ref{item:Type1} for the model \eqref{eq:with_change_popu}.
\item \label{item1b}
$\rank( \Lambda_1' \Lambda_2)=0$: Then $\Lambda_1'\Lambda_2 = 0$ and \eqref{eq:maineqcase1} becomes \textit{(i)} $P_1 \Lambda_2 = 0 \neq \Lambda_2$. We can then infer that \textit{(ii)} $\col( P_1 \Lambda_2) \cap \col(\Lambda_2) = \{ 0 \}$.
\vspace{0.2cm}\\
This corresponds to a change of Type \ref{item:Type1} for the model \eqref{eq:with_change_popu}.
\item \label{item1c}
$\rank( \Lambda_1' \Lambda_2) = s \neq 0,r$: Then $\Lambda_1'\Lambda_2 = \alpha \beta'$ with $\alpha, \beta \in \RR^{r \times s}$ and $\rank( \alpha) = \rank( \beta) = \rank( \Lambda_1' \Lambda_2)$. 
\begin{enumerate}[label=\textit{(\roman*)}]
\item $P_1 \Lambda_2 \neq \Lambda_2$: Arguing by contradiction, equality in \eqref{eq:maineqcase1} gives $\Lambda_1 \widetilde{\alpha} = \Lambda_2 \beta ( \beta' \beta)^{-1}$ with $\widetilde{\alpha} = (\Lambda_1'\Lambda_1)^{-1} \alpha$, contradicting $\col(\Lambda_1) \cap \col(\Lambda_2) = \{0\}$. 
\item $\col( P_1 \Lambda_2) \cap \col(\Lambda_2) \neq \{0\}$: We show that there are nonzero vectors $v_1,v_2 \in \RR^{r}$ such that $P_1 \Lambda_2 v_1 = \Lambda_2 v_2$. Choose $v_1 = 2v_2$. Then, $P_1 \Lambda_2v_1 = \Lambda_1 \widetilde{\alpha} \beta' v_2 + \Lambda_2v_2 = \Lambda_2v_2$ whenever $\beta' v_2 = 0$. We can find a nonzero vector $v_2$ such that $\beta' v_2 = 0$ since $\beta$ is $r \times s$, $\beta'$ is $s\times r$ and $s <r$.
\item $\col( P_1 \Lambda_2) \neq \col(\Lambda_2)$: By contradiction, suppose there is a nonsingular matrix $\Phi$ such that $P_1 \Lambda_2 = \Lambda_2 \Phi$. Then, $\Lambda_1 \widetilde{\alpha} \beta' = \Lambda_2 \Phi $ and, choosing a vector $v \in \RR^r$ with  $ \widetilde{v} := \widetilde{\alpha} \beta' v \neq 0$, we get
$\Lambda_1 \widetilde{v} = \Lambda_2 \Phi v$, which contradicts $\col(\Lambda_1) \cap \col(\Lambda_2) = \{0\}$ or $\Lambda_1$ having full column rank.
\end{enumerate}
This corresponds to the change of Type \ref{item:Type3} for the model \eqref{eq:with_change_popu}.
\end{enumerate}

\noindent
\textit{Case 2:} Suppose $\col(\Lambda_1) \cap \col(\Lambda_2) \neq \{0\}$ and $\col(\Lambda_1) \neq \col(\Lambda_2)$. We consider the following subcases:
\begin{enumerate}[label=\textit{2\alph*}.,ref=\textit{2\alph*}]
\item \label{item2a}
$\rank( \Lambda_1' \Lambda_2)=r$: Then, 
\begin{enumerate}[label=\textit{(\roman*)}]
\item $P_1 \Lambda_2 \neq \Lambda_2$: Arguing by contradiction, assume equality in \eqref{eq:maineqcase1} such that $P_1 \Lambda_2 = \Lambda_2$. This yields 
$\Lambda_1 = \Lambda_2 (\Lambda_1'\Lambda_2)^{-1} \Lambda_1'\Lambda_1$ which contradicts $\col(\Lambda_1) \neq \col(\Lambda_2)$.
\item $\col( P_1 \Lambda_2) \cap \col(\Lambda_2) \neq \{0\}$: 
Since $\col(\Lambda_1) \cap \col(\Lambda_2) \neq \{0\}$, there are nonzero vectors $v_1, v_2 \in \RR^r$ such that $\Lambda_1v_1 = \Lambda_2 v_2$. 
Choose $w_1 = \Psi^{-1} v_1$ and $w_2 = v_2$ with $\Psi = (\Lambda_1'\Lambda_1)^{-1} \Lambda_1'\Lambda_2$. Then, $P_1 \Lambda_2w_1 = \Lambda_1 \Psi w_1 = \Lambda_1 v_1 = \Lambda_2 v_2  = \Lambda_2 w_2$. Therefore, we found vectors $w_1, w_2 \in \RR^r$ such that $P_1 \Lambda_2 w_1 = \Lambda_2 w_2$.
\item $\col( P_1 \Lambda_2) \neq \col(\Lambda_2)$: By contradiction, suppose there is a nonsingular matrix $\Phi$ such that $P_1 \Lambda_2 = \Lambda_2 \Phi$. Then, $\Lambda_1 \Psi = \Lambda_2 \Phi$ with $\Psi = (\Lambda_1'\Lambda_1)^{-1} \Lambda_1'\Lambda_2$, which contradicts $\col(\Lambda_1) \neq \col(\Lambda_2)$.
\end{enumerate}
This corresponds to a change of Type \ref{item:Type3} for the model \eqref{eq:with_change_popu}.
\item
$\rank( \Lambda_1' \Lambda_2)=0$: Then, $\Lambda_1'\Lambda_2 = 0$. This implies that all column vectors of $\Lambda_1$ and $\Lambda_2$ are orthogonal which contradicts $\col( \Lambda_1) \cap \col(\Lambda_2) \neq \{0\}$. That is, this subcase cannot occur.
\item 
$\rank( \Lambda_1' \Lambda_2) = s \neq 0,r$: Then, 
\begin{enumerate}[label=\textit{(\roman*)}]
\item $P_1 \Lambda_2 \neq \Lambda_2$: We follow the arguments in \ref{item1c}\textit{(i)}. Arguing by contradiction, equality in \eqref{eq:maineqcase1} gives $\Lambda_1 \widetilde{\alpha} = \Lambda_2 \beta ( \beta' \beta)^{-1}$ with $\widetilde{\alpha} = (\Lambda_1'\Lambda_1)^{-1} \alpha$, contradicting $\col(\Lambda_1) \neq \col(\Lambda_2)$. 
\item $\col( P_1 \Lambda_2) \cap \col(\Lambda_2) \neq \{0\}$: This follows exactly as in \ref{item1c}\textit{(ii)}.
\item $\col( P_1 \Lambda_2) \neq \col(\Lambda_2)$: Arguing as in \ref{item1c}\textit{(iii)}, we get
$\Lambda_1 \widetilde{v} = \Lambda_2 \Phi v$. This now contradicts $\col(\Lambda_1) \neq \col(\Lambda_2)$.
\end{enumerate}
This corresponds to a change of Type \ref{item:Type3} for the model \eqref{eq:with_change_popu}.
\end{enumerate}


In summary, depending on the type of change  and the relationship between $\Lambda_1$ and $\Lambda_2$ in the original model \eqref{eq:cp_model}, we can get two of the three different types of breaks described in Section \ref{se:applicationCPs} for the model \eqref{eq:with_change_popu}.

For the remainder of the paper and for the shortness sake, we focus on two subcases above: Case \ref{item1a}, which results in a Type \ref{item:Type1} break and Case \ref{item2a} assuming that $\Lambda_{2} = (\Lambda_1 \Phi , \Pi_2)$ with $\Phi \in \RR^{r \times r_1}$, and $\Pi_2 \in \RR^{d \times r_2}$ with $r_2 = r-r_1$ being linearly independent to $\Lambda_1$ which results in a Type \ref{item:Type3} break.

For the proofs of our theoretical results under the alternative, one needs to distinguish each subcase, and we will use Cases \ref{item1a} and \ref{item2a} to illustrate the main ideas. The remaining cases can be handled similarly. From here on, we will refer to those cases as \textit{TA 1} and \textit{TA 2} but keep in mind how those types result from the original model \eqref{eq:cp_model} according to Cases \ref{item1a} or \ref{item2a}.

We aim to rewrite \eqref{eq:with_change_popu} in terms of so-called pseudo-factors. Consider the following matrix representations: 
 \begin{equation} \label{eq:all_def_alternative}
\begin{gathered}
 Z^{b} = (Z_{1}, \dots, Z_{T/2})',
 \hspace{0.2cm}
 Z^{a} = (Z_{T/2+1}, \dots, Z_{T})';
 \\
 F^{b} = (F_{1}, \dots, F_{T/2})',
 \hspace{0.2cm}
 F^{a} = (F_{T/2+1}, \dots, F_{T})';
\\
 G^{b} = (G_{1}, \dots, G_{T/2})',
 \hspace{0.2cm}
 G^{a} = (G_{T/2+1}, \dots, G_{T})';
 \\
 \eta^{b} = ( \eta_{1}, \dots,  \eta_{T/2})',
 \hspace{0.2cm}
 \eta^{a} = ( \eta_{T/2+1}, \dots,  \eta_{T})',
 \end{gathered}
 \end{equation}
where $G^b$ and $G^a$ will refer to pseudo-factors and both have dimensions $T/2 \times \npf$, and $\npf$ is the number of pseudo-factors. The pseudo-factors are chosen such that the pre- and post-break loadings can be modeled as $\Theta_1 = \Theta B$ and $\Theta_2 = \Theta C$ with \textit{square} matrices $B,C \in \RR^{\npf \times \npf}$ and some $\Theta \in \RR^{d \times \npf}$ with full column rank. 
Then, \eqref{eq:with_change_popu} can be written as
\begin{align} \label{eq:rep_with_pseudo_factors}
\begin{pmatrix}
Z^{b} \\
Z^{a}
 \end{pmatrix}
&= 
 \begin{pmatrix}
F^{b} \Theta_{1}' \\
F^{a} \Theta_{2}'
 \end{pmatrix}
 +
\begin{pmatrix}
\eta^{b} \\
\eta^{a}
 \end{pmatrix}
=  
 \begin{pmatrix}
G^{b} B' \\
G^{a} C'
 \end{pmatrix}
\Theta'
  +
\begin{pmatrix}
\eta^{b} \\
\eta^{a}
 \end{pmatrix}
=: 
G \Theta' + \eta.
 \end{align}

In this model, $r = \rank(B) \leq \npf$ and $r = \rank(C) \leq \npf$ denote the numbers of original factors before and after the break, respectively. Note that \eqref{eq:rep_with_pseudo_factors} provides an observationally equivalent representation without a change in the loading matrix $\Theta \in \RR^{d \times \npf}$.

We consider now the alternatives of Types \ref{item:Type1} and \ref{item:Type3} resulting respectively from \textit{TA 1} (Case \ref{item1a}) and \textit{TA 2} (Case \ref{item2a}) separately to write them in terms of pseudo-factors as in \eqref{eq:rep_with_pseudo_factors}. For that, we need to choose $G$ and $\Theta$ in \eqref{eq:rep_with_pseudo_factors} appropriately. Instead of expressing Types \ref{item:Type1} and \ref{item:Type3} through linear independence of the loadings as described in Section \ref{se:applicationCPs}, we can express the different types in terms of the column rank of the matrix $(B, C)$.

\textit{TA 1:} 
We suppose $\col(\Lambda_1) \cap \col(\Lambda_2) = \{0\}$ and $\rank( \Lambda_1' \Lambda_2)=r$. Then, according to Case \ref{item1a}, the alternative model admits a change of Type \ref{item:Type1} since $\col( P_1 \Lambda_2) \cap \col(\Lambda_2) = \{0\}$. 
Let $F^{b} \in \RR^{T/2 \times r}$ and $F^{a} \in \RR^{T/2 \times r}$ denote the original factors as in \eqref{eq:all_def_alternative}. Then, the model \eqref{eq:with_change_popu} can be written as
\begin{equation}\label{eq:type1BC}
\begin{aligned}
\begin{pmatrix}
Z^{b} \\
Z^{a}
\end{pmatrix}
&=
\begin{pmatrix}
F^{b} \Theta_{1}' \\
F^{a}  \Theta_{2}'
\end{pmatrix}
+
\eta
=  
\begin{pmatrix}
F^{b} & 0_{T/2 \times r}\\
0_{T/2 \times r} & F^{a}
\end{pmatrix}
\begin{pmatrix}
(P_1 \Lambda_{2})' \\
\Lambda_{2}'
\end{pmatrix}
+ 
\eta
\\& =  
\begin{pmatrix}
\mleft[ F^{b} : \star \mright] B' 
\vspace{0.1cm}
\\
\mleft[ \star : F^{a} \mright] C'
 \end{pmatrix}
\Theta'
+
\eta
= G \Theta' + \eta
\end{aligned}
\end{equation}
with $\Theta = \mleft( \Theta_{1}, \Theta_{2} \mright) = \frac{1}{\sqrt{2}} ( P_1 \Lambda_2, \Lambda_2 )$, $B = \diag( \sqrt{2} I_{r}, 0_{r \times r})$, $C = \diag( 0_{r \times r}, \sqrt{2} I_{r} )$ and $G^{b} = \mleft[ F^{b} : \star \mright] $, $G^{a} = \mleft[ \star : F^{a} \mright]$ with $\star$ denoting some unidentified numbers. The matrix $(B, C)$ has column rank $2r$.
The number of original factors is strictly less than that of the pseudo-factors.

\textit{TA 2:} 
We suppose $\col(\Lambda_1) \cap \col(\Lambda_2) \neq \{0\}$, $\col(\Lambda_1) \neq \col(\Lambda_2)$ and $\rank( \Lambda_1' \Lambda_2)=r$. Then, according to Case \ref{item2a}, the alternative model admits a change of Type \ref{item:Type3} since $\col( P_1 \Lambda_2) \cap \col(\Lambda_2) \neq \{0\}$ and $\col( P_1 \Lambda_2) \neq \col(\Lambda_2)$.
In order to write the model in terms of pseudo-factors, we need to impose some more structural assumptions on the loadings. Suppose $\Lambda_{2} = (\Lambda_1 \Phi , \Pi_2)$ with $\Phi \in \RR^{r \times r_1}$, and $\Pi_2 \in \RR^{d \times r_2}$ with $r_2 = r-r_1$ being linearly independent of $\Lambda_1$. Then, 
$P_1 \Lambda_{2} = (\Lambda_1 \Phi , P_1 \Pi_2)$, $P_2 \Lambda_{2} = \Lambda_2$, and some (but not all) columns of $\Theta_1 = P_1 \Lambda_{2}$ and $\Theta_2 = \Lambda_{2}$ are linearly independent. 
In particular, we have that $\col(P_1 \Pi_2) \cap \col(\Pi_2) = \{0\}$. By contradiction, suppose that there are nonzero vectors $v_1, v_2 \in \RR^s$ such that $P_1 \Pi_2 v_1 = \Pi_2 v_2$. Then, $\Lambda_1 (\Lambda_1'\Lambda_1)^{-1} \Lambda_1' \Pi_2 v_1 = \Pi_2 v_2$ which contradicts 
$\col(\Lambda_1) \cap \col(\Pi_2) = \{0\}$ since $\rank(\Lambda_1' \Pi_2) = \rank(\Pi_2) = s$.

Then, 
\begin{equation}\label{eq:type2BC}
\begin{aligned}
\begin{pmatrix}
Z^{b} \\
Z^{a}
\end{pmatrix}
&=
\begin{pmatrix}
F^{b} \Theta_{1}' \\
F^{a} \Theta_{2}'
\end{pmatrix}
+ \eta
= 
\begin{pmatrix}
F^{b} (\Lambda_1 \Phi , P_1 \Pi_2)' \\
F^{a} \Lambda_{2}'
\end{pmatrix}
+ \eta
\\&=  
\begin{pmatrix}
\mleft[ F^{b} : \star \mright] B' (\Lambda_1, (1/\sqrt{2}) \Pi_2 , (1/\sqrt{2}) P_1 \Pi_2)' \\
\mleft[ F^{a} : \star \mright]  C'(\Lambda_1 , (1/\sqrt{2}) \Pi_2, (1/\sqrt{2}) P_1 \Pi_2)'
\end{pmatrix}
+ \eta
\\&=  
\begin{pmatrix}
\mleft[ F^{b} : \star \mright] B'\\
\mleft[ F^{a} : \star \mright] C'
\end{pmatrix}
\begin{pmatrix} 
\Lambda_{1}' \\
(1/\sqrt{2}) \Pi_2' \\
(1/\sqrt{2}) \Pi_2' P_{1}'
\end{pmatrix}
+ \eta
= G \Theta' + \eta
 \end{aligned}
 \end{equation}
with $\Theta =  ( \Lambda_1, (1/\sqrt{2}) \Pi_2, (1/\sqrt{2}) P_1 \Pi_2 )$, 
\begin{equation*}
B' = 
\begin{pmatrix}
\Phi'  & 0_{r_1 \times r_2} & 0_{r_1 \times r_2} \\
0_{r_2 \times r}	& 0_{r_2 \times r_2} & \sqrt{2} I_{r_{2} \times r_2} \\
0_{r_2 \times r}	& 0_{r_2 \times r_2} & 0_{r_{2} \times r_2} 
\end{pmatrix}_{(r+r_2)\times (r+r_2)},
\hspace{0.2cm}
C' = 
\begin{pmatrix}
\Phi'  & 0_{r_1 \times r_2} & 0_{r_1 \times r_2} \\
0_{r_2 \times r}	& \sqrt{2} I_{r_{2} \times r_{2}} & 0_{r_2 \times r_2} \\
0_{r_2 \times r}	& 0_{r_2 \times r_2} & 0_{r_{2} \times r_2} 
\end{pmatrix}_{(r+r_2)\times (r+r_2)}
\end{equation*}
and $G^{b} = \mleft[ F^{b} : \star \mright]$, $G^{a} = \mleft[ F^{a} : \star \mright] $. The matrix $(B, C)$ has column rank between $r$ and $2r$, namely, $r + r_2$. 

For both \textit{TA 1} and \textit{TA 2}, the normalization with $\sqrt{2}$ is out of convenience because it allows us to express the covariance matrix of $G$ in terms of $\Sigma_F$ (the covariance matrix of $F$) without a multiplicative constant; see for example \eqref{eq:type1_pop_cov} below.

Following Section \ref{sec3.2}, we define the PCA estimator $\widehat{G}$ as in \eqref{eq:PCAestimates_individual_G} using the transformed series $Y := Y'_{1:T}$ in \eqref{eq:with_change}. Under the alternative hypothesis, however, $\widehat{G}$ is an estimator of factors in the equivalent model \eqref{eq:rep_with_pseudo_factors}. In particular, we work with $\npf$ factors. In practice, the number of pseudo-factors is estimated, for example, through the information criterion introduced in 
\cite{bai2002determining}. Their information criterion estimates the number of factors consistently under the alternative model; see Proposition 1 in \cite{han2015tests} and the preceding paragraph for a discussion.

Note that in contrast to \cite{han2015tests}, we choose the notation $\widehat{G}$ for the estimators of the pseudo-factors to emphasize that they have a different dimension than the estimators under the null hypothesis.

We also define an analogous quantity to $H$ in \eqref{eq:HconvH0} as the matrix 
\begin{equation} \label{eq:Jdef}
J = (\Theta' \Theta /d) (G' \widehat{G} /T) \widehat{V}^{-1}_{\npf}
\end{equation}
with $\widehat{V}_{\npf}$ being the $\npf \times \npf$ diagonal matrix of the $\npf$ largest eigenvalues of $\frac{1}{Td} \widehat{\Sigma}_T$. Then, as discussed in Appendix \ref{se:proofs_auxiliary_results_A}, one can show that there is a matrix $J_0$ such that $J \overset{\prob}{\to} J_0$.

Next we calculate for both \textit{TA 1} and \textit{TA 2} the pre- and post-break covariances. 

\textit{TA 1:} Based on the representation \eqref{eq:type1BC}, we get
\begin{equation} \label{eq:type1_pop_cov}
\begin{aligned}
D_1 &:= \frac{1}{T}\E (BG^{b'}G^bB') = \frac{1}{T}B \E \mleft( \Big[ F^{b} : \star \Big]' \Big[ F^{b} : \star \Big] \mright) B' = 
\begin{pmatrix}
\Sigma_{F} & 0_{r \times r} \\
0_{r \times r} & 0_{r \times r}
\end{pmatrix},
\\
D_2 &:= \frac{1}{T}\E (C G^{a'}G^a C') = \frac{1}{T}C \E \mleft( \Big[ \star : F^{a} \Big]' \Big[ \star : F^{a} \Big] \mright) C' = 
\begin{pmatrix}
0_{r \times r} & 0_{r \times r} \\
0_{r \times r} &  \Sigma_{F}
\end{pmatrix}.
\end{aligned}
 \end{equation}
 
\textit{TA 2:} Based on the representation \eqref{eq:type2BC}, we get
\begin{equation} \label{eq:type2_pop_cov}
\begin{aligned}
 D_1 &:= \frac{1}{T}\E (BG^{b'}G^b B') 
 = \frac{1}{T}B \E \mleft( \Big[ F^{b} : \star \Big]' \Big[ F^{b} : \star \Big] \mright) B' = 
B \begin{pmatrix}
\frac{1}{2} \Sigma_{F} & \star_{r \times r_2} \\
\star_{r_2 \times r} &  \star_{r_2 \times r_2}
\end{pmatrix} B',
\\
 D_2 &:= \frac{1}{T}\E (C G^{a'}G^a C') 
 = \frac{1}{T}C \E \mleft( \Big[ F^{a} : \star \Big]' \Big[ F^{a} : \star \Big] \mright) C' = 
C \begin{pmatrix}
\frac{1}{2} \Sigma_{F} & \star_{r \times r_2} \\
\star_{r_2 \times r} &  \star_{r_2 \times r_2}
\end{pmatrix} C'.
\end{aligned}
 \end{equation}

For both types, we define the matrix 
\begin{equation} \label{defJDDJ}
\mathcal{C} := J_0'(D_1 - D_2)J_0. 
\end{equation}
Here, $J_0$ and $D_{1},D_{2}$ can be different depending on the respective type of change as given in \eqref{eq:type1_pop_cov} and \eqref{eq:type2_pop_cov}.

\section{Proofs of auxiliary results under null hypothesis} \label{se:appB}

The proofs of \eqref{eq:key1}--\eqref{eq:key3} use existing results for the estimation of the model parameters of the individual series as given in \cite{doz2011two}. For the reader's convenience, we collect the results of \cite{doz2011two} as needed here.

Set $\widehat{\Sigma}_{k} = \frac{1}{T} \sum_{t=1}^{T} X^k_{t} X^{k'}_{t}$, $k=1,2$, and write $\widebar{Q}_{r}^{k}$ for the eigenvectors corresponding to the $r$ largest eigenvalues of $\widehat{\Sigma}_{k}$. 
We further consider the eigendecomposition 
\begin{equation*}
\Lambda'_k \Lambda_k = R^k \Pi_k R^{k'},
\end{equation*}
where the diagonal matrix $\Pi_k$ consists of the eigenvalues in decreasing order and the orthogonal matrix $R^k$ consists of the corresponding eigenvectors. Choose the $d \times r$ matrices 
\begin{equation} \label{eq:def_Q}
Q^k = \Lambda_k R^k \Pi_k^{-\frac{1}{2}} 
\end{equation}
such that $\Lambda_k \Lambda'_k = Q^k \Pi_k Q^{k'}$ and $Q^{k'}Q^k = I_{r}$, 
$k=1,2$. 
Then, under Assumptions \ref{ass:1}--\ref{ass:2}, \ref{ass:C1}--\ref{ass:C4}, we get
\begin{enumerate}[label=\textbf{DGR.\arabic*.},ref=DGR.\arabic*, align=left]
\item \label{Doz1} By Lemma 2(i) in \cite{doz2011two}, $ \frac{1}{d} \| \widehat{\Sigma}_{k} - \Lambda_{k} \Lambda_{k}' \| = \mathcal{O}\mleft( \frac{1}{d} \mright) + \mathcal{O}_{\prob} \mleft( \frac{1}{\sqrt{T}} \mright)$ for $k=1,2$.
\item \label{Doz3} By Lemma 4(i) in \cite{doz2011two}, $\widebar{Q}_{r}^{k'} Q^{k} - I_r = \mathcal{O}_{\prob}\mleft( \frac{1}{d} \mright) + \mathcal{O}_{\prob} \mleft( \frac{1}{\sqrt{T}} \mright)$ for $k=1,2$.
\item \label{Doz2} By Lemma 4(ii) in \cite{doz2011two}, $\| \widebar{Q}_{r}^{k} - Q^{k} \|^2 = \mathcal{O}_{\prob}\mleft( \frac{1}{d} \mright) + \mathcal{O}_{\prob} \mleft( \frac{1}{\sqrt{T}} \mright)$ for $k=1,2$.
\end{enumerate}
Note that $\widebar{Q}_{r}^{k}, Q^k$ and $\Pi^k$ are denoted by $\widehat{P}, P_{0}$ and $D_{0}$, respectively, in \cite{doz2011two}.
Note also that \ref{Doz1} implies the analogous results for $\frac{1}{T/2} \sum_{t=1}^{T/2} X^k_{t} X^{k'}_{t}$ and $\frac{1}{T/2} \sum_{t=T/2 +1 }^{T} X^k_{t} X^{k'}_{t}$, $k=1,2$. We use the notation $\widehat{\Sigma}_2$ for either of those quantities.

Our analysis could as well be based on results similar to \ref{Doz1}--\ref{Doz3} in \cite{bai2002determining} and \cite{bai2003inferential}. We refer to \cite{doz2011two} since their statements are expressed more conveniently for our purposes.

The proofs below use a range of quantities involving $F$ and $H$.
To help the reader follow the arguments, we display the quantities indicating their dimensions:
\begin{align*}
\underbrace{F'}_{r \times T} = (\underbrace{F^{b'}}_{r \times T/2}, \underbrace{F^{a'}}_{r \times T/2}),
\hspace{0.2cm}
\underbrace{\widehat{F}'}_{r \times T},
\hspace{0.2cm}
\underbrace{\widetilde{F} \widetilde{H}}_{T\times r} = ( \underbrace{ H_{1}' F^{b'} }_{r \times T/2},  \underbrace{ H_{2}' F^{a'} }_{r \times T/2} )',
\hspace{0.2cm}
\underbrace{\widetilde{F}}_{T \times 2r} \underbrace{ \widetilde{H}}_{2r \times r},
\hspace{0.2cm}
\underbrace{H}_{r \times r}.
\end{align*}

We recall some relationships between different PCA estimators used in the factor model literature. 
Those relationships are collected in a more comprehensive way in Chapter 3 in \cite{bai2008large}.
Recall from \eqref{eq:PCAestimates_individual_G} the PCA estimators for the factors and loadings
\begin{align} \label{eq:PCA1}
\widehat{F} = \sqrt{T} \widehat{Q}_{r},
\hspace{0.2cm} 
\widehat{\Lambda}' = \frac{1}{T} \widehat{F}' Y,
\end{align}
where $\widehat{Q}_{r} $ are the eigenvectors corresponding to the $r$ largest eigenvalues of the $T \times T$ matrix $YY'$. On the other hand, we can also define
\begin{equation} \label{eq:PCA2}
\widebar{F} = \frac{1}{d} Y \widebar{\Lambda}, \hspace{0.2cm} \widebar{\Lambda} = \sqrt{d} \widebar{Q}_{r}
\end{equation}
with $\widebar{Q}_{r}$ being the $r$ eigenvectors corresponding to the $r$ largest eigenvalues of the $d \times d$ matrix $Y'Y$. 
Note that $\widebar{Q}_{r}$ and $\widebar{Q}^k_{r}$ are defined in the same way but for $Y$ and $X^k$, respectively.
We further introduce the matrix $\widehat{V}_{r}$, which is a diagonal matrix consisting of the $r$ largest eigenvalues of $\frac{1}{dT} YY'$, in particular, $\lambda_{i}\mleft( \frac{1}{dT} YY'\mright) = \lambda_{i}\mleft( \frac{1}{dT} Y'Y \mright)$ for all $i$ such that $\lambda_{i}$ is nonzero.

By equation (3.2) in \cite{bai2008large}, the two ways of writing the PCA estimators relate to each other as
\begin{align} \label{eq:PCA1PCA2}
\widehat{F}' 
= \widehat{V}_{r}^{-\frac{1}{2}} \widebar{F}',
\hspace{0.2cm} 
\widebar{\Lambda} 
= \widehat{\Lambda} \widehat{V}_{r}^{-\frac{1}{2}}.
\end{align}
Then,
\begin{align} \label{eq:PCA1PCA2_sub}
\widehat{\Lambda}
=
\widebar{\Lambda} \widehat{V}_{r}^{\frac{1}{2}} 
=
 \sqrt{d} \widebar{Q}_{r} \widehat{V}_{r}^{\frac{1}{2}} 
 =
 \widebar{Q}_{r} \widehat{\Pi}_{r}^{\frac{1}{2}},
\end{align}
where $\widehat{\Pi}_{r}$ is the diagonal matrix with the $r$ largest eigenvalues of the $d \times d$ matrix $\frac{1}{T}Y'Y$. The same relation also hold for $k=1,2$ when using data $X^k$ rather than $Y$.
This leads to an important observation regarding our estimators \eqref{eq:def_Phat12} for the projection matrices:
\begin{equation} \label{eq:projection_eigenvectors_QQ}
\widehat{P}_{k} 
= \widehat{\Lambda}_{k} ( \widehat{\Lambda}'_{k} \widehat{\Lambda}_{k})^{-1} \widehat{\Lambda}'_{k}
= \widebar{Q}_{r}^{k} \widebar{Q}_{r}^{k'},
\end{equation}
which follows since projection matrices are invariant under transformation with nonsingular matrices (see \eqref{eq:projection_invariance}) and $\widebar{Q}_r^{k'}\widebar{Q}_r^{k} = I_{r}$ due to orthonormality of the eigenvectors.

Before moving on to our auxiliary results, note from \eqref{eq:YT2def} and $Y = Y'_{1:T}$ in \eqref{eq:defofY} that
\begin{align} \label{eq:YY_block}
YY'
&=
\begin{pmatrix}
X^{2'}_{1:T/2} \widehat{P}_{1}' \widehat{P}_{1} X^{2}_{1:T/2} 
& 
X^{2'}_{1:T/2} \widehat{P}_{1}' \widehat{P}_{2} X^{2}_{(T/2+1):T}  
\\
X^{2'}_{(T/2+1):T} \widehat{P}_{2}' \widehat{P}_{1} X^{2}_{1:T/2}  
& 
X^{2'}_{(T/2+1):T} \widehat{P}_{2}' \widehat{P}_{2} X^{2}_{(T/2+1):T}
\end{pmatrix}.
\end{align}
For instance, by \eqref{eq:defofY}, the upper right block is
\begin{equation} \label{eq:B5-0.5}
\begin{aligned}
X^{2'}_{1:T/2} \widehat{P}_{1}' \widehat{P}_{2} X^{2}_{(T/2+1):T}  
&=
F^b \Lambda_2' \widehat{P}_{1}' \widehat{P}_{2} \Lambda_2 F^{a'}
+
F^b \Lambda_2' \widehat{P}_{1}' \widehat{P}_{2} \varepsilon^{a'}
\\&\hspace{1cm}+
\varepsilon^{b} \widehat{P}_{1}' \widehat{P}_{2} \Lambda_2 F^{a'}
+
\varepsilon^{b} \widehat{P}_{1}' \widehat{P}_{2} \varepsilon^{a'}.
\end{aligned}
\end{equation}
Similarly to $F^b, F^a$, set $\widehat{F}^b = \widehat{F}'_{1:T/2}$ and $\widehat{F}^b = \widehat{F}'_{(T/2+1):T}$.
We introduce the matrices
\begin{equation} \label{def:H1H2}
\begin{aligned}
H_1 &:= \frac{1}{dT} (\Lambda_2' \widehat{P}_{1}' \widehat{P}_{1} \Lambda_2 F^{b'} \widehat{F}^{b} 
+ 
\Lambda_2' \widehat{P}_{1}' \widehat{P}_{2} \Lambda_2 F^{a'} \widehat{F}^{a})
\widehat{V}^{-1}_{r},
\\
H_2 &:= \frac{1}{dT} (\Lambda_2' \widehat{P}_{2}' \widehat{P}_{2} \Lambda_2 F^{b'} \widehat{F}^{b} 
+ 
\Lambda_2' \widehat{P}_{2}' \widehat{P}_{1} \Lambda_2 F^{a'} \widehat{F}^{a})
\widehat{V}^{-1}_{r}
\end{aligned}
\end{equation}
such that
\begin{align}
\widetilde{F} \widetilde{H}
&:=
\begin{pmatrix}
F^b & 0
\\
0 & F^a
\end{pmatrix}
\begin{pmatrix}
H_1
\\
H_2
\end{pmatrix}
=
\begin{pmatrix}
F^b H_1
\\
F^a H_2
\end{pmatrix}
\nonumber
\\&=
\frac{1}{dT}
\widetilde{F}
\begin{pmatrix}
\Lambda_2' \widehat{P}_{1}' \widehat{P}_{1} \Lambda_2
&
\Lambda_2' \widehat{P}_{1}' \widehat{P}_{2} \Lambda_2
\\
\Lambda_2' \widehat{P}_{2}' \widehat{P}_{1} \Lambda_2
&
\Lambda_2' \widehat{P}_{2}' \widehat{P}_{2} \Lambda_2
\end{pmatrix}
\widetilde{F}'
\widehat{F} \widehat{V}^{-1}_{r}
\nonumber
\\&=:
\frac{1}{dT}
\widetilde{F} \Lambda_{P} \widetilde{F}'\widehat{F} \widehat{V}^{-1}_{r}.
\label{eq:rep_FtildeHtilde}
\end{align}

We now formalize and prove the key results \eqref{eq:key1}--\eqref{eq:key3}. Those results are crucial to prove Propositions \ref{prop1}--\ref{prop3} as argued in Appendix \ref{se:appB}.
Recall the notation $\delta_{dT} = \min\{ \sqrt{d}, \sqrt{T} \}$.

\begin{proposition}\label{prop2.1}
Under Assumptions \ref{ass:1}--\ref{ass:2}, \ref{ass:C1}--\ref{ass:C3},
\begin{equation*}
\frac{1}{T} \| \widehat{F} - \widetilde{F} \widetilde{H} \|_{F}^2 =
\mathcal{O}_{\prob}\mleft(\frac{1}{\delta_{dT}^2}\mright), \text{ as } d,T\to\infty.
\end{equation*}
\end{proposition}

\begin{proof}
Recall the transformed series $Y_{1:T} = \mleft( \widehat{P}_{1} X^{2}_{1:T/2}, \widehat{P}_{2} X^{2}_{ (T/2 +1) :T} \mright)$ from \eqref{eq:YT2def} as well as $Y':=Y_{1:T}$ from \eqref{eq:defofY}.
Due to \eqref{eq:PCA1}, we get $\frac{1}{dT} YY' \widehat{F} = \widehat{F} \widehat{V}_{r}$ such that
$ \widehat{F} = \frac{1}{dT} YY' \widehat{F} \widehat{V}^{-1}_{r}$.
Then, given \eqref{eq:YY_block} and \eqref{eq:rep_FtildeHtilde},
\begin{align}
\widehat{F} - \widetilde{F} \widetilde{H}
&=
\frac{1}{dT} YY' \widehat{F} \widehat{V}^{-1}_{r} - \frac{1}{dT} \widetilde{F} \Lambda_{P} \widetilde{F}' \widehat{F} \widehat{V}^{-1}_{r}
\nonumber
\\&=
\frac{1}{dT} 
( YY' - \widetilde{F} \Lambda_{P} \widetilde{F}' ) (\widehat{F} - F ) \widehat{V}^{-1}_{r} +
\frac{1}{dT} 
( YY' - \widetilde{F} \Lambda_{P} \widetilde{F}' ) F \widehat{V}^{-1}_{r}.
\label{align:both_begin}
\end{align}
Note further from \eqref{eq:YY_block}, \eqref{eq:B5-0.5} and \eqref{def:H1H2} that
\begin{align} \label{eq:block_representation}
YY' - \widetilde{F} \Lambda_{P} \widetilde{F}'
=
\begin{pmatrix}
A & C \\
C' & B
\end{pmatrix}
\end{align}
with
\begin{align*}
A &= F^b \Lambda_2' \widehat{P}_{1}' \widehat{P}_{1} \varepsilon^{b'}
+
\varepsilon^{b} \widehat{P}_{1}' \widehat{P}_{1} \Lambda_2 F^{b'}
+
\varepsilon^{b} \widehat{P}_{1}' \widehat{P}_{1} \varepsilon^{b'},
\\B &= F^{a} \Lambda_2' \widehat{P}_{2}' \widehat{P}_{2} \varepsilon^{a'}
+
\varepsilon^{a} \widehat{P}_{2}' \widehat{P}_{2} \Lambda_2 F^{a'}
+
\varepsilon^{a} \widehat{P}_{2}' \widehat{P}_{2} \varepsilon^{a'},
\\C &= F^b \Lambda_2' \widehat{P}_{1}' \widehat{P}_{2} \varepsilon^{a'}
+
\varepsilon^{b} \widehat{P}_{1}' \widehat{P}_{2} \Lambda_2 F^{a'}
+
\varepsilon^{b} \widehat{P}_{1}' \widehat{P}_{2} \varepsilon^{a'}.
\end{align*}
After applying the Frobenius norm and triangular inequality, we consider the two summands in \eqref{align:both_begin}
separately. 
For the first summand in \eqref{align:both_begin}, we get
\begin{align}
&\frac{1}{T} \| 
\frac{1}{dT} ( YY' - \widetilde{F} \Lambda_{P} \widetilde{F}' ) (\widehat{F} - F) \widehat{V}^{-1}_{r} \|_{F}^2
\nonumber
\\&\leq
\| \frac{1}{dT} ( YY' - \widetilde{F} \Lambda_{P} \widetilde{F}') \|^2 
\| \widehat{V}^{-1}_{r} \|^2
\frac{1}{T} \| \widehat{F} - F \|_{F}^2
\label{align:XX1}
\\&\leq
\frac{1}{(dT)^2}
(
\| A \|^2 + \| B \|^2 + 2\| C \|^2 )
\| \widehat{V}^{-1}_{r} \|^2
\frac{1}{T} \| \widehat{F} - F \|_{F}^2
\label{align:XX2}
\\&=
\frac{1}{(dT)^2}
(
\| A \|^2 + \| B \|^2 + 2\| C \|^2 )
\| \widehat{V}^{-1}_{r} \|^2
\mathcal{O}_{\prob}\mleft(\frac{1}{\delta^2_{dT}}\mright)
\label{align:XX3}
\\&=
\mathcal{O}_{\prob}\mleft(\frac{1}{\delta^2_{dT}}\mright),
\label{align:XX4}
\end{align}
where \eqref{align:XX1} follows by \ref{item:M0} in Appendix \ref{se:matrixnorminequalities} below. The relation \eqref{align:XX2} is due to the block representation \eqref{eq:block_representation} and \ref{item:M1}. Then, \eqref{align:XX3} follows by Proposition \ref{prop2.2} below.
For the remaining quantities in \eqref{align:XX3}, we focus on $\|A\|^2$ since the matrices $B,C$ can be treated similarly. We have
\begin{align}
\frac{1}{(dT)^2}
\| A \|^2 \| \widehat{V}^{-1}_{r} \|^2
&=
\frac{1}{(dT)^2}
\|
F^b \Lambda_2' \widehat{P}_{1}' \widehat{P}_{1} \varepsilon^{b'}
+
\varepsilon^{b} \widehat{P}_{1}' \widehat{P}_{1} \Lambda_2 F^{b'}
+
\varepsilon^{b} \widehat{P}_{1}' \widehat{P}_{1} \varepsilon^{b'} 
\|^2 \| \widehat{V}^{-1}_{r} \|^2
\nonumber
\\&\leq
3 \| \widehat{P}_{1} \|^2 \| \widehat{P}_{1} \|^2 \mleft(
\frac{1}{T} \| F \|^2 \frac{1}{d} \| \Lambda_{2} \|^2 \frac{1}{dT} \| \varepsilon \|^2 + \frac{1}{(dT)^2} \| \varepsilon \|^4
\mright)\| \widehat{V}^{-1}_{r} \|^2
\label{align:XX5}
\\&=
\mathcal{O}_{\prob}(1),
\label{align:XX6}
\end{align}
where \eqref{align:XX5} is due to submultiplicativity of the spectral norm, and the fact that for instance, 
$\| F^b \|^2 = \lambda_{\max}( F^{b'}F^b) \leq \lambda_{\max}( F^{b'}F^b + F^{a'}F^a) = \| F \|^2$ by \ref{item:M3}.
For \eqref{align:XX6}, observe that
\begin{equation} \label{eq:operatornormprojection}
\begin{aligned}
\| P_{k} \| 
&= \sqrt{\lambda_{\max}( P_{k}'P_{k} )} 
= \sqrt{\lambda_{\max}( P_{k} )}
= \sqrt{\lambda_{\max}( \Lambda_{k} ( \Lambda'_{k} \Lambda_{k} )^{-1} \Lambda'_{k} )}
\\&= \sqrt{\lambda_{\max}( \Lambda'_{k}\Lambda_{k} ( \Lambda'_{k} \Lambda_{k} )^{-1} )}
= \sqrt{\lambda_{\max}( I_{r} )} = 1
, \hspace{0.2cm} k =1,2,
\end{aligned}
\end{equation}
and similarly for their estimated counterparts $\| \widehat{P}_{k} \| = 1 $, $k=1,2$.
For the remaining terms in \eqref{align:XX5}, we get:
\begin{enumerate}[label=\textbf{S.\arabic*.},ref=S.\arabic*, align=left]
\item 
$\frac{1}{T} \| F \|^2 = \mathcal{O}_{\prob}(1)$ since $\frac{1}{T} \sum_{t=1}^T F_{t} F_{t}' \overset{\prob}{\to} \Sigma_{F}$ by the discussion following Assumption \ref{ass:2}. 
\label{cond:s1}
\item  
$\| \varepsilon \varepsilon' \| 
\leq 
T \| \frac{1}{T} \sum_{t=1}^{T} \varepsilon_t \varepsilon_t' - \Sigma_{\varepsilon} \| 
+ T \| \Sigma_{\varepsilon} \| 
= \mathcal{O}_{\prob}\mleft( d\sqrt{T} \mright) + \mathcal{O}(T)$ which follows by Assumption \ref{ass:2} (see p.\ 199 in \cite{doz2011two}) and by Assumption \ref{ass:C4}, respectively. 
 \label{cond:s2}
\item 
$\| \widehat{V}^{-1}_{r} \| = \| d \widehat{\Pi}^{-1}_{r} \| \leq d \| \widehat{\Pi}^{-1}_{r} - \Pi^{-1}_{2}\| + \| (\Pi_{2}/d)^{-1} \| = \mathcal{O}_{\prob}\mleft(\frac{1}{d}\mright) + \mathcal{O}_{\prob}\mleft(\frac{1}{\sqrt{T}}\mright) + \mathcal{O}(1)=\mathcal{O}_{\prob}(1)$ by Corollary \ref{cor:resultsforPI}\ref{cor:resultsforPI_item2}.
 \label{cond:s3}
\end{enumerate}

For the second summand in \eqref{align:both_begin}, we can infer, with further explanations given below,
\begin{align}
&\frac{1}{T} \| \frac{1}{dT} ( YY' - \widetilde{F} \Lambda_{P} \widetilde{F}' ) F \widehat{V}^{-1}_{r} \|^2_{F}
\nonumber
\\&\leq
\frac{1}{T} \| \frac{1}{dT} ( YY' - \widetilde{F} \Lambda_{P} \widetilde{F}' ) F \|^2_{F} \| \widehat{V}^{-1}_{r} \|^2
\label{eq:labelXX1}
\\&=
\frac{1}{T} \frac{1}{(dT)^2}
\mleft\|
\begin{pmatrix}
A & C \\
C' & B
\end{pmatrix}
F
\mright\|^2_F
\| \widehat{V}^{-1}_{r} \|^2
\label{eq:labelXX2}
\\&=
\frac{1}{T} \frac{1}{(dT)^2}
\mleft\|
\begin{pmatrix}
AF^b + CF^a \\
C'F^b + BF^a
\end{pmatrix}
\mright\|^2_F
\| \widehat{V}^{-1}_{r} \|^2
\nonumber
\\&\leq
\frac{1}{T} \frac{1}{(dT)^2}
2 ( \| AF^b \|^2_F + \| CF^a \|^2_F + \| C'F^b \|^2_F + \| BF^a \|^2_F ) \| \widehat{V}^{-1}_{r} \|^2
\label{eq:labelXX3}
\\&=
\mathcal{O}_{\prob}\mleft(\frac{1}{\delta^2_{dT}}\mright),
\label{eq:labelXX3YY004}
\end{align}
where \eqref{eq:labelXX1} follows again by \ref{item:M0}, \eqref{eq:labelXX2} uses the notation \eqref{eq:block_representation}, and \eqref{eq:labelXX3} is due to properties of the Frobenius norm. 
The term $\| \widehat{V}_{r}^{-1} \|^2$ in \eqref{eq:labelXX3} can be bounded by \ref{cond:s3} above. For the other terms in \eqref{eq:labelXX3}, 
we focus on the first summand
\begin{align}
AF^b
&=
F^b \Lambda_2' \widehat{P}_{1}' \widehat{P}_{1} \varepsilon^{b'} F^b 
+
\varepsilon^{b} \widehat{P}_{1}' \widehat{P}_{1} \Lambda_2 F^{b'} F^b 
+
\varepsilon^{b} \widehat{P}_{1}' \widehat{P}_{1} \varepsilon^{b'} F^b.
\label{eq:wqwqwqw12345}
\end{align}
For the first summand in \eqref{eq:wqwqwqw12345}, we get
\begin{align*}
\frac{1}{d^2T^3} \| F^b \Lambda_2' \widehat{P}_{1}' \widehat{P}_{1} \varepsilon^{b'} F^b  \|_F^2
&\leq
\frac{1}{d} \frac{1}{T} \| F \|^2 \| \widehat{P}_{1} \|^2 \| \widehat{P}_{1} \|^2 \frac{1}{d} \| \Lambda_{2}' \|^2 \| \frac{1}{T} \varepsilon^{b'} F^b \|_F^2
\\&= \frac{1}{d} \mathcal{O}_{\prob}(1) \mathcal{O}(1) \mathcal{O}_{\prob}\mleft( \frac{d}{T} \mright)
= \mathcal{O}_{\prob}\mleft( \frac{1}{T} \mright) 
= \mathcal{O}_{\prob}\mleft( \frac{1}{\delta^2_{dT}} \mright),
\end{align*}
where the asymptotics of the terms is gotten as in \eqref{align:XX6}, except for $\varepsilon^{b'} F^b$.
The asymptotics of $(\varepsilon^{b'} F^b)/T$ follow by:
\begin{enumerate}[label=\textbf{S.\arabic*.},ref=S.\arabic*, align=left]
\setcounter{enumi}{3}
\item $\frac{1}{T} \| \varepsilon' F \|^2_F = \mathcal{O}_{\prob}(dT)$ due to the same calculations as done on p.\ 199 in \cite{doz2011two} as part of their proof of Lemma 2(i) under Assumptions  \ref{ass:2} and \ref{ass:C4}. 
\label{cond:s4}
\end{enumerate}
For the second summand in \eqref{eq:wqwqwqw12345}, 
\begin{align*}
\frac{1}{d^2 T^3} \| \varepsilon^{b} \widehat{P}_{1}' \widehat{P}_{1} \Lambda_2 F^{b'} F^{b} \|_{F}^2
&\leq  
\frac{1}{d^2 T} \| \varepsilon^{b} \widehat{P}_{1}' \widehat{P}_{1} \Lambda_2 \|^2 \mleft( \frac{1}{T} \| F \|^2_F \mright)^2
 = \mathcal{O}_{\prob}\mleft( \frac{1}{\delta^2_{dT}} \mright)
 \end{align*}
by \ref{cond:s1}, and since
\begin{align}
\frac{1}{d^2 T} \| \varepsilon^{b} \widehat{P}_{1}' \widehat{P}_{1} \Lambda_2 \|^2
&\leq
\frac{2}{d^2 T} \| \varepsilon^{b} \Lambda_2 \|^2
+
\frac{2}{d^2 T} \| \varepsilon^{b} ( \widehat{P}_{1}' \widehat{P}_{1} \Lambda_2 - \Lambda_2) \|^2
\nonumber
\\&=
\frac{1}{d^2 T} \mathcal{O}_{\prob}\mleft( dT \mright) + \frac{1}{d^2 T} \mleft(\mathcal{O}_{\prob}\mleft( d\sqrt{T} \mright) + \mathcal{O}(T)\mright)
d \mleft( \mathcal{O}_{\prob}\mleft( \frac{1}{d^2} \mright) + \mathcal{O}_{\prob}\mleft( \frac{1}{T} \mright) \mright)
\label{al:poutdcvbhjhgfvdkfnek0}
\\&=
 \mathcal{O}_{\prob}\mleft( \frac{1}{\delta^2_{dT}} \mright),
\label{al:poutdcvbhjhgfvdkfnek}
\end{align}
where \eqref{al:poutdcvbhjhgfvdkfnek0} is obtained from the following observations. The contribution of $\| \varepsilon^b \|^2 \leq \| \varepsilon \|^2 = \| \varepsilon \varepsilon' \|$ by \ref{item:M3} is bounded by \ref{cond:s2}. The term $\| \widehat{P}_{1}' \widehat{P}_{1} \Lambda_2 - \Lambda_2 \|^2$ can be dealt with as \eqref{eq:crucial_addedlemma}, with the difference of extra $\Lambda_2'$ and $\sqrt{d}$.
Furthermore, with $ \sigma^2_{\varepsilon,ij}$ denoting the $ij$th element of $\Sigma_{\varepsilon}$, 
 \begin{align}
\E \| \varepsilon^b \Lambda_{2} \|^2_{F}
&= 
\sum_{s=1}^r \sum_{t=1}^{T/2} \sum_{i,j=1}^{d} \lambda_{2,is}\lambda_{2,js} \E (\varepsilon_{t,i} \varepsilon_{t,j})
\nonumber
\\&\leq 
\widebar{\lambda}^2 r \sum_{t=1}^T \sum_{i,j=1}^{d}  \sigma^2_{\varepsilon,ij}
\label{al:poutdcvbhjhgfv}
\\&=
T \widebar{\lambda}^2 r \| \Sigma_{\varepsilon} \|^2_F 
\leq 
dT \widebar{\lambda}^2 r \| \Sigma_{\varepsilon} \|^2,
\label{al:poutdcvbhjhgfvdkfnek_1}
\end{align}
where \eqref{al:poutdcvbhjhgfv} is due to Assumption \ref{ass:C3}. We further need Assumptions \ref{ass:2} and \ref{ass:C4} for the last bound to be $\mathcal{O}(1/d)$. The calculations in \eqref{al:poutdcvbhjhgfv} and \eqref{al:poutdcvbhjhgfvdkfnek_1} are essentially the same as those done for the proof of Lemma 1(ii) in \cite{bai2002determining}.
For the last summand in \eqref{eq:wqwqwqw12345}, 
\begin{align*}
\frac{1}{d^2 T^3} \| \varepsilon^{b} \widehat{P}_{1}' \widehat{P}_{1} \varepsilon^{b'} F^b \|_{F}^2
&\leq 
\frac{1}{d^2 T} \|  \varepsilon \|^2 \| \widehat{P}_{1} \|^2 \| \widehat{P}_{1} \|^2 \| \frac{1}{T} \varepsilon^{b'} F^b \|_{F}^2
\\&= \frac{1}{d^2 T} \mleft( \mathcal{O}_{\prob}\mleft( \sqrt{T}d \mright) + \mathcal{O}(T) \mright) \mathcal{O}_{\prob}\mleft( \frac{d}{T} \mright)
= \mathcal{O}_{\prob}\mleft( \frac{1}{\delta^2_{dT}} \mright)
\end{align*}
by \ref{cond:s2} and again by the calculations as on p.\ 199 in \cite{doz2011two} under Assumptions  \ref{ass:2} and \ref{ass:C4}, and $\| \widehat{P}^{k} \| =1$.

The desired result follows from \eqref{align:both_begin}, \eqref{align:XX4} and \eqref{eq:labelXX3YY004}.
\end{proof}

\begin{proposition}\label{prop2.1.0}
Under Assumptions \ref{ass:1}--\ref{ass:3}, \ref{ass:C1}--\ref{ass:C3},
\begin{equation*}
\frac{1}{T^2} \| (\widehat{F} - \widetilde{F} \widetilde{H})'F \|_{F}^2 =
\mathcal{O}_{\prob}\mleft(\frac{1}{\delta_{dT}^4}\mright), \text{ as } d,T\to\infty.
\end{equation*}
\end{proposition}

\begin{proof}
As in \eqref{align:both_begin}, we get
\begin{align}
F'(\widehat{F} - \widetilde{F} \widetilde{H})
&=
\frac{1}{dT} 
F'(YY' \widehat{F} \widehat{V}^{-1}_{r} - \widetilde{F} \Lambda_{P} \widetilde{F}' \widehat{F} \widehat{V}^{-1}_{r})
\nonumber
\\&=
\frac{1}{dT} 
F'( YY' - \widetilde{F} \Lambda_{P} \widetilde{F}' ) (\widehat{F} - F ) \widehat{V}^{-1}_{r} +
\frac{1}{dT} 
F'( YY' - \widetilde{F} \Lambda_{P} \widetilde{F}' ) F \widehat{V}^{-1}_{r}.
\label{eq:ndndndnndnXX}
\end{align}
After applying the Frobenius norm and triangle inequality, we consider the two summands in \eqref{eq:ndndndnndnXX} separately. For the first summand in \eqref{eq:ndndndnndnXX}, using the representation \eqref{eq:block_representation}, we proceed similarly to \eqref{align:XX1}--\eqref{align:XX4} but employ different norms on respective quantities. Note that
\begin{align}
&\frac{1}{T^2} \| 
\frac{1}{dT} F'( YY' - \widetilde{F} \Lambda_{P} \widetilde{F}' ) (\widehat{F} - F) \widehat{V}^{-1}_{r} \|_{F}^2
\nonumber
\\&\leq
\frac{1}{T} \frac{1}{(dT)^2}
\mleft\|
F'
\begin{pmatrix}
A & C \\
C' & B
\end{pmatrix}
\mright\|^2
\frac{1}{T}
\| \widehat{F} - F \|_{F}^2
\| \widehat{V}^{-1}_{r} \|^2
\label{align:XX1XX}
\\&=
\frac{1}{T} \frac{1}{(dT)^2}
\mleft\|
\begin{pmatrix}
F^{b'}A + F^{a'}C',
F^{b'}C + F^{a'}B
\end{pmatrix}
\mright\|^2
\frac{1}{T}
\| \widehat{F} - F \|_{F}^2
\| \widehat{V}^{-1}_{r} \|^2
\nonumber
\\&\leq
\frac{1}{T} \frac{1}{(dT)^2} \mleft(
\| F^{b'}A + F^{a'}C' \|^2 + \| F^{b'}C + F^{a'}B \|^2
\mright)
\frac{1}{T}
\| \widehat{F} - F \|_{F}^2
\| \widehat{V}^{-1}_{r} \|^2
\label{align:XX1XX00}
\\&\leq
\frac{1}{T} \frac{1}{(dT)^2} 2 \mleft(
\| F^{b'}A \|^2 + \| F^{a'}C' \|^2 + \| F^{b'}C \|^2 + \| F^{a'}B \|^2
\mright)
\frac{1}{T}
\| \widehat{F} - F \|_{F}^2
\| \widehat{V}^{-1}_{r} \|^2
\label{align:XX1XX0}
\\&=
\mathcal{O}_{\prob}\mleft(\frac{1}{\delta^4_{dT}}\mright),
\label{align:XX3XX}
\end{align}
where \eqref{align:XX1XX} follows by \ref{item:M0} and the inequality \eqref{align:XX1XX0} is due to \ref{item:M4}. The quantities in \eqref{align:XX1XX0} can then be treated separately. We get \eqref{align:XX3XX} by using \ref{cond:s3} in the proof of Proposition \ref{prop2.1}, Proposition \ref{prop2.2} and suitably bounding the four summands in \eqref{align:XX1XX0} as follows. For example, we give the detailed arguments for the second summand in \eqref{align:XX1XX0}; the remaining ones can be treated similarly. We have
\begin{align}
F^{a'}C
=
F^{a'}F^b \Lambda_2' \widehat{P}_{1}' \widehat{P}_{2} \varepsilon^{b'}
+
F^{a'}\varepsilon^{b} \widehat{P}_{1}' \widehat{P}_{2} \Lambda_2 F^{a'}
+
F^{a'}\varepsilon^{b} \widehat{P}_{1}' \widehat{P}_{2} \varepsilon^{a'}.
\label{align:XX3XX0}
\end{align}
After applying the spectral norm and triangle inequality, we consider the summands in \eqref{align:XX3XX0} separately. For the first one, 
\begin{align*}
\frac{1}{T} \frac{1}{(dT)^2} 
\| F^{a'}F^b \Lambda_2' \widehat{P}_{1}' \widehat{P}_{2} \varepsilon^{a'} \|^2 
\leq
\frac{1}{T} \| F^{a'} \|^2 \frac{1}{T} \| F^b \|^2 \frac{1}{d^2 T} \| \Lambda_2' \widehat{P}_{1}' \widehat{P}_{2} \varepsilon^{a'} \|^2 
=
\mathcal{O}_{\prob}\mleft(\frac{1}{\delta^2_{dT}}\mright)=\mathcal{O}_{\prob}(1), 
\end{align*}
by \ref{cond:s1} and arguing as for \eqref{al:poutdcvbhjhgfvdkfnek}. For the second summand in \eqref{align:XX3XX0}, 
\begin{align*}
\frac{1}{T} \frac{1}{(dT)^2} 
\| F^{a'}\varepsilon^{b} \widehat{P}_{1}' \widehat{P}_{2} \Lambda_2 F^{a'} \|^2 
\leq
\frac{1}{T^2} 
\| F^{a'} \|^4 \frac{1}{d^2T} \| \varepsilon^{b} \widehat{P}_{1}' \widehat{P}_{2} \Lambda_2 \|^2 
=
\mathcal{O}_{\prob}\mleft(\frac{1}{\delta^2_{dT}}\mright)=\mathcal{O}_{\prob}(1),
\end{align*}
by \ref{cond:s1} and \eqref{al:poutdcvbhjhgfvdkfnek}.
For the third summand in \eqref{align:XX3XX0}, 
\begin{align*}
\frac{1}{T} \frac{1}{(dT)^2} 
\| F^{a'}\varepsilon^{b} \widehat{P}_{1}' \widehat{P}_{2} \varepsilon^{a'} \|^2
\leq
\frac{1}{T} \| F^{a'} \|^2 \frac{1}{dT} \| \varepsilon^{b} \|^2 \| \widehat{P}_{1}' \widehat{P}_{2} \|^2 \frac{1}{dT} \| \varepsilon^{a} \|^2
=
\mathcal{O}_{\prob}(1)
\end{align*}
by \ref{cond:s1} and \ref{cond:s2} and since $\| \widehat{P}^{k} \| =  1 $.

For the second summand in \eqref{eq:ndndndnndnXX}, we get as in \eqref{eq:labelXX1}--\eqref{eq:labelXX3YY004}, with further explanations given below,
\begin{align}
&
\frac{1}{T^2}
\frac{1}{(dT)^2} 
\| F'( YY' - \widetilde{F} \Lambda_{P} \widetilde{F}' ) F \widehat{V}^{-1}_{r} \|^2_F
\nonumber
\\&\leq
\frac{2}{T^2}
\frac{1}{(dT)^2} \mleft(
\| F^{b'}AF^b \|^2_{F} + \| F^{a'}C'F^b \|^2_{F} + \| F^{b'}CF^a \|^2_{F} + \| F^{a'}BF^a \|^2_F \mright) \| \widehat{V}^{-1}_{r} \|^2
\label{eq:popopopopopXX}
\\&=
\mathcal{O}_{\prob}\mleft(\frac{1}{\delta^4_{dT}}\mright).
\label{eq:popopopopopXXYY}
\end{align}
For \eqref{eq:popopopopopXXYY}, we focus on $F^{b'}AF^b$, that is,
\begin{align} \label{eq:popopopopopXXYYZZ}
F^{b'}AF^b
=
F^{b'}F^b \Lambda_2' \widehat{P}_{1}' \widehat{P}_{1} \varepsilon^{b'} F^b
+
 F^{b'}\varepsilon^{b} \widehat{P}_{1}' \widehat{P}_{1} \Lambda_2 F^{b'} F^b
+
 F^{b'}\varepsilon^{b} \widehat{P}_{1}' \widehat{P}_{1} \varepsilon^{b'} F^b.
\end{align}
Before we consider the individual summands in \eqref{eq:popopopopopXXYYZZ} to prove \eqref{eq:popopopopopXXYY} for $\| F^{b'}AF^b \|^2_{F}$, note that
\begin{align}
\frac{1}{d^2T^2}
\| \Lambda_2' \widehat{P}_{1}' \widehat{P}_{1} \varepsilon^{b'} F^b \|_F^2
&\leq
\frac{2}{d^2T^2}\| \Lambda_2' \varepsilon^{b'} F^b \|^2_{F} 
+
\frac{2}{d^2T^2}\| ( \Lambda_2' - \Lambda_2' \widehat{P}_{1}' \widehat{P}_{1} ) \varepsilon^{b'} F^b \|^2_{F}
\nonumber
\\&\leq
\frac{2}{d^2T^2}\| \Lambda_2' \varepsilon^{b'} F^b \|^2_{F} 
+
\frac{2}{d^2T}\| \Lambda_2' - \Lambda_2' \widehat{P}_{1}' \widehat{P}_{1} \|^2 \frac{1}{T} \| \varepsilon^{b'} F^b \|^2_{F}
\label{eq:popoqwqwqpopqwqwAA00BB00}
\\&=
\frac{1}{d^2T^2}
\mathcal{O}_{\prob}\mleft( dT \mright) + 
\frac{1}{d^2T} d \mleft( \mathcal{O}_{\prob}\mleft( \frac{1}{d^2} \mright) + \mathcal{O}_{\prob}\mleft( \frac{1}{T} \mright) \mright) \mathcal{O}_{\prob}(dT)
\label{eq:popoqwqwqpopqwqwAA00}
\\&=
\mathcal{O}_{\prob}\mleft(\frac{1}{\delta^4_{dT}}\mright),
\label{eq:popoqwqwqpopqwqwAA} 
\end{align}
where \eqref{eq:popoqwqwqpopqwqwAA00BB00} is due to \ref{item:M0} and \eqref{eq:popoqwqwqpopqwqwAA00} can be argued similarly to \eqref{al:poutdcvbhjhgfvdkfnek0}. 
Indeed, the asymptotics of $\| \varepsilon^{b'} F^b \|_{F}^2$ follow from \ref{cond:s4} and $\Lambda_2' - \Lambda_2' \widehat{P}_{1}' \widehat{P}_{1}$ can be dealt with as \eqref{eq:crucial_addedlemma}.
For 
$\Lambda_2' \varepsilon^{b'} F^b = \mleft( \sum_{i=1}^{d} \sum_{t=1}^{T/2} \lambda_{2,is_1} \varepsilon_{t,i} F_{t,s_2} \mright)_{s_1,s_2=1,\dots,r}$, 
\begin{align}
\E \| \Lambda_2' \varepsilon^{b'} F^b \|^2_{F} 
&= 
\sum_{s_{1},s_{2}=1}^r \sum_{t_1,t_2=1}^{T/2} \sum_{i,j=1}^{d} 
\lambda_{2,is_1} \lambda_{2,js_{1}} \E ( F_{t_{1},s_{1}} F_{t_{2},s_{2}} \varepsilon_{t_1,i} \varepsilon_{t_2,j})
\nonumber
\\&= 
\sum_{s_{1},s_{2}=1}^r \sum_{t_1,t_2=1}^{T/2} \sum_{i,j=1}^{d}  \lambda_{2,is_1} \lambda_{2,js_{1}} 
\Gamma_{F,s_{1}s_{2}}(t_1 - t_2) \E(\varepsilon_{t_1,i} \varepsilon_{t_2,j})
\label{eq:popoqwqwqpopqwqw000}
\\&\leq
\widebar{\lambda}^2
\sum_{s_{1},s_{2}=1}^r T \sum_{k \in \ZZ}
| \Gamma_{F,s_{1}s_{2}}(k) | \max_{t_{1},t_{2}=1,\dots,T} \sum_{i,j=1}^{d} | \E(\varepsilon_{t_1,i} \varepsilon_{t_2,j}) |
\label{eq:popoqwqwqpopqwqw000111}
\\&\leq
\widebar{\lambda}^2
r^2 T \sum_{k \in \ZZ}
\| \Gamma_{F}(k) \| dM
=\mathcal{O}(dT),
\label{eq:popoqwqwqpopqwqw}
\end{align}
where \eqref{eq:popoqwqwqpopqwqw000} is due to the independence of factors and errors by Assumption \ref{ass:2}, and \eqref{eq:popoqwqwqpopqwqw000111} is due to Assumption \ref{ass:C3}. Finally,
\eqref{eq:popoqwqwqpopqwqw} is due to Assumption \ref{ass:3}(i). Note that \eqref{eq:popoqwqwqpopqwqw} is what is stated as Assumption 6(b) in \cite{han2015tests}. In \cite{doz2011two} as well as here, this assumption is covered by the additional structural assumptions on the factors in Assumption \ref{ass:2} which imply the summability of their autocovariances.

For the first summand in \eqref{eq:popopopopopXXYYZZ}, we get
\begin{align*}
\frac{1}{d^2 T^4} \| F^{b'}F^b \Lambda_2' \widehat{P}_{1}' \widehat{P}_{1} \varepsilon^{b'} F^b \|_F^2
&\leq
\frac{1}{T^2} \| F \|^4 
\frac{1}{d^2T^2}
\| \Lambda_2' \widehat{P}_{1}' \widehat{P}_{1} \varepsilon^{b'} F^b \|_F^2
\\&= \mathcal{O}_{\prob}(1) \mathcal{O}_{\prob}\mleft(\frac{1}{\delta^4_{dT}}\mright)
= \mathcal{O}_{\prob}\mleft(\frac{1}{\delta^4_{dT}}\mright),
\end{align*}
by \eqref{eq:popoqwqwqpopqwqwAA} and \ref{cond:s1} above. For the second summand in \eqref{eq:popopopopopXXYYZZ}, 
\begin{align*}
&
\frac{1}{d^2T^4} \| F^{b'} \varepsilon^{b} \widehat{P}_{1}' \widehat{P}_{1} \Lambda_2 F^{b'} F^{b} \|_{F}^2
\\&\leq  
\frac{1}{d^2T^2}\| F^{b'} \varepsilon^{b} \widehat{P}_{1}' \widehat{P}_{1} \Lambda_2 \|^2_{F} \mleft( \frac{1}{T} \| F \|^2 \mright)^2
 = \mathcal{O}_{\prob}\mleft(\frac{1}{\delta^4_{dT}}\mright),
 \end{align*}
again by \eqref{eq:popoqwqwqpopqwqwAA} and \ref{cond:s1}.
For the last summand in \eqref{eq:popopopopopXXYYZZ}, 
\begin{align*}
\frac{1}{d^2 T^4} \| F^{b'} \varepsilon^{b} \widehat{P}_{1}' \widehat{P}_{1} \varepsilon^{b'} F^b \|_{F}^2
\leq 
\frac{1}{d^2}  \| \frac{1}{T} \varepsilon^{b'} F^b \|_{F}^4
= \frac{1}{d^2} \mathcal{O}_{\prob}\mleft( \frac{d^2}{T^2} \mright)
= \mathcal{O}_{\prob}\mleft( \frac{1}{T^2} \mright)
= \mathcal{O}_{\prob}\mleft(\frac{1}{\delta^4_{dT}}\mright)
\end{align*}
by \ref{cond:s4} which yields \eqref{eq:popopopopopXXYY}. 

The desired result follows from \eqref{eq:ndndndnndnXX}, \eqref{align:XX3XX} and \eqref{eq:popopopopopXXYY}.
\end{proof}

\begin{lemma} \label{le:convergenceH12toH}
Recall $H_1,H_2$ from \eqref{def:H1H2}. Under Assumptions \ref{ass:1}--\ref{ass:2}, \ref{ass:C1}--\ref{ass:C4}, 
\begin{equation*}
H_1 - H_0 = \mathcal{O}_{\prob}\mleft(\frac{1}{\delta_{dT}}\mright)
\text{ and }
H_2 - H_0 = \mathcal{O}_{\prob}\mleft(\frac{1}{\delta_{dT}}\mright),
\text{ as } d,T\to\infty,
\end{equation*}
where
\begin{equation} \label{eq:def:H0}
H_0 = Q^2 \widetilde{\Pi}_{2} Q^{2'} \Sigma_{F} \widetilde{\Pi}_{2}^{-1}.
\end{equation}
\end{lemma}
 
\begin{proof}
We only show the result for $H_1$; the second result follows from analogous considerations.
Recall from \eqref{eq:HconvH0} that $H = (\Lambda'_{2} \Lambda_{2}/d) (F' \widehat{F} /T) \widehat{V}^{-1}_{r}$ and 
consider
\begin{align*}
\| H_1 - H_0 \| \leq \| H_1 - H \| + \| H - H_0 \|.
\end{align*}
Given Lemma \ref{le:convergenceHtoH0} below, it is enough to show the asymptotics of $\| H_1 - H \|$.

With further explanations given below, we get
\begin{align}
&
\| H_1 - H \|
\nonumber
\\&=
\frac{1}{dT}
\mleft\| 
\mleft(
\Lambda_2' \widehat{P}_{1}' \widehat{P}_{1} \Lambda_2 F^{b'} \widehat{F}^{b}
+ 
\Lambda_2' \widehat{P}_{1}' \widehat{P}_{2} \Lambda_2 F^{a'} \widehat{F}^{a}
\mright)
\widehat{V}^{-1}_{r}
-
\Lambda'_{2} \Lambda_{2} F' \widehat{F} \widehat{V}^{-1}_{r}
\mright\|
\nonumber
\\&\leq
\frac{1}{dT}
\mleft\| 
\mleft( \Lambda_2' \widehat{P}_{1}' \widehat{P}_{1} \Lambda_2 F^{b'}
,
\Lambda_2' \widehat{P}_{1}' \widehat{P}_{2} \Lambda_2 F^{a'} 
\mright)
-
\Lambda'_{2} \Lambda_{2} F' 
\mright\| 
\| \widehat{F}\|
\| \widehat{V}^{-1}_{r} \|
\nonumber
\\&= 
\frac{1}{dT}
\mleft\| 
\mleft( \Lambda_2' \widehat{P}_{1}' \widehat{P}_{1} \Lambda_2 F^{b'}
,
\Lambda_2' \widehat{P}_{1}' \widehat{P}_{2} \Lambda_2 F^{a'} 
\mright)
-
( \Lambda'_{2} \Lambda_{2} F^{b'}, \Lambda'_{2} \Lambda_{2} F^{a'})
\mright\| 
\| \widehat{F}\|
\| \widehat{V}^{-1}_{r} \|
\nonumber
\\&\leq
\frac{1}{d} \frac{1}{\sqrt{T}}
\mleft( 
\mleft\| \Lambda_2' \widehat{P}_{1}' \widehat{P}_{1} \Lambda_2 F^{b'} - \Lambda'_{2} \Lambda_{2} F^{b'} \mright\|
+
\mleft\| 
\Lambda_2' \widehat{P}_{1}' \widehat{P}_{2} \Lambda_2 F^{a'} 
-
\Lambda'_{2} \Lambda_{2} F^{a'}
\mright\| 
\mright)
\frac{1}{\sqrt{T}} \| \widehat{F}\|
\| \widehat{V}^{-1}_{r} \|
\label{al:proofH1_al1}
\\&= 
\mathcal{O}_{\prob}\mleft(\frac{1}{\delta_{dT}}\mright).
\label{al:proofH1_al2}
\end{align}
Here, \eqref{al:proofH1_al1} is due to \ref{item:M4}. We focus on the second summand in \eqref{al:proofH1_al1} to derive \eqref{al:proofH1_al2}. Then, 
\begin{align}
\frac{1}{d} \frac{1}{\sqrt{T}}
\mleft\| 
\Lambda_2' \widehat{P}_{1}' \widehat{P}_{2} \Lambda_2 F^{a'} 
-
\Lambda'_{2} \Lambda_{2} F^{a'}
\mright\| 
&\leq
\frac{1}{d}
\mleft\| 
\Lambda_2' \widehat{P}_{1}' \widehat{P}_{2} \Lambda_2
-
\Lambda'_{2} \Lambda_{2} \mright\| \frac{1}{\sqrt{T}} \| F^{a'} \|
\nonumber
\\&=
\mleft(\mathcal{O}_{\prob} \mleft(\frac{1}{d}\mright) + \mathcal{O}_{\prob}\mleft(\frac{1}{\sqrt{T}}\mright) \mright).
\label{al:proofH1_al3}
\end{align}
The asymptotic of $\frac{1}{\sqrt{T}} \| F^{a'} \|$ is due to \ref{cond:s1} in the proof of Proposition \ref{prop2.1}. The behavior of the estimated projection matrices for \eqref{al:proofH1_al3} follows from
\begin{align}
&
\| \Lambda_2' \widehat{P}_{1}' \widehat{P}_{2} \Lambda_2 - 
\Lambda_2' \Lambda_2 \| 
\nonumber
\\&=
\| \Lambda_2' \widehat{P}_{1}' \widehat{P}_{2} \Lambda_2 - \Lambda_2' \widehat{P}_{1}' \Lambda_2 + 
\Lambda_2' \widehat{P}_{1}' \Lambda_2 - \Lambda_2' \Lambda_2 \|
\nonumber
\\&\leq
\| \Lambda_2' \widehat{P}_{1}' \widehat{P}_{2} \Lambda_2 - \Lambda_2' \widehat{P}_{1}' \Lambda_2 \| 
+ 
\| \Lambda_2' \widehat{P}_{1}' \Lambda_2 - \Lambda_2' \Lambda_2 \|
\nonumber
\\&\leq
\| \Lambda_2' \widehat{P}_{1}' \| \| \widehat{P}_{2} \Lambda_2 - \Lambda_2 \| 
+ 
\| \Lambda_2 \| \| \widehat{P}_{1}' \Lambda_2 - \Lambda_2 \|
\label{eq:crucial_addedlemma0}
\\&=
d \mleft(\mathcal{O}_{\prob} \mleft(\frac{1}{d}\mright) + \mathcal{O}_{\prob}\mleft(\frac{1}{\sqrt{T}}\mright) \mright),
\label{eq:crucial_addedlemma}
\end{align}
since $\| \widehat{P}^{k} \| = 1 $ and for $\| \widehat{P}_{1}' \Lambda_2 - \Lambda_2 \|$ in \eqref{eq:crucial_addedlemma0} we consider the following. 
Since $P_{1} = Q^{1} Q^{1'}$ and therefore $Q^{1} Q^{1'} \Lambda_{2} = \Lambda_{2}$, 
\begin{align}
\| \widehat{P}_{1} \Lambda_2 - \Lambda_2 \|
&=
\| \widebar{Q}_{r}^{1} \widebar{Q}_{r}^{1'} \Lambda_2 - \Lambda_2 \|
\nonumber
\\&\leq
\| \widebar{Q}_{r}^{1} (\widebar{Q}_{r}^{1} - Q^{1})' \Lambda_2 \| + 
\| (\widebar{Q}_{r}^{1} - Q^{1}) Q^{1'} \Lambda_2 \|
\label{al:qqppqqpp001}
\\&=
\sqrt{d} \mleft(\mathcal{O}_{\prob} \mleft(\frac{1}{d}\mright) + \mathcal{O}_{\prob}\mleft(\frac{1}{\sqrt{T}}\mright) \mright).
\label{al:qqppqqpp01}
\end{align}
The asymptotics of \eqref{al:qqppqqpp001} can be inferred from the calculations in \eqref{eq:newineqtoexplain1} to \eqref{eq:klklklkl} below.
The quantity $\| \widehat{P}_{2} \Lambda_2 - \Lambda_2 \| $ in \eqref{eq:crucial_addedlemma0} can be treated similarly.
\end{proof}

The following lemma was used in the preceding proof and is similar in nature but concerns the matrix $H$ in \eqref{eq:HconvH0}. 

 \begin{lemma} \label{le:convergenceHtoH0}
 Under Assumptions \ref{ass:1}--\ref{ass:2}, \ref{ass:C1}--\ref{ass:C4}, 
 \begin{equation}
H - H_0 = \mathcal{O}_{\prob}\mleft(\frac{1}{\delta_{dT}}\mright)
\hspace{0.2cm}
\text{ and }
\hspace{0.2cm}
\| H_0 \| = \mathcal{O}(1),
\text{ as } d,T\to\infty
\end{equation}
with $H_0$ as in \eqref{eq:def:H0}.
\end{lemma}
 
\begin{proof}
Recall from \eqref{eq:HconvH0} that $H = (\Lambda'_{2} \Lambda_{2}/d) (F' \widehat{F} /T) \widehat{V}^{-1}_{r}$ and also that there is a matrix $\widetilde{\Pi}_2$ with $\| \Pi_2/d - \widetilde{\Pi}_2 \| = \mathcal{O}(1/\sqrt{d})$ by Assumption \ref{ass:C1}. Then, with further explanations given below,
\begin{align}
&
\| H - H_{0} \|
\nonumber
\\&=
\mleft\| \mleft(\frac{\Lambda'_{2} \Lambda_{2}}{d}\mright) \mleft( \frac{F' \widehat{F}}{T} \mright) \widehat{V}^{-1}_{r} - 
Q^2 \widetilde{\Pi}_{2} Q^{2'} \Sigma_{F} \widetilde{\Pi}_{2}^{-1}
\mright\|
\nonumber
\\&=
\mleft\|
Q^2 \mleft(\frac{\Pi_2}{d}\mright) Q^{2'} \mleft( \frac{F' \widehat{F}}{T} \mright) d\widehat{\Pi}^{-1}_{r} - 
Q^2 \widetilde{\Pi}_{2} Q^{2'} \Sigma_{F} \widetilde{\Pi}_{2}^{-1}
\mright\|
\label{al:H-H0_1}
\\&\leq
\mleft\| Q^2 \mleft(\frac{\Pi_2}{d} - \widetilde{\Pi}_{2} \mright) Q^{2'} \mleft( \frac{F' \widehat{F}}{T} \mright) d\widehat{\Pi}^{-1}_{r} \mright\|
+
\mleft\|
Q^2 \widetilde{\Pi}_{2} Q^{2'} \mleft( \frac{F' \widehat{F}}{T} \mright) d\widehat{\Pi}^{-1}_{r}
- 
Q^2 \widetilde{\Pi}_{2} Q^{2'} \Sigma_{F} \widetilde{\Pi}_{2}^{-1}
\mright\|
\nonumber
\\&\leq
\mleft\| Q^2 \mleft( \frac{\Pi_2}{d} - \widetilde{\Pi}_{2} \mright) Q^{2'} \mleft( \frac{F' \widehat{F}}{T} \mright) d\widehat{\Pi}^{-1}_{r} \mright\| 
+
\mleft\|
Q^2 \widetilde{\Pi}_{2} Q^{2'} \frac{1}{\sqrt{T}} F' \frac{1}{\sqrt{T}} ( \widehat{F} - F ) d\widehat{\Pi}^{-1}_{r} \mright\|
\nonumber
\\ &\hspace{1cm}+
\mleft\|
Q^2 \widetilde{\Pi}_{2} Q^{2'} \mleft( \frac{F'F}{T} - \Sigma_{F} \mright) d\widehat{\Pi}^{-1}_{r}
\mright\|
+
\mleft\|
Q \widetilde{\Pi}_{2} Q' \Sigma_{F} ( d\widehat{\Pi}^{-1}_{r} - \widetilde{\Pi}_{2}^{-1})
\mright\|
\nonumber
\\&\leq
\| Q \|^2 \mleft\| \frac{\Pi}{d} - \widetilde{\Pi}_{2} \mright\| \frac{1}{\sqrt{T}} \| F\| \frac{1}{\sqrt{T}} \| \widehat{F} \| \| d\widehat{\Pi}^{-1}_{r} \|
+
\| Q \|^2
\| \widetilde{\Pi}_{2} \| \frac{1}{\sqrt{T}} \| F \| \frac{1}{\sqrt{T}} \| \widehat{F} - F \| \| d\widehat{\Pi}^{-1}_{r} \|
\nonumber
\\ &\hspace{1cm}+
\| Q \|^2 \| \widetilde{\Pi}_{2} \| \mleft\| \frac{F'F}{T} - \Sigma_{F} \mright\| \| d\widehat{\Pi}^{-1}_{r} \|
+
\| Q \|^2 \| \widetilde{\Pi}_{2} \| \| \Sigma_{F} \| \| d \widehat{\Pi}^{-1}_{r} - \widetilde{\Pi}_{2}^{-1} \|
\label{al:H-H0_2}
\\&= 
\mathcal{O}_{\prob}\mleft(\frac{1}{\delta_{dT}}\mright).
\nonumber
\end{align}
For \eqref{al:H-H0_1} recall that $\Lambda'_{2} \Lambda_{2} = Q^2 \Pi_2 Q^{2'} $ in \eqref{eq:def:H0} and $\widehat{V}^{-1}_{r} = d\widehat{\Pi}^{-1}_{r}$.
The asymptotics of the quantities in \eqref{al:H-H0_2} are the following:
\begin{enumerate}[label=(\roman*)]
\item $\mleft\| \frac{\Pi_2}{d} - \widetilde{\Pi}_{2} \mright\| = \mathcal{O}\mleft(\frac{1}{\sqrt{d}}\mright)$ by Assumption \ref{ass:C1}.
\item $\frac{1}{T} \| F \|^2 = \mathcal{O}_{\prob}(1)$ by \ref{cond:s1} in the proof of Proposition \ref{prop2.1} and $\frac{1}{\sqrt{T}} \| \widehat{F} - F \| = \mathcal{O}_{\prob}\mleft(\frac{1}{\delta_{dT}}\mright)$ by Proposition \ref{prop2.2}.
\item $\mleft\| \frac{F'F}{T} - \Sigma_{F} \mright\| = \mathcal{O}_{\prob}\mleft(\frac{1}{\sqrt{T}}\mright)$; see p.\ 199 in \cite{doz2011two}.
\item $\mleft( \frac{1}{d}\widehat{\Pi}_{r} \mright)^{-1} = \mathcal{O}_{\prob}(1)$ by Corollary \ref{cor:resultsforPI}\ref{cor:resultsforPI_item2}.
\end{enumerate}
Finally, we have
\begin{equation*}
\| H_0 \| = \| Q^2 \widetilde{\Pi}_{2} Q^{2'} \Sigma_{F} \widetilde{\Pi}_{2}^{-1} \| = \mathcal{O}(1).
\end{equation*}
\end{proof}

\section{Results for consistency of PCA estimators under null hypothesis} \label{se:consistency_PCA_NH}

We present some results and their proofs to establish consistent estimation of PCA estimators under the null hypothesis. Those results were used in Appendix \ref{se:appB}.

\begin{lemma}\label{lem2.1}
Set $\widehat{\Sigma}_{Y} = \frac{1}{T} Y'Y$. Then, under Assumptions \ref{ass:1}--\ref{ass:2}, \ref{ass:C1}--\ref{ass:C4}, 
\begin{equation*}
\frac{1}{d} \| \widehat{\Sigma}_{Y} - \Lambda_{2} \Lambda'_{2} \| = \mathcal{O}_{\prob}\mleft( \frac{1}{d} \mright) + \mathcal{O}_{\prob} \mleft( \frac{1}{\sqrt{T}} \mright).
\end{equation*}
\end{lemma}

\begin{proof}
We follow the proof of Lemma 2 in \cite{doz2011two}. Consider
\begin{equation} \label{pr:eq1}
\frac{1}{d} \| \widehat{\Sigma}_{Y} - \Lambda_{2} \Lambda'_{2} \|
\leq
\frac{1}{d} \| \widehat{\Sigma}_{Y} - \Sigma \|
+
\frac{1}{d} \| \Sigma - \Lambda_{2} \Lambda'_{2} \|
\end{equation}
with 
\begin{align} \label{eq:Sigma-sep}
\Sigma 
&=
\Lambda_{2} \Lambda'_{2} + P_{2} \Sigma_{\varepsilon} P_{2}' 
\nonumber
\\&= 
P_{2} \Lambda_{2} \Lambda'_{2} P_{2}' + P_{2} \Sigma_{\varepsilon} P_{2}' 
\nonumber
\\&=
\frac{1}{2} \Big( P_{1} \Lambda_{2} \Lambda'_{2} P_{1}' + P_{1} \Sigma_{\varepsilon} P_{1}' \Big) + 
\frac{1}{2} \Big( P_{2} \Lambda_{2} \Lambda'_{2} P_{2}' + P_{2} \Sigma_{\varepsilon} P_{2}' \Big),
\end{align}
where $P_{1} = P_{2}$ under the null hypothesis.
We consider the two summands in \eqref{pr:eq1} separately. For the second summand, we get
\begin{align*}
\frac{1}{d} \| \Sigma - \Lambda_2 \Lambda'_2 \|
=
\frac{1}{d} \| P_{2} \Sigma_{\varepsilon} P_{2}' \|
\leq
\frac{1}{d} \| P_{2} \|^2 \| \Sigma_{\varepsilon} \|
=
\frac{1}{d} \| \Sigma_{\varepsilon} \|
=
O\mleft(\frac{1}{d} \mright)
\end{align*}
due to Assumption \ref{ass:C4} and our observation \eqref{eq:operatornormprojection}.
For the first summand in \eqref{pr:eq1}, we separate the series according to the transformation \eqref{eq:YT2def} as
\begin{align}
\| \widehat{\Sigma}_{Y} - \Sigma \| \nonumber
&=
\mleft\| \frac{1}{2} \widehat{P}_{1} \frac{1}{T/2} \sum_{t=1}^{T/2} X^2_{t} X^{2'}_{t} \widehat{P}_{1}'
+ 
\frac{1}{2} \widehat{P}_{2} \frac{1}{T/2} \sum_{t=T/2+1}^{T} X^2_{t} X^{2'}_{t} \widehat{P}_{2}'
- 
\Sigma \mright\| \nonumber
\\&\leq
\frac{1}{2} 
\mleft\| \widehat{P}_{1} \frac{1}{T/2} \sum_{t=1}^{T/2} X^2_{t} X_{t}^{2'} \widehat{P}_{1}' - 
\Big( P_{1} \Lambda_{2} \Lambda'_{2} P_{1}' + P_{1} \Sigma_{\varepsilon} P_{1}' \Big)
\mright\|
\nonumber
\\&\hspace{1cm}+
\frac{1}{2} \mleft\| 
\widehat{P}_{2} \frac{1}{T/2} \sum_{t=T/2+1}^{T} X^2_{t} X^{2'}_{t} \widehat{P}_{2}' - 
\Big( P_{2} \Lambda_{2} \Lambda'_{2} P_{2}' + P_{2} \Sigma_{\varepsilon} P_{2}' \Big) 
\mright\|, \label{eq:sec:proofs:al1}
\end{align}
where we used \eqref{eq:Sigma-sep}.
We consider the two summands in \eqref{eq:sec:proofs:al1} separately. For the first summand in \eqref{eq:sec:proofs:al1}, with explanations given below,
\begin{align}
&
\mleft\| \widehat{P}_{1} \frac{1}{T/2} \sum_{t=1}^{T/2} X^2_{t} X^{2'}_{t} \widehat{P}_{1}' - 
\Big( P_{1} \Lambda_{2} \Lambda'_{2} P_{1}' + P_{1} \Sigma_{\varepsilon} P_{1}' \Big)
\mright\|
\nonumber
\\&\leq
\mleft\| \widehat{P}_{1} \Big( \widehat{\Sigma}_{2} - \Lambda_{2} \Lambda'_{2} \Big) \widehat{P}_{1}' \mright\|
+
\mleft\| \widehat{P}_{1}
\Lambda_{2} \Lambda'_{2} \widehat{P}_{1}' - 
\Big( P_{1} \Lambda_{2} \Lambda'_{2} P_{1}' + P_{1} \Sigma_{\varepsilon} P_{1}' \Big)
\mright\|
\nonumber
\\&\leq
\mleft\| \widehat{P}_{1} \Big( \widehat{\Sigma}_{2} - \Lambda_{2} \Lambda'_{2} \Big) \widehat{P}_{1}' \mright\|
+
\mleft\| \widehat{P}_{1} \Lambda_{2} \Lambda_{2}' \widehat{P}_{1}' - P_{1} \Lambda_{2} \Lambda_{2}' P_{1}'
\mright\|
+ 
\| P_{1} \Sigma_{\varepsilon} P_{1}' \| 
\nonumber
\\&\leq
\| \widehat{P}_{1} \|^2 \| \widehat{\Sigma}_{2} - \Lambda_{2} \Lambda'_{2} \|
+
\mleft\| \widehat{P}_{1} \Lambda_{2} \Lambda_{2}' \widehat{P}_{1}' - P_{1} \Lambda_{2} \Lambda_{2}' P_{1}' 
\mright\|
+
\| P_{1} \|^2 \| \Sigma_{\varepsilon} \|
\label{al:SigmaY:eq1}
\\&\leq
\| \widehat{\Sigma}_{2} - \Lambda_{2} \Lambda'_{2} \|
+
\mleft\| \widehat{P}_{1} \Lambda_{2} \Lambda_{2}' \widehat{P}_{1}' - P_{1} \Lambda_{2} \Lambda_{2}' P_{1}' 
\mright\|
+
\| \Sigma_{\varepsilon} \|
\label{al:SigmaY:eq2}
\\&
=
d \mathcal{O}_{\prob}\Big( \frac{1}{d} \Big) + d \mathcal{O}_{\prob}\Big( \frac{1}{\sqrt{T}} \Big)
+\mathcal{O}(1),
\label{al:SigmaY:eq3}
\end{align}
where \eqref{al:SigmaY:eq1} uses the submultiplicativity of the spectral norm, \eqref{al:SigmaY:eq2} is due to \eqref{eq:operatornormprojection} and similarly for their estimated counterparts $\| \widehat{P}_{k} \| = 1 $, $k=1,2$.
Finally, \eqref{al:SigmaY:eq3} is due to \ref{Doz1} in Appendix \ref{se:appB}, Lemma \ref{le:difference_projections} below and $\| \Sigma_{\varepsilon} \| = \mathcal{O}(1)$ by Assumption \ref{ass:C4}. The second summand in \eqref{eq:sec:proofs:al1} can be handled by analogous considerations.
\end{proof}

The next result  was used in the preceding proof.

\begin{lemma} \label{le:difference_projections} 
Recall $\widehat{P}_1$ in \eqref{eq:def_Phat12}. Then, under Assumptions \ref{ass:1}--\ref{ass:2}, \ref{ass:C1}--\ref{ass:C4},
\begin{equation*}
\frac{1}{d} \mleft\| \widehat{P}_{1} \Lambda_{2} \Lambda_{2}' \widehat{P}_{1}' - P_{1} \Lambda_{2} \Lambda_{2}' P_{1}' 
\mright\|
=
\mathcal{O}_{\prob}\Big( \frac{1}{d} \Big) + \mathcal{O}_{\prob}\Big( \frac{1}{\sqrt{T}} \Big).
\end{equation*}
\end{lemma}

\begin{proof}
Under the null hypothesis, $P_{1}=P_{2}$ and hence $P_{1} \Lambda_{2}=P_{2}\Lambda_2=\Lambda_2$.
Recall from \eqref{eq:projection_eigenvectors_QQ} that $\widehat{P}_{k} = \widebar{Q}^{k}_r \widebar{Q}^{k'}_r$, where $\widebar{Q}^{k}_r$ are the eigenvectors corresponding to the first $r$ largest eigenvalues of the matrix $X^{k'} X^k$. Recall from \eqref{eq:def_Q} that $Q^k=\Lambda_k R^k \Pi_k^{-1/2}$, $k=1,2$. Then, with explanations given below,
\begin{align}
&
\mleft\| \widehat{P}_{1} \Lambda_{2} \Lambda_{2}' \widehat{P}'_{1} - P_{1} \Lambda_{2} \Lambda_{2}' P_{1}' 
\mright\|
\nonumber
\\&=
\mleft\| \widebar{Q}_{r}^{1} \widebar{Q}_{r}^{1'} \Lambda_{2} \Lambda_{2}' \widebar{Q}_{r}^{1} \widebar{Q}_{r}^{1'} - \Lambda_{2} \Lambda_{2}'
\mright\|
\nonumber
\\&\leq
\mleft\| \widebar{Q}_{r}^{1} \widebar{Q}_{r}^{1'} \Lambda_{2} \Lambda_{2}' \widebar{Q}_{r}^{1} \widebar{Q}_{r}^{1'} - 
\widebar{Q}_{r}^{1} Q^{1'} \Lambda_{2} \Lambda_{2}' Q^{1} \widebar{Q}_{r}^{1'}
\mright\|
+
\mleft\| \widebar{Q}_{r}^{1} Q^{1'} \Lambda_{2} \Lambda_{2}' Q^{1} \widebar{Q}_{r}^{1'} - \Lambda_{2} \Lambda_{2}'
\mright\|
\nonumber
\\&\leq
\| \widebar{Q}_{r}^{1} \|^2
\mleft\| \widebar{Q}_{r}^{1'} \Lambda_{2} \Lambda_{2}' \widebar{Q}_{r}^{1} - Q^{1'} \Lambda_{2} \Lambda_{2}' Q^{1}
\mright\|
+
\mleft\| \widebar{Q}_{r}^{1} Q_{r}^{1'} \Lambda_{2} \Lambda_{2}' Q^{1} \widebar{Q}_r^{1'} - \Lambda_{2} \Lambda_{2}'
\mright\|
\nonumber
\\&\leq
\| (\widebar{Q}_{r}^{1} - Q^{1})' \Lambda_{2} \Lambda_{2}' (\widebar{Q}_r^{1} - Q^{1}) \|
+ 
\| (\widebar{Q}_{r}^{1} - Q^{1})' \Lambda_{2} \Lambda_{2}' Q^{1} \|
\nonumber
\\&\hspace{1cm}
+
\| Q^{1'} \Lambda_{2} \Lambda_{2}' (\widebar{Q}_{r}^{1} - Q^{1}) \|
+
\| (\widebar{Q}_{r}^{1} - Q^{1}) Q^{1'} \Lambda_{2} \Lambda_{2}' Q^{1} (\widebar{Q}_{r}^{1} - Q^{1})' \|
\nonumber
\\&\hspace{2cm}
+
\| (\widebar{Q}^{1}_r - Q^{1}) Q^{1'}_r \Lambda_{2} \Lambda_{2}' Q^{1} Q^{1'} \|
+
\| Q^{1} Q^{1'} \Lambda_{2} \Lambda_{2}' Q_{1} (\widebar{Q}^{1} - Q^{1})' \|
\label{eq:weppwpfkomvrv0}
\\&\leq
\| \widebar{Q}^{1}_r - Q^{1} \|^2 \| \Lambda_{2} \Lambda_{2}' \| +
2\| (\widebar{Q}_{r}^{1} - Q^{1})' \Lambda_{2} \Lambda_{2}' Q^{1} \|
\nonumber
\\& \hspace{1cm}+ 
\| \widebar{Q}^{1}_r - Q^{1} \|^2 \| Q^1 \|^2 \| \Lambda_{2} \Lambda_{2}' \|
 +
2\| (\widebar{Q}^{1}_r - Q^{1}) Q^{1'} \Lambda_{2} \Lambda_{2}' Q^{1} Q^{1'} \|
\label{eq:weppwpfkomvrv}
\\&=
d \mleft( \mathcal{O}_{\prob}\Big( \frac{1}{d} \Big) + \mathcal{O}_{\prob}\Big( \frac{1}{\sqrt{T}} \Big) \mright).
\nonumber
\end{align}
Note that \eqref{eq:weppwpfkomvrv0} follows since $P_{1} = Q^{1} Q^{1'}$ and therefore $Q^{1} Q^{1'} \Lambda_{2} \Lambda_{2}' Q^{1} Q^{1'} = \Lambda_{2} \Lambda_{2}' $.
The asymptotics of the first and third summand in \eqref{eq:weppwpfkomvrv} follow by Assumption \ref{ass:C1} and \ref{Doz2}. For the remaining summands in \eqref{eq:weppwpfkomvrv}, let $Q_{\bot}^{1} $ be a $d \times (d - r)$ matrix whose columns consist of vectors forming an orthonormal basis of the orthogonal space of $Q^{1} $ such that $I_d = Q^{1} Q^{1'} + Q_{\bot}^{1} Q_{\bot}^{1'}$. Then, for the second summand in \eqref{eq:weppwpfkomvrv}, 
\begin{align}
\| (\widebar{Q}_{r}^{1} - Q^{1})' \Lambda_{2} \Lambda_{2}' Q^{1} \|
&\leq
\| (\widebar{Q}_{r}^{1} - Q^{1})' \Lambda_{2} \Lambda_{2}' \|
\label{eq:newineqtoexplain1}
\\&=
\| (\widebar{Q}_{r}^{1} - Q^{1})' ( Q^{1} Q^{1'} + Q_{\bot}^{1} Q_{\bot}^{1'}) \Lambda_{2} \Lambda_{2}' \|
\nonumber
\\&=
\| \widebar{Q}_{r}^{1'} Q^{1} Q^{1'} \Lambda_{2} \Lambda_{2}'  - 
Q^{1'} \Lambda_{2} \Lambda_{2}'  - 
\widebar{Q}_{r}^{1'} Q_{\bot}^{1} Q_{\bot}^{1'} \Lambda_{2} \Lambda_{2}' \|
\nonumber
\\&=
\| ( \widebar{Q}_{r}^{1'} Q^{1}  - I_r)
Q^{1'} \Lambda_{2} \Lambda_{2}'  - 
\widebar{Q}_{r}^{1'} Q_{\bot}^{1} Q_{\bot}^{1'} \Lambda_{2} \Lambda_{2}' \|
\nonumber
\\&\leq
\| \widebar{Q}_{r}^{1'} Q^{1}  - I_r \|
\| Q^{1'} \Lambda_{2} \Lambda_{2}' \|
+
\| \widebar{Q}_{r}^{1'} Q_{\bot}^{1} \| \| Q_{\bot}^{1'} \Lambda_{2} \Lambda_{2}' \|
\nonumber
\\&\leq
\| \widebar{Q}_{r}^{1'} Q^{1}  - I_r \|
\| \Lambda_{2} \Lambda_{2}' \|
+
\| \widebar{Q}_{r}^{1'} Q_{\bot}^{1} \| \| \Lambda_{2} \Lambda_{2}' \|
\nonumber
\\&= 
d \mleft( \mathcal{O}_{\prob}\Big( \frac{1}{d} \Big) + \mathcal{O}_{\prob}\Big( \frac{1}{\sqrt{T}} \Big) \mright),
\label{eq:klklklkl}
\end{align}
where \eqref{eq:klklklkl} is due to Assumption \ref{ass:C1}, \ref{Doz3} and Corollary \ref{cor2.3} below. The last summand in \eqref{eq:weppwpfkomvrv} can be treated similarly.
%
\end{proof}

The next result is a consequence of Lemma \ref{lem2.1}. The proof arguments follow those for Lemma 2 in \cite{doz2011two}.
\begin{corollary} \label{cor:resultsforPI}
Set $\widehat{\Pi}_{r} = d\widehat{V}_{r}$ as in \eqref{eq:PCA1PCA2_sub} and recall that $\widehat{V}_{r}$ is a diagonal matrix consisting of the $r$ largest eigenvalues of $\frac{1}{dT} YY'$. Recall also the diagonal $\Pi_{2}$ from Appendix \ref{se:appB}.
Under Assumptions \ref{ass:1}--\ref{ass:2}, \ref{ass:C1}--\ref{ass:C3}, 
\begin{enumerate}[label=(\roman*), ref=\textit{(\roman*)}]
\item $\frac{1}{d} \|\widehat{\Pi}_{r} - \Pi_2 \| = 
\mathcal{O}_{\prob}\mleft(\frac{1}{d}\mright) + \mathcal{O}_{\prob}\mleft(\frac{1}{\sqrt{T}}\mright)$,
\label{cor:resultsforPI_item1}
\item $d \|\widehat{\Pi}_{r}^{-1} - \Pi^{-1}_2 \| = 
\mathcal{O}_{\prob} \mleft(\frac{1}{d}\mright) + \mathcal{O}_{\prob}\mleft(\frac{1}{\sqrt{T}}\mright)$,
\label{cor:resultsforPI_item2}
\item $\Pi_2 \widehat{\Pi}_{r}^{-1}-I_{r} = \mathcal{O}_{\prob} \mleft(\frac{1}{d}\mright) + \mathcal{O}_{\prob}\mleft(\frac{1}{\sqrt{T}}\mright)$.
\label{cor:resultsforPI_item3}
\end{enumerate}
\end{corollary}

\begin{proof}
With $\widehat{\Sigma}_{Y} = \frac{1}{T} Y'Y$, we have
\begin{displaymath}
 \Pi_2=\textrm{diag}( \lambda_{1}(\Lambda_2'\Lambda_2),\ldots, \lambda_{r}(\Lambda'_2\Lambda_2))
 =
 \textrm{diag}( \lambda_{1}(\Lambda_{2} \Lambda'_{2}),\ldots, \lambda_{r}(\Lambda_{2} \Lambda'_{2})), \
\widehat{\Pi}_{r}
 = 
 \textrm{diag}(\lambda_{1}(\widehat{\Sigma}_{Y}),\ldots,\lambda_{r}(\widehat{\Sigma}_{Y})).
\end{displaymath}

\textit{(i)}
By Weyl's Theorem (Theorem 4.3.1 in \cite{horn2012matrix}) and Lemma \ref{lem2.1}, for any $j=1,\ldots,r$, 
\begin{displaymath}
|\lambda_{j}(\widehat{\Sigma}_{Y})-\lambda_{j}(\Lambda_{2} \Lambda'_{2})| 
\leq 
\|\widehat{\Sigma}_{Y} - \Lambda_{2} \Lambda'_{2} \| 
= 
d\mleft( \mathcal{O}\mleft(\frac{1}{d}\mright) + \mathcal{O}_{\prob}\mleft(\frac{1}{\sqrt{T}}\mright)\mright),
\end{displaymath}
which gives the desired result, since $\| A \| = \max_{j=1,\dots, d} | \lambda_{j}(A)|$ for diagonal matrix $A$.

\textit{(ii)}
Write
\begin{displaymath}
d(\widehat{\Pi}_{r}^{-1} - \Pi^{-1}_2 ) 
= \mleft( \frac{1}{d}\widehat{\Pi}_{r} \mright)^{-1} \cdot \frac{1}{d}(\Pi_2-\widehat{\Pi}_{r}) \cdot \mleft(\frac{1}{d}\Pi_2 \mright)^{-1},
\end{displaymath}
so that
\begin{displaymath}
d \| \widehat{\Pi}_{r}^{-1} - \Pi_2^{-1} \|
\leq \mleft\| \mleft( \frac{1}{d}\widehat{\Pi}_{r} \mright)^{-1} \mright\| \mleft\| \mleft(\frac{1}{d}\Pi_2 \mright)^{-1} \mright\| 
\frac{1}{d} \| \Pi_2-\widehat{\Pi}_{r} \|.
\end{displaymath}
The desired result then follows by (i) and since Assumption \ref{ass:C1} implies
\begin{equation*}
\mleft(\frac{1}{d}\Pi_2\mright)^{-1} = \mathcal{O}(1)
\hspace{0.2cm}
\text{ and hence }
\hspace{0.2cm}
\mleft( \frac{1}{d}\widehat{\Pi}_{r} \mright)^{-1} = \mathcal{O}_{\prob}(1).
\end{equation*}

\textit{(iii)}
Write 
\begin{align*}
\Pi_2\widehat{\Pi}_{r}^{-1}
=
\Pi_2 ( \widehat{\Pi}_{r}^{-1} - \Pi_2^{-1} + \Pi_2^{-1})
=
I_{r} + \frac{\Pi_2}{d} d \mleft( \widehat{\Pi}_{r}^{-1} - \Pi_{2}^{-1} \mright),
\end{align*}
so that
\begin{align*}
\Pi_2 \widehat{\Pi}_{r}^{-1} - I_{r} 
=
\frac{\Pi_2}{d} d \mleft( \widehat{\Pi}_{r}^{-1} - \Pi_{2}^{-1} \mright)
\end{align*}
and use \ref{cor:resultsforPI_item2}.
\end{proof}

The next result and its corollaries are used in this appendix. See also Lemmas 3 and 4 in \cite{doz2011two}, and their proofs.
\begin{lemma}\label{lem2.2}
Let $\widehat{A}=(\widehat{a}_{ij})_{i,j=1,\ldots,r}=\widebar{Q}_{r}'Q$ with $Q:=Q^2$ defined in Appendix \ref{se:appB} and $\widebar{Q}_r$ defined in \eqref{eq:PCA2}. Under Assumptions \ref{ass:1}--\ref{ass:2}, \ref{ass:C1}--\ref{ass:C4},
$\widehat{a}_{ij} = \mathcal{O}_{\prob}\mleft(\frac{1}{d}\mright) + \mathcal{O}_{\prob}\mleft(\frac{1}{\sqrt{T}}\mright)$, $i\neq j$,
$\widehat{a}_{ii}^2 = 1 +\mathcal{O}_{\prob}\mleft(\frac{1}{d}\mright)+\mathcal{O}_{\prob}\mleft(\frac{1}{\sqrt{T}}\mright)$, $i=1,\ldots,r$.
\end{lemma}

\begin{proof} 
For the off diagonal terms $\widehat{a}_{ij}$, write 
\begin{align*}
\widehat{A} 
= \widebar{Q}_{r}'Q
&= 
\widehat{\Pi}_{r}^{-1}\widebar{Q}_{r}'\widehat{\Sigma}_{Y}Q
\\&= 
\widehat{\Pi}_{r}^{-1}\widebar{Q}_{r}'(\widehat{\Sigma}_{Y}-\Lambda_2 \Lambda'_2)Q + \widehat{\Pi}_{r}^{-1}\widebar{Q}_{r}' \Lambda_2\Lambda'_2 Q
\\&= 
\mleft( \frac{\widehat{\Pi}_{r}}{d}\mright)^{-1} \widebar{Q}_{r}' \mleft(\frac{\widehat{\Sigma}_{Y}-\Lambda_2 \Lambda'_2}{d}\mright) Q + \mleft( \frac{\widehat{\Pi}_{r}}{d}\mright)^{-1} 
\widebar{Q}_{r}' Q \mleft(\frac{\Pi_2}{d}\mright) 
\\&= 
\mathcal{O}_{\prob}\mleft(\frac{1}{d}\mright) + \mathcal{O}_{\prob}\mleft( \frac{1}{\sqrt{T}} \mright)+\mleft(\frac{\Pi_2}{d}\mright)^{-1} \widehat{A}\mleft(\frac{\Pi_2}{d}\mright),
\end{align*}
where we used Lemma \ref{lem2.1} and Corollary \ref{cor:resultsforPI}.
This means that $\widehat{a}_{ij} = \frac{\pi_{jj}}{\pi_{ii}} \widehat{a}_{ij} + \mathcal{O}_{\prob}\mleft(\frac{1}{d}\mright)+\mathcal{O}_{\prob}\mleft( \frac{1}{\sqrt{T}} \mright)$ with $\frac{\pi_{jj}}{\pi_{ii}} \neq 1$ by Assumption \ref{ass:C2} which yields the desired result, when $i\neq j$.

For the diagonal terms $\widehat{a}_{ii}$, write
\begin{align*}
\frac{\Pi_2}{d} 
&= \frac{\widehat{\Pi}_{r}}{d} + \mathcal{O}\mleft(\frac{1}{d}\mright) + \mathcal{O}\mleft( \frac{1}{\sqrt{T}} \mright) 
\\
&= \widebar{Q}_{r}'\frac{\widehat{\Sigma}_{Y}}{d}\widebar{Q}_{r} + \mathcal{O}\mleft(\frac{1}{d}\mright) + \mathcal{O}_{\prob}\mleft( \frac{1}{\sqrt{T}} \mright) 
\\
&= \widebar{Q}_{r}' \frac{\Lambda_2\Lambda'_2}{d}\widebar{Q}_{r}+\mathcal{O}_{\prob}\mleft(\frac{1}{d}\mright) + \mathcal{O}_{\prob}\mleft( \frac{1}{\sqrt{T}} \mright) 
\\
&= \widebar{Q}_{r}'Q\frac{\Pi_2}{d}Q'\widebar{Q}_{r}' + \mathcal{O}_{\prob}\mleft(\frac{1}{d}\mright) + \mathcal{O}_{\prob}\mleft( \frac{1}{\sqrt{T}} \mright) 
\\
&= \widehat{A} \frac{\Pi_2}{d} \widehat{A}' + \mathcal{O}_{\prob}\mleft(\frac{1}{d}\mright) + \mathcal{O}_{\prob}\mleft( \frac{1}{\sqrt{T}} \mright).
\end{align*}
This means that, for $i=1,\ldots,r$, $\frac{\pi_{ii}}{d}=\sum_{k=1}^{r}\frac{\pi_{kk}}{d}\widehat{a}_{ik}^2+\mathcal{O}_{\prob}\mleft(\frac{1}{d}\mright) + \mathcal{O}_{\prob}\mleft( \frac{1}{\sqrt{T}} \mright)$ which yields the desired result by using the result for the off-diagonal elements.
\end{proof}

\begin{corollary}\label{cor2.2}
Under Assumptions \ref{ass:1}--\ref{ass:2}, \ref{ass:C1}--\ref{ass:C4}, one can take $\widebar{Q}_{r}$, $Q^2=:Q$ such that
\begin{displaymath}
\widebar{Q}_{r}'Q = I_{r} + \mathcal{O}_{\prob}\mleft(\frac{1}{d}\mright) + \mathcal{O}_{\prob}\mleft(\frac{1}{\sqrt{T}}\mright),\quad 
\|\widebar{Q}_{r}-Q\|^2 = \mathcal{O}_{\prob}\mleft(\frac{1}{d}\mright) + \mathcal{O}_{\prob}\mleft(\frac{1}{\sqrt{T}}\mright).
\end{displaymath}
\end{corollary}
	
\begin{proof} 
We consider the two asymptotics separately. 
For the first result, since the columns of $Q$ are the orthonormal eigenvectors of $\Lambda_2 \Lambda'_2$ associated with distinct eigenvalues by assumption, these columns are defined uniquely up to a sign change. Then one can choose them in such a way that the diagonal terms of $\widehat{A}=\widebar{Q}_{r}'Q$ are positive. The desired result follows from Lemma \ref{lem2.2}.

For the second result, note that for any $x\in\mathbb{R}^{r}$ with $\|x\|=1$, 
\begin{displaymath}
x'(\widebar{Q}_{r}-Q)'(\widebar{Q}_{r}-Q)x 
=
x'(2I_{r} - \widebar{Q}_{r}'Q -Q'\widebar{Q}_{r})x =  \mathcal{O}_{\prob}\mleft(\frac{1}{d}\mright) + \mathcal{O}_{\prob}\mleft(\frac{1}{\sqrt{T}}\mright),
\end{displaymath}
since $\widebar{Q}_r'\widebar{Q}_r =Q'Q = I_{r}$.
\end{proof}

\begin{corollary}\label{cor2.3}
Under Assumptions \ref{ass:1}--\ref{ass:2}, \ref{ass:C1}--\ref{ass:C4}, 
\begin{equation*}
\widebar{Q}_{r}' Q_{\bot} = \mathcal{O}_{\prob}\mleft(\frac{1}{d}\mright) + \mathcal{O}_{\prob}\mleft(\frac{1}{\sqrt{T}}\mright).
\end{equation*}
\end{corollary}

\begin{proof} 
Note that $\widebar{Q}_{r} = \widehat{\Sigma}_Y \widebar{Q}_{r} \widehat{\Pi}^{-1}_{r}$. Hence, 
\begin{align*}
\widebar{Q}_{r}' Q_{\bot} 
&= \widehat{\Pi}^{-1}_{r} \widebar{Q}'_{r} \widehat{\Sigma}_Y  Q_{\bot}
\\&= \mleft( \frac{\widehat{\Pi}_{r}}{d}\mright)^{-1} \widebar{Q}'_{r} \frac{1}{d}(\widehat{\Sigma}_Y - \Lambda_2 \Lambda'_2) Q_{\bot} + \widehat{\Pi}^{-1}_{r} \widebar{Q}'_{r}  \Lambda_2 \Lambda'_2 Q_{\bot}
\\&= \mathcal{O}_{\prob}\mleft(\frac{1}{d}\mright) + \mathcal{O}_{\prob}\mleft(\frac{1}{\sqrt{T}}\mright)
\end{align*}
due to Lemma \ref{lem2.1}, Corollary \ref{cor:resultsforPI}\ref{cor:resultsforPI_item2} and since $\Lambda'_2 Q_{\bot} = 0$ because of $Q_{\bot} = Q_{\bot}^2 = (\Lambda_2 R^2 \Pi_2^{-\frac{1}{2}})_{\bot}$.
\end{proof}

The next result is a consequence of the preceding results and was used in the proofs of Propositions \ref{prop2.1} and \ref{prop2.1.0}.
\begin{proposition}\label{prop2.2}
Under Assumptions \ref{ass:1}--\ref{ass:2}, \ref{ass:C1}--\ref{ass:C4},
\begin{equation*}
\frac{1}{T} \| \widehat{F} - F \|_{F}^2 
= \frac{1}{T} \sum_{t=1}^{T} \| \widehat{F}_t-F_{t} \|^2_{F}
= \mathcal{O}_{\prob}\mleft(\frac{1}{d}\mright) + \mathcal{O}_{\prob}\mleft(\frac{1}{T}\mright), 
\text{ as } d,T\to\infty.
\end{equation*}
\end{proposition}

\begin{proof}
Using the representations \eqref{eq:PCA1}, \eqref{eq:PCA2}, their relationship \eqref{eq:PCA1PCA2}, and the definition \eqref{eq:YT2def},
\begin{equation*}
\widehat{F}_t = 
\widehat{V}_{r}^{-\frac{1}{2}} \widebar{F}_t = 
\widehat{V}_{r}^{-\frac{1}{2}} \frac{1}{d} \widebar{\Lambda}' Y_t = 
\widehat{V}_{r}^{-\frac{1}{2}} \frac{1}{\sqrt{d}} \widebar{Q}'_r Y_t = 
\begin{cases}
\widehat{\Pi}_r^{-1/2} \widebar{Q}_{r}' \widehat{P}_{1} X^2_{t},
\hspace{0.2cm}
&\text{ for }
\hspace{0.2cm}
t = 1, \dots, T/2,
\\
\widehat{\Pi}_r^{-1/2} \widebar{Q}_{r}' \widehat{P}_{2} X^2_{t},
\hspace{0.2cm}
&\text{ for }
\hspace{0.2cm}
t = T/2+1, \dots, T.
\end{cases}
\end{equation*}
We consider only the case $t = 1, \dots, T/2$. Then,
\begin{align}
\widehat{F}_t-F_{t} 
&= 
\widehat{\Pi}_r^{-1/2} \widebar{Q}_{r}' \widehat{P}_{1} X^2_{t} - F_{t}
\nonumber
\\&= 
\widehat{\Pi}_r^{-1/2}\widebar{Q}_{r}' \widehat{P}_{1} ( \Lambda_2 F_{t} + \varepsilon_{t}) - F_{t}
\nonumber
\\&= 
(\widehat{\Pi}_r^{-1/2}\widebar{Q}_{r}' \widehat{P}_{1} \Lambda_2 - I_{r})F_{t} + 
\widehat{\Pi}_r^{-1/2}\widebar{Q}_{r}' \widehat{P}_{1} \varepsilon_{t}
\nonumber
\\&= 
\mleft( \widehat{\Pi}_r^{-1/2}\widebar{Q}_{r}' (\widehat{P}_{1} \Lambda_2 - \Lambda_2) + \widehat{\Pi}_r^{-1/2}\widebar{Q}_{r}'\Lambda_2 - I_{r}\mright)F_{t} + 
\widehat{\Pi}_r^{-1/2}\widebar{Q}_{r}' \widehat{P}_{1} \varepsilon_{t}
\nonumber
\\&= 
\widehat{\Pi}_r^{-1/2}\widebar{Q}_{r}' (\widehat{P}_{1} \Lambda_2 - \Lambda_2)F_{t} + 
(\widehat{\Pi}_r^{-1/2}\widebar{Q}_{r}'\Lambda_2 - I_{r})F_{t} + 
\widehat{\Pi}_r^{-1/2}\widebar{Q}_{r}' \widehat{P}_{1} \varepsilon_{t}
\nonumber
\\&= 
\widehat{\Pi}_r^{-1/2}\widebar{Q}_{r}' (\widehat{P}_{1} \Lambda_2 - \Lambda_2)F_{t} 
+ \widehat{\Pi}_r^{-1/2} ( \widebar{Q}'_r Q - \widehat{\Pi}_r^{1/2} \Pi_2^{-1/2} ) \Pi_2^{1/2} F_{t} 
+ \widehat{\Pi}_{r}^{-1/2}\widebar{Q}_{r}' \widehat{P}_{1} \varepsilon_{t}.
\label{align_xsxsxsx1}
\end{align}
We consider the three summands in \eqref{align_xsxsxsx1} separately.
For the first one, note that
\begin{align*}
\frac{1}{T} \sum_{t=1}^{T/2} \| \widehat{\Pi}_r^{-1/2}\widebar{Q}_{r}' (\widehat{P}_{1} \Lambda_2 - \Lambda_2)F_{t} \|_{F}^2
&\leq
\| \widehat{\Pi}_r^{-1/2}\widebar{Q}_{r}' (\widehat{P}_{1} \Lambda_2 - \Lambda_2)\|^2 \frac{1}{T} \| F \|^2_{F}
\\&\leq
d \| \widehat{\Pi}_r^{-1/2} \|^2 \| \widebar{Q}_{r}' \|^2 
\frac{1}{d} \| \widehat{P}_{1} \Lambda_2 - \Lambda_2 \|^2 \frac{1}{T} \| F \|^2_{F}
\\&=
\mathcal{O}_{\prob} \mleft(\frac{1}{d^2}\mright) + \mathcal{O}_{\prob}\mleft(\frac{1}{T}\mright),
\end{align*}
since $\frac{1}{T} \| F \|^2_{F} \leq  \frac{r}{T} \| F \|^2 = \mathcal{O}_{\prob}(1)$ by \ref{cond:s1} in the proof of Proposition \ref{prop2.1}  
and $\widehat{\Pi}_{r}^{-1/2} = \frac{1}{\sqrt{d}} \mleft( \frac{1}{d} \widehat{\Pi}_{r} \mright)^{-1/2} = \mathcal{O}_{\prob}\mleft(\frac{1}{\sqrt{d}}\mright)$ by Corollary \ref{cor:resultsforPI}. Furthermore, by \eqref{al:qqppqqpp01},
\begin{align*}
\| \widehat{P}_{1} \Lambda_2 - \Lambda_2 \|
=
\sqrt{d} \mleft(\mathcal{O}_{\prob} \mleft(\frac{1}{d}\mright) + \mathcal{O}_{\prob}\mleft(\frac{1}{\sqrt{T}}\mright) \mright).
\end{align*}
For the second and third summands in \eqref{align_xsxsxsx1}, we do the same calculations as in the proof of Proposition 2 in \cite{doz2011two}. That is, for the second summand, 
\begin{align*}
\frac{1}{T} \sum_{t=1}^{T/2} \|
\widehat{\Pi}_r^{-1/2} ( \widebar{Q}'_r Q - \widehat{\Pi}_r^{1/2} \Pi_2^{-1/2} ) \Pi_2^{1/2} F_{t} \|^2_{F}
&\leq
\| \widehat{\Pi}_r^{-1/2} ( \widebar{Q}'_r Q - \widehat{\Pi}_r^{1/2} \Pi_2^{-1/2} ) \Pi_2^{1/2} \|^2 \frac{1}{T} \| F \|^2_{F} 
\\&= 
\mathcal{O}_{\prob}\mleft(\frac{1}{d^2}\mright) + \mathcal{O}_{\prob}\mleft(\frac{1}{T}\mright) 
\end{align*}
by Corollaries \ref{cor:resultsforPI}, \ref{cor2.2} and $\frac{1}{T} \| F \|^2_{F} \leq  \frac{r}{T} \| F \|^2 = \mathcal{O}_{\prob}(1)$ by \ref{cond:s1} in the proof of Proposition \ref{prop2.1}.

For the third summand, let $Q_{\perp}$ be a $d\times (d-r)$ matrix whose columns consist of vectors forming an orthonormal basis of the orthogonal space of $Q$. Then $I_{d}=QQ'+Q_{\perp}Q_{\perp}'$ and we can write
\begin{align}
\frac{1}{T} \sum_{t=1}^{T/2} \|
\widehat{\Pi}_{r}^{-1/2}\widebar{Q}_{r}' \widehat{P}_{1} \varepsilon_{t} \|^2_{F}
&=
\frac{2}{T} \sum_{t=1}^{T/2} \|
\widehat{\Pi}_{r}^{-1/2}\widebar{Q}_{r}' Q Q' \widehat{P}_{1} \varepsilon_{t} \|^2_{F}
+
\frac{2}{T} \sum_{t=1}^{T/2} \|
\widehat{\Pi}_{r}^{-1/2}\widebar{Q}_{r}' Q_{\bot} Q'_{\bot} \widehat{P}_{1} \varepsilon_{t} \|^2_{F}
\label{al:hdhdhdhdhdhdh}
\\&=
\mathcal{O}_{\prob}\mleft(\frac{1}{d^2}\mright) + \mathcal{O}_{\prob}\mleft(\frac{1}{T}\mright).
\label{al:hdhdhdhdhdhdh1}
\end{align}
To show \eqref{al:hdhdhdhdhdhdh1}, note for the first summand in \eqref{al:hdhdhdhdhdhdh},
 \begin{align}
\frac{1}{T} \sum_{t=1}^{T/2} \| Q' \widehat{P}_{1} \varepsilon_{t} \|^2_{F} 
&\leq
\frac{2}{T} \sum_{t=1}^{T/2} \| Q' (\widehat{P}_{1} - P_1) \varepsilon_{t} \|^2_{F} 
+
\frac{2}{T} \sum_{t=1}^{T/2} \| Q' P_1 \varepsilon_{t} \|^2_{F}
\nonumber
\\
&=
\frac{1}{2T} \sum_{t=1}^{T/2} \| Q' ( \widebar{Q}^{1} \widebar{Q}^{1'} - Q^{1} Q^{1'} ) \varepsilon_{t} \|^2_{F} 
+
\frac{2}{T} \sum_{t=1}^{T/2} \| Q' P_1 \varepsilon_{t} \|^2_{F}
\label{al_new1_al1}
\\
&\leq
\| \widebar{Q}^{1} \widebar{Q}^{1'} - Q^{1} Q^{1'} \|^2 \frac{1}{T} \sum_{t=1}^{T/2} \| \varepsilon_{t} \|^2_{F} 
+
\frac{2}{T} \sum_{t=1}^{T/2} \| Q' P_1 \varepsilon_{t} \|^2_{F}
\nonumber
\\
&\leq
\| \widebar{Q}^{1} \widebar{Q}^{1'} - Q^{1} Q^{1'} \|^2 \mathcal{O}_{\prob}(d)
+
\mathcal{O}_{\prob}(1)
\label{al_new1_al2}
\\
&\leq
\mleft( \| (\widebar{Q}^{1} - Q^{1'} ) \widebar{Q}^{1'} \|^2 + \| Q^{1'} ( \widebar{Q}^{1'} - Q^{1'}) \|^2 \mright) \mathcal{O}_{\prob}(d)
+
\mathcal{O}_{\prob}(1)
\nonumber
\\
&=
\mleft( \mathcal{O}_{\prob}\mleft( \frac{1}{d} \mright) + \mathcal{O}_{\prob} \mleft( \frac{1}{\sqrt{T}} \mright) \mright) \mathcal{O}_{\prob}(d)
+
\mathcal{O}_{\prob}(1)
=
\mathcal{O}_{\prob}(1),
\label{al_new1_al3}
\end{align}
where \eqref{al_new1_al1} uses the definitions of $\widehat{P}_{1}$ and $P_1$, \eqref{al_new1_al2} is explained in more detail below and \eqref{al_new1_al3} follows by \ref{Doz2}.
For \eqref{al_new1_al2}, we get
 \begin{align}
\frac{1}{T} \sum_{t=1}^{T/2} \E \| Q' P_1 \varepsilon_{t} \|^2_{F} 
&=
\frac{1}{T} \sum_{t=1}^{T/2} \E \varepsilon_{t}' P_{1}' Q Q' P_1 \varepsilon_{t}
=
\frac{1}{T} \sum_{t=1}^{T/2} \E \tr (Q' P_1 \varepsilon_{t} \varepsilon_{t}' P_{1}' Q )
\nonumber
\\&=
\frac{1}{T} \sum_{t=1}^{T/2} \tr (Q' P_1 \Sigma_{\varepsilon} P_{1}' Q )
\nonumber
\\&\leq
\frac{1}{T} \sum_{t=1}^{T/2} r \| Q' P_1 \Sigma_{\varepsilon} P_{1}' Q \|
\leq
\frac{r}{2} \| P_1 \Sigma_{\varepsilon} P_{1}' \|
\nonumber
\\&\leq 
\frac{r}{2} \| P_1 \|^2 \| \Sigma_{\varepsilon} \|
=
\frac{r}{2} \| \Sigma_{\varepsilon} \|
=
\mathcal{O}(1),
\label{al:wewewewewewe}
\end{align}
due to \eqref{eq:operatornormprojection}. The asymptotics in \eqref{al:wewewewewewe} follow by Assumption \ref{ass:C4}. 
Then, it suffices to show that 
\begin{equation*}
\|
\widehat{\Pi}_{r}^{-1/2}\widebar{Q}_{r}' Q
\|^2
=
\mathcal{O}_{\prob}\mleft(\frac{1}{d^2}\mright),
\end{equation*}
which holds due to Corollary \ref{cor2.2} and $\widehat{\Pi}_{r}^{-1/2} = \frac{1}{\sqrt{d}} \mleft( \frac{1}{d} \widehat{\Pi}_{r} \mright)^{-1/2} = \mathcal{O}_{\prob}\mleft(\frac{1}{\sqrt{d}}\mright)$ by Corollary \ref{cor:resultsforPI}. A similar argument applies to show that $\frac{1}{T} \sum_{t=1}^{T/2} \| \varepsilon_{t} \|^2_{F} = \mathcal{O}_{\prob}(d)$ as also stated in \eqref{al_new1_al2}.

For the second summand in \eqref{al:hdhdhdhdhdhdh},
 \begin{align*}
\frac{1}{T} \sum_{t=1}^{T/2} 
\| Q'_{\bot} \widehat{P}_{1} \varepsilon_{t} \|^2_{F}
&\leq
\frac{1}{T} \sum_{t=1}^{T/2} 
\| Q'_{\bot} (\widehat{P}_{1} - P_1 ) \varepsilon_{t} \|^2_{F}
+
\frac{1}{T} \sum_{t=1}^{T/2} 
\| Q'_{\bot} P_1 \varepsilon_{t} \|^2_{F}
= \mathcal{O}(d),
\end{align*}
where the first summand can be treated similarly to \eqref{al_new1_al3} and for the second summand, we get
 \begin{align*}
\frac{1}{T} \sum_{t=1}^{T/2} 
\E \| Q'_{\bot} P_1 \varepsilon_{t} \|^2_{F}
&=
\frac{1}{T} \sum_{t=1}^{T/2} 
\E \varepsilon_{t}' P_{1}' Q_{\bot} Q'_{\bot} P_1 \varepsilon_{t}
=
\frac{1}{T} \sum_{t=1}^{T/2} 
\E \tr (Q'_{\bot} P_1 \varepsilon_{t} \varepsilon_{t}' P_{1}' Q_{\bot} )
\\&=
\frac{1}{T} \sum_{t=1}^{T/2} 
\tr (Q'_{\bot} P_1 \Sigma_{\varepsilon} P_{1}' Q_{\bot} )
\leq
\frac{d-r}{2} \| P_1 \Sigma_{\varepsilon} P_{1}' \|
\\&\leq 
\frac{d-r}{2} \| \Sigma_{\varepsilon} \|
= \mathcal{O}(d)
\end{align*}
following the arguments in \eqref{al:wewewewewewe} and by Assumption \ref{ass:C4}. 
Then, it suffices to show that
\begin{equation*}
\|
\widehat{\Pi}_{r}^{-1/2}\widebar{Q}_{r}' Q_{\bot} \|^2
=
\mathcal{O}_{\prob}\mleft(\frac{1}{d^3}\mright) + \mathcal{O}_{\prob}\mleft(\frac{1}{dT}\mright),
\end{equation*}
which holds since $\widebar{Q}_{r}' Q_{\bot} = \mathcal{O}_{\prob}\mleft(\frac{1}{d}\mright) + \mathcal{O}_{\prob}\mleft(\frac{1}{\sqrt{T}}\mright)$ by Corollary \ref{cor2.3}
and $\frac{1}{\sqrt{d}} \mleft( \frac{1}{d} \widehat{\Pi}_{r} \mright)^{-1/2} = \mathcal{O}_{\prob}\mleft(\frac{1}{\sqrt{d}}\mright)$ by Corollary \ref{cor:resultsforPI}. 
\end{proof}

\section{Proofs of results in Section \ref{sec:Mainresultshypothesis}} \label{app:D}
We state here the detailed proofs of our main results which are also rephrased in Appendix \ref{se:appB}.
 \begin{proof}[Proof of Proposition \ref{prop1}]
By Assumption \ref{ass:9}(ii), it suffices to prove that
\begin{equation*}
|W(\widehat{F}) - W(FH_{0})| = o_{\prob}(1).
\end{equation*}
For that, we follow the proof of Theorem 3(i) in \cite{han2015tests}. That is, 
\begin{align}
&\|W(\widehat{F}) - W(FH_{0})\|
\nonumber
\\&\leq
\mleft\|V(\widehat{F})' \mleft( \Omega^{-1}(\widehat{F}) - \Omega^{-1}(FH_{0}) \mright) V(\widehat{F}) \mright\| 
+
\mleft\| \mleft( V(\widehat{F}) - V(FH_{0}) \mright)' \Omega^{-1}(FH_{0}) V(\widehat{F}) \mright\|
\nonumber
\\&\hspace{1cm}+
\mleft\| V(FH_{0})' \Omega^{-1}(FH_0) \mleft( V(\widehat{F}) - V(FH_{0}) \mright) \mright\|
\nonumber
\\&\leq
\mleft\|V(\widehat{F}) \mright\|^2 \mleft\| \Omega^{-1}(\widehat{F}) - \Omega^{-1}(FH_{0}) \mright\| 
+
\mleft\| V(\widehat{F}) - V(FH_{0}) \mright\| \mleft\| \Omega^{-1}(FH_{0}) \mright\| \mleft\| V(\widehat{F}) \mright\|
\nonumber
\\&\hspace{1cm}+
\mleft\| V(FH_{0}) \mright\| \mleft\| \Omega^{-1}(FH_0) \mright\| \mleft\| V(\widehat{F}) - V(FH_{0}) \mright\| 
=
o_{\prob}(1).
\label{al:proof_prop31}
\end{align}
In \ref{item:Q1}--\ref{item:Q3} below we consider the different quantities in \eqref{al:proof_prop31} separately to show the claimed asymptotic behavior. In addition, Propositions \ref{prop2} and \ref{prop3} are needed.
\begin{enumerate}[label=\textbf{Q.\arabic*.},ref=Q.\arabic*, align=left]
\item \label{item:Q1}
Note that
\begin{align}
\mleft\| \Omega^{-1}(\widehat{F}) - \Omega^{-1}(FH_{0}) \mright\| 
&=
\mleft\| \Omega^{-1}(\widehat{F}) \mleft( \Omega(\widehat{F}) - \Omega(FH_{0}) \mright) \Omega^{-1}(FH_{0}) \mright\| 
\nonumber
\\&\leq
\| \Omega^{-1}(\widehat{F}) \| \mleft\| \Omega(\widehat{F}) - \Omega(FH_{0}) \mright\| \| \Omega^{-1}(FH_{0}) \|
\nonumber
\\&=
\| \Omega^{-1}(\widehat{F}) \| \| \Omega^{-1}(FH_{0}) \| o_{\prob}(1)
\label{al:proofprop3.1al:1}
\\&=
o_{\prob}(1),
\label{al:proofprop3.1al:1.1}
\end{align}
where \eqref{al:proofprop3.1al:1} follows by Proposition \ref{prop3} and \eqref{al:proofprop3.1al:1.1} is due to \ref{item:Q3} below.
\item \label{item:Q2}
We get further that
\begin{align*}
\mleft\|V(\widehat{F}) \mright\| 
\leq 
\mleft\| V(\widehat{F}) - V(FH_{0}) \mright\| + \| V(FH_{0}) \|
= o_{\prob}(1) + \| V(FH_{0}) \|
= \mathcal{O}_{\prob}(1),
\end{align*}
where the first equality follows by Proposition \ref{prop2}. For the behavior of $\mleft\| V(FH_0) \mright\|$,
\begin{align*}
\mleft\| V(FH_0) \mright\|
&=
\mleft\|
\vechop \mleft( \frac{1}{\sqrt{T}} \sum_{t=1}^{T/2} H'_0 F_{t} F_{t}' H_0 - \frac{1}{\sqrt{T}} \sum_{t=T/2 + 1}^{T} H'_0 F_{t} F_{t}' H_0 \mright)
\mright\|
\\&\leq
\sqrt{r}
\mleft\|
\frac{1}{\sqrt{T}} \sum_{t=1}^{T/2} F_{t} F_{t}' - \frac{1}{\sqrt{T}} \sum_{t=T/2 + 1}^{T} F_{t} F_{t}'
\mright\| \| H_0 \|^2
=
\mathcal{O}_{\prob}(1)
\end{align*}
since $\| \vechop(A) \| \leq \| \vechop(A) \|_{F} \leq \| A \|_{F} \leq \sqrt{r} \| A \|$ for any matrix $A \in \RR^{r \times r}$ and by Assumption \ref{ass:2}(i).
\item \label{item:Q3}
We first consider $\mleft\| \Omega^{-1}(FH_0) \mright\|$ and then $\mleft\| \Omega^{-1}(\widehat{F}) \mright\|$. Note that $\lambda_{\min}(\Omega) > 0$ since $\Omega$ in \eqref{eq:Omega} is positive definite by Assumption \ref{ass:9}(i). Then, 
\begin{align}
\mleft\| \Omega^{-1}(FH_0) \mright\| 
&=
\lambda_{\max}\mleft(  \Omega^{-1}(FH_0) \mright)
= 
\lambda_{\min}(\Omega(FH_0))
\nonumber
\\&\leq
\| \Omega(FH_0) - \Omega \| + \lambda_{\min}(\Omega)
= o_{\prob}(1) + \lambda_{\min}(\Omega),
\label{al:Q3_2}
\end{align}
where \eqref{al:Q3_2} is due to Weyl's Theorem (Theorem 4.3.1 in \cite{horn2012matrix}) and Assumption \ref{ass:9}(i).

Moving on to $\mleft\| \Omega^{-1}(\widehat{F}) \mright\|$, we first note that
\begin{align*}
\mleft\| \Omega(\widehat{F}) - \Omega \mright\|
\leq
\| \Omega(\widehat{F}) - \Omega(FH_0) \| + \| \Omega(FH_0) - \Omega \| 
= o_{\prob}(1)
\end{align*}
by Proposition \ref{prop3} and Assumption \ref{ass:9}(i).
We can then show that
\begin{align*}
\mleft\| \Omega^{-1}(\widehat{F}) \mright\| 
&\leq 
| \lambda_{\min}(\Omega(\widehat{F})) - \lambda_{\min}(\Omega) | + \lambda_{\min}(\Omega)
\\&\leq
\| \Omega(\widehat{F}) - \Omega \| + \lambda_{\min}(\Omega)
= o_{\prob}(1) + \lambda_{\min}(\Omega).
\end{align*}
\end{enumerate}
\end{proof}

\begin{proof}[Proof of Proposition \ref{prop2}]
We follow the arguments of the proof of Theorem 1 in \cite{han2015tests}. The proof requires our results \eqref{eq:key1} (Proposition \ref{prop2.1}), \eqref{eq:key2} (Proposition \ref{prop2.1.0}) and \eqref{eq:key3} (Lemma \ref{le:convergenceH12toH}).
First, note that
\begin{align}
\| V(\widehat{F}) - V(FH_{0}) \| 
&=
\| V(\widehat{F}) - V(\widetilde{F} \widetilde{H}) + V(\widetilde{F} \widetilde{H}) - V(FH_{0}) \| 
\nonumber
\\&\leq
\| V(\widehat{F}) - V(\widetilde{F} \widetilde{H}) \| + \| V(\widetilde{F} \widetilde{H}) - V(FH_{0}) \|.
\label{al:V-V:tri}
\end{align}
We consider the two summands in \eqref{al:V-V:tri} separately and prove that they are both $o_{\prob}(1)$. 

For the first term in \eqref{al:V-V:tri}, we get 
\begin{align}
&
\| V(\widehat{F}) - V( \widetilde{F} \widetilde{H}) \|
\nonumber
\nonumber
\\&= \frac{1}{\sqrt{T}} \mleft\| \vechop \mleft( 
\sum_{t=1}^{T/2} \widehat{F}_{t} \widehat{F}_{t}' - 
\sum_{t=1}^{T/2} H'_1 F_{t} F_{t}' H_1 +
\sum_{t=T/2 + 1}^{T} H'_2 F_{t} F_{t}' H_2 -
\sum_{t=T/2 + 1}^{T} \widehat{F}_{t} \widehat{F}_{t}' \mright) \mright\|
\nonumber
\\&\leq \frac{1}{\sqrt{T}} \mleft\| \vechop \mleft( 
\sum_{t=1}^{T/2} \widehat{F}_{t} \widehat{F}_{t}' - 
\sum_{t=1}^{T/2} H'_1 F_{t} F_{t}' H_1 \mright) \mright\|
+
\frac{1}{\sqrt{T}} \mleft\| \vechop \mleft( 
\sum_{t=T/2 + 1}^{T} H'_2 F_{t} F_{t}' H_2 -
\sum_{t=T/2 + 1}^{T} \widehat{F}_{t} \widehat{F}_{t}' \mright) \mright\|.
\label{al:V-V:tri2}
\end{align}
The two summands in \eqref{al:V-V:tri2} can both be treated similarly. Therefore, we focus on the first summand. Using the same arguments as on page 36 in \cite{han2015tests_WP}, we get
\begin{align}
&
\frac{1}{\sqrt{T}} \mleft\| \sum_{t=1}^{T/2} 
\vechop \mleft( 
\widehat{F}_{t} \widehat{F}_{t}' - 
H'_1 F_{t} F_{t}' H_1 \mright) \mright\|
\nonumber
\\&=
\frac{1}{\sqrt{T}} \mleft\| \sum_{t=1}^{T/2} 
\vechop \mleft( 
\widehat{F}_{t} ( \widehat{F}_{t}' - F_{t}' H_1 )
+
( \widehat{F}_{t} - H'_1 F_t) F_{t}' H_1 \mright) \mright\|
\nonumber
\\&=
\frac{1}{\sqrt{T}} \mleft\| \sum_{t=1}^{T/2} 
\vechop \mleft( 
(\widehat{F}_{t} - H'_1 F_t) ( \widehat{F}_{t}' - F_{t}' H_1 )
+
H'_1 F_t ( \widehat{F}_{t}' - F'_t H_1)
+ 
( \widehat{F}_{t} - H'_1 F_t) F_t' H_1 \mright) \mright\|
\nonumber
\\&\leq
\frac{\sqrt{r}}{\sqrt{T}} \mleft\| \sum_{t=1}^{T/2} 
(\widehat{F}_{t} - H'_1 F_t) ( \widehat{F}_{t}' - F_{t}' H_1 )
\mright\|
+
\frac{2\sqrt{r}}{\sqrt{T}} \mleft\| \sum_{t=1}^{T/2} 
H'_1 F_t ( \widehat{F}_{t}' - F'_t H_1)
\mright\|
\nonumber
\\&\leq
\frac{\sqrt{r}}{\sqrt{T}} \sum_{t=1}^{T/2} 
\mleft\| \widehat{F}_{t} - H'_1 F_t \mright\|^2
+
\frac{2\sqrt{r}}{\sqrt{T}} \mleft\| \sum_{t=1}^{T/2} 
F_t ( \widehat{F}_{t}' - F'_t H_1)
\mright\| \| H_1 \|
\label{pr:theorem1_al1}
\\&= 
\sqrt{T}
\mathcal{O}_{\prob}\mleft(\frac{1}{\delta_{dT}^2}\mright)
= o_{\prob}(1),
\label{pr:theorem1_al2}
\end{align}
since by assumption $\sqrt{T}/d \to \infty$.
The asymptotic \eqref{pr:theorem1_al2} follows by 
Propositions \ref{prop2.1} and \ref{prop2.1.0} as well as Lemma \ref{le:convergenceH12toH}.

For the second summand in \eqref{al:V-V:tri}, we get
\begin{align}
&
\| V(\widetilde{F} \widetilde{H}) - V(FH_{0}) \|
\nonumber
\\&= \frac{1}{\sqrt{T}} \mleft\| \vechop \mleft( 
\sum_{t=1}^{T/2} H'_1 F_{t} F_{t}' H_1 - 
\sum_{t=1}^{T/2} H'_0 F_{t} F_{t}' H_{0} +
\sum_{t=T/2+ 1}^{T} H'_0 F_{t} F_{t}' H_{0} -
\sum_{t=T/2 + 1}^{T} H'_2 F_{t} F_{t}' H_2  
\mright) \mright\|
\nonumber
\\&\leq \frac{1}{\sqrt{T}} \mleft\| 
\sum_{t=1}^{T/2} 
\vechop \mleft( H'_1 F_{t} F_{t}' H_1 - 
H'_0 F_{t} F_{t}' H_{0}
\mright) \mright\|
+
\frac{1}{\sqrt{T}} \mleft\|
\sum_{t=T/2+ 1}^{T} 
\vechop \mleft( H'_0 F_{t} F_{t}' H_{0} -
H'_2 F_{t} F_{t}' H_2  
\mright) \mright\|.
\label{eq:numberbla}
\end{align}
Focusing on the first summand in \eqref{eq:numberbla}, we can infer
\begin{align}
&
\frac{1}{\sqrt{T}} \mleft\| 
\sum_{t=1}^{T/2} 
\vechop \mleft( H'_1 F_{t} F_{t}' H_1 - 
H'_0 F_{t} F_{t}' H_{0}
\mright) \mright\|
\nonumber
\\&\leq
\frac{\sqrt{r}}{\sqrt{T}} \sum_{t=1}^{T/2} 
\mleft\| H'_1F_{t} - H'_0 F_t \mright\|^2
+
\frac{2\sqrt{r}}{\sqrt{T}} \mleft\| \sum_{t=1}^{T/2} 
F_t ( F_{t}'H_1 - F'_t H_0)
\mright\| \| H_0 \|
\label{pr:theorem1_al300}
\\&\leq
\frac{\sqrt{r}}{\sqrt{T}} \sum_{t=1}^{T/2} 
\mleft\| F_{t} \mright\|^2 \| H_1 - H_0 \|^2
+
\frac{2\sqrt{r}}{\sqrt{T}} \mleft\| \sum_{t=1}^{T/2} 
F_t F_{t}' 
\mright\| \| H_1 - H_0 \| \| H_0 \|
\nonumber
\\&=
\mathcal{O}_{\prob}\mleft(\sqrt{T}\mright) \mathcal{O}_{\prob}\mleft(\frac{1}{\delta_{dT}^2}\mright)
+
\mathcal{O}_{\prob}(1) \mathcal{O}_{\prob}\mleft(\frac{1}{\delta_{dT}}\mright)
=
\mathcal{O}_{\prob}\mleft(\frac{1}{\delta_{dT}}\mright),
\label{pr:theorem1_al4}
\end{align}
where \eqref{pr:theorem1_al300} follows by the same arguments as \eqref{pr:theorem1_al1} above.
The asymptotics \eqref{pr:theorem1_al4} then follow by Assumption \ref{ass:2}(i) and Lemma \ref{le:convergenceH12toH}.
\end{proof}

\begin{proof}[Proof of Proposition \ref{prop3}]
The proof follows that of Theorem 2 in \cite{han2015tests}. Note that
\begin{align}
\| \Omega(\widehat{F}) - \Omega(FH_{0}) \|
&\leq
\| \Omega(\widehat{F}) - \Omega(\widetilde{F} \widetilde{H}) \| +
\| \Omega(\widetilde{F} \widetilde{H}) - \Omega(FH_{0}) \|
=
\mathcal{O}_{\prob}\mleft(\frac{T^{\frac{1}{3}}}{\delta_{dT}}\mright).
\label{al:prooftheorem2_prop3.3_al1}
\end{align}
Consider the first summand in \eqref{al:prooftheorem2_prop3.3_al1}. We get
\begin{align}
&
\| \Omega(\widehat{F}) - \Omega(\widetilde{F} \widetilde{H}) \|
\nonumber
\\&\leq
\mleft\| \widehat{\Gamma}_{0}(\widehat{F}) - \widehat{\Gamma}_{0}(\widetilde{F} \widetilde{H}) \mright\|
\nonumber
\\&\hspace{1cm}+
\mleft\| \sum_{j=1}^{T-1} \kappa \mleft( \frac{j}{b_{T}} \mright) \Big( \widehat{\Gamma}_{j}(\widehat{F}) + \widehat{\Gamma}_{j}'(\widehat{F}) \Big)
-
\sum_{j=1}^{T-1} \kappa \mleft( \frac{j}{b_{T}} \mright) \Big( \widehat{\Gamma}_{j}(\widetilde{F} \widetilde{H}) + \widehat{\Gamma}_{j}'(\widetilde{F} \widetilde{H}) \Big) \mright\|
\nonumber
\\&\leq
\mleft\| \widehat{\Gamma}_{0}(\widehat{F}) - \widehat{\Gamma}_{0}(\widetilde{F} \widetilde{H}) \mright\|
+
2 \sum_{j=1}^{T-1} \mleft| \kappa \mleft( \frac{j}{b_{T}} \mright) \mright| \Big\| \widehat{\Gamma}_{j}(\widehat{F}) - \widehat{\Gamma}_{j}(\widetilde{F} \widetilde{H}) \Big\|
\nonumber
\\&\leq
\mleft\| \widehat{\Gamma}_{0}(\widehat{F}) - \widehat{\Gamma}_{0}(\widetilde{F} \widetilde{H}) \mright\|
+
2 \sum_{j=1}^{b_{T}} \Big\| \widehat{\Gamma}_{j}(\widehat{F}) - \widehat{\Gamma}_{j}(\widetilde{F} \widetilde{H}) \Big\|
=
\mathcal{O}_{\prob}\mleft(\frac{T^{\frac{1}{3}}}{\delta_{dT}}\mright),
\label{al:gammagamma_al2}
\end{align}
where \eqref{al:gammagamma_al2} follows since the Bartlett kernel is supported on the interval $[-1,1]$. The asymptotic behavior \eqref{al:gammagamma_al2} then follows by \eqref{lemma8_al1} in Lemma \ref{le:Gammas} and using Assumption \ref{ass:10}. We can infer further that the term is $o_{\prob}(1)$ if $\frac{T^{\frac{2}{3}}}{d} \to 0$ as $d,T \to \infty$.

The result for the second summand in \eqref{al:prooftheorem2_prop3.3_al1} can then be inferred by similar arguments but using \eqref{lemma8_al2} in Lemma \ref{le:Gammas}.
\end{proof}

\section{Auxiliary results and their proofs} \label{app:H}
We present here some auxiliary results used in the proofs of Propositions \ref{prop1}--\ref{prop3}. 
Note that our results focus on the case when the change-point is known, i.e. when $\pi=1/2$ in \cite{han2015tests_WP} and we consider only the long-run variance estimator based on the Bartlett kernel. 
The next lemma is analogous to Lemmas 7 and 8 in \cite{han2015tests_WP}. 
\begin{lemma} \label{le:Gammas}
Under Assumptions \ref{ass:1}--\ref{ass:3}, \ref{ass:C1}--\ref{ass:C3},
\begin{align}
\sup_{0\leq j \leq T-1} \Big\| \widehat{\Gamma}_{j}(\widehat{F}) - \widehat{\Gamma}_{j}( \widetilde{F}\widetilde{H} ) \Big\|
 = \mathcal{O}_{\prob}\mleft( \frac{1}{\delta_{dT}} \mright)
\label{lemma8_al1}
\end{align}
and
\begin{align} \label{lemma8_al2}
\sup_{0\leq j \leq T-1}  \Big\| \widehat{\Gamma}_{j}(\widetilde{F}\widetilde{H}) - \widehat{\Gamma}_{j}(FH_{0}) \Big\|
 = \mathcal{O}_{\prob}\mleft( \frac{1}{\delta_{dT}} \mright).
\end{align}
\end{lemma}

\begin{proof}
Recall from \eqref{eq:def:Gamma(G)} that
\begin{align*}
\widehat{\Gamma}_{j}(F H_0) 
&= \frac{1}{T} \sum_{t=1+j}^{T} \vechop (H'_{0} F_{t} F_{t}' H_{0} - I_{r}) \vechop (H'_{0} F_{t-j} F_{t-j}' H_{0} - I_{r})'.
\end{align*} 
We further note that
\begin{align}
\widehat{\Gamma}_{j}( \widetilde{F} \widetilde{H}) 
&=
\frac{1}{T} \sum_{t=1+j}^{T/2} \vechop (H'_{1} F_{t} F_{t}' H_{1} - I_{r}) \vechop (H'_{1} F_{t-j} F_{t-j}' H_{1} - I_{r})'
\nonumber
\\&\hspace{1cm}+
\frac{1}{T} \sum_{t=1+j+T/2}^{T} \vechop (H'_{2} F_{t} F_{t}' H_{2} - I_{r}) \vechop (H'_{2} F_{t-j} F_{t-j}' H_{2} - I_{r})'
\nonumber
\\&\hspace{2cm}+
\frac{1}{T} \sum_{t=T/2+1}^{T/2+j} \vechop (H'_{2} F_{t} F_{t}' H_{2} - I_{r}) \vechop (H'_{1} F_{t-j} F_{t-j}' H_{1} - I_{r})'
\nonumber
\\&=: \widehat{\Gamma}_{1,j}(\widetilde{F} \widetilde{H}) + \widehat{\Gamma}_{2,j}(\widetilde{F} \widetilde{H}) + \widehat{\Gamma}_{3,j}(\widetilde{F} \widetilde{H}).
\label{def:Gamma1,2,3}
\end{align} 
Similarly, split $\widehat{\Gamma}_{j}( \widehat{F} ) $ into $\widehat{\Gamma}_{1,j}(\widehat{F} ) + \widehat{\Gamma}_{2,j}(\widehat{F} ) + \widehat{\Gamma}_{3,j}(\widehat{F} )$ according to the summations $\sum_{t=1+j}^{T/2}, \sum_{t=1+j+T/2}^{T}$ and $\sum_{t=T/2+1}^{T/2+j}$.

\textit{Proof of \eqref{lemma8_al1}:}
Using \eqref{def:Gamma1,2,3}, we can write
\begin{align}
\Big\| \widehat{\Gamma}_{j}(\widehat{F}) - \widehat{\Gamma}_{j}(\widetilde{F}\widetilde{H}) \Big\|
\leq 
\Big\| \widehat{\Gamma}_{1,j}(\widehat{F}) - \widehat{\Gamma}_{1,j}(\widetilde{F}\widetilde{H}) \Big\|
+
\Big\| \widehat{\Gamma}_{2,j}(\widehat{F}) - \widehat{\Gamma}_{2,j}(\widetilde{F}\widetilde{H}) \Big\|
+
\Big\| \widehat{\Gamma}_{3,j}(\widehat{F}) - \widehat{\Gamma}_{3,j}(\widetilde{F}\widetilde{H}) \Big\|.
\nonumber
\end{align}
For the proofs, we focus on the first summand. The other two follow by similar considerations.
Following the calculations on page 43 in \cite{han2015tests_WP}, we get
\begin{align}
&
\Big\| \widehat{\Gamma}_{1,j}(\widehat{F}) - \widehat{\Gamma}_{1,j}(\widetilde{F}\widetilde{H}) \Big\|
\nonumber
\\&\leq
\frac{\sqrt{r}}{T} \sum_{t=1+j}^{T/2} \| \widehat{F}_t \widehat{F}_t' \| \mleft\| \widehat{F}_{t-j} \widehat{F}'_{t-j} - H'_1 F_{t-j} F'_{t-j} H_1 \mright\|
+
\frac{\sqrt{r}}{T} \sum_{t=1+j}^{T/2} \mleft\| \widehat{F}_{t} \widehat{F}'_{t} - H'_1 F_{t} F'_{t} H_1 \mright\| \| H'_1 F_{t-j} F_{t-j}' H_1 \|
\nonumber
\\&\hspace{1cm}
+
\frac{r}{T} \sum_{t=1+j}^{T/2}  \mleft\| \widehat{F}_{t-j} \widehat{F}'_{t-j} - H'_1 F_{t-j} F'_{t-j} H_1 \mright\|
+
\frac{r}{T} \sum_{t=1+j}^{T/2} \mleft\| \widehat{F}_{t} \widehat{F}'_{t} - H'_1 F_{t} F'_{t} H_1 \mright\|.
\label{al:le:Gammas_al1}
\end{align}
All four summands in \eqref{al:le:Gammas_al1} can be treated similarly. Therefore, we focus on the first one and increase $T/2$ to $T$:
\begin{align}
&
\frac{1}{T} \sum_{t=1+j}^{T} \| \widehat{F}_t \widehat{F}_t' \| \mleft\| \widehat{F}_{t-j} \widehat{F}'_{t-j} - H'_1 F_{t-j} F'_{t-j} H_1 \mright\|
\nonumber
\\&\leq
\frac{1}{T} \sum_{t=1+j}^{T} \| \widehat{F}_t \widehat{F}_t' \| 
\mleft(
\mleft\| \widehat{F}_{t-j} ( \widehat{F}'_{t-j} - F'_{t-j} H_1) \mright\|
+
\mleft\| ( \widehat{F}_{t-j} - H'_1 F_{t-j} ) F'_{t-j}H_1 \mright\|
\mright)
\nonumber
\\&\leq
\mleft( \frac{1}{T} \sum_{t=1+j}^{T} \| \widehat{F}_t \|^4 \mright)^{\frac{1}{2}} 
\Bigg(
\Bigg( \frac{1}{T} \sum_{t=1+j}^{T} \| \widehat{F}_{t-j} \|^4 \frac{1}{T} \sum_{t=1}^{T} \| \widehat{F}'_{t} - F'_{t}H_1 \|^4 \Bigg)^{\frac{1}{4}}
\nonumber
\\&\hspace{1cm}+
\Bigg( \frac{1}{T} \sum_{t=1}^{T} \| \widehat{F}'_{t} - F'_{t}H_1 \|^4 \frac{1}{T} \sum_{t=1+j}^{T} \| F_{t-j} \|^4 \| H_1 \|^4 \Bigg)^{\frac{1}{4}} 
\Bigg)
 = \mathcal{O}_{\prob}\mleft( \frac{1}{\delta_{dT}} \mright),
 \label{al:le:Gammas_al2}
\end{align}
where \eqref{al:le:Gammas_al2} is due to Lemmas \ref{le:fourthmoments} and \ref{le:convergenceH12toH}.

\textit{Proof of \eqref{lemma8_al2}:}
Using \eqref{def:Gamma1,2,3}, we can write
\begin{align}
\Big\| \widehat{\Gamma}_{j}(\widetilde{F}\widetilde{H}) - \widehat{\Gamma}_{j}(FH_0) \Big\|
&\leq 
\Big\| \widehat{\Gamma}_{1,j}(\widetilde{F}\widetilde{H}) - \widehat{\Gamma}_{1,j}(FH_0) \Big\|
+
\Big\| \widehat{\Gamma}_{2,j}(\widetilde{F}\widetilde{H}) - \widehat{\Gamma}_{2,j}(FH_0) \Big\|
\nonumber
\\& \hspace{1cm}+
\Big\| \widehat{\Gamma}_{3,j}(\widetilde{F}\widetilde{H}) - \widehat{\Gamma}_{3,j}(FH_0) \Big\|.
\nonumber
\end{align}
We focus on the first summand. The other two can be dealt with similarly. As in \eqref{al:le:Gammas_al2}, we get
\begin{align}
&
\Big\| \widehat{\Gamma}_{1,j}(\widetilde{F}\widetilde{H}) - \widehat{\Gamma}_{1,j}(FH_0) \Big\|
\nonumber
\\&\leq
\frac{\sqrt{r}}{T}
\sum_{t=j+1}^{T/2}
\Big\| H'_1 F_t F_t' H_1 \Big\| \Big\| H'_1 F_{t-j} F_{t-j}' H_1 - H'_0 F_{t-j} F_{t-j}' H_0 \Big\|
\nonumber
\\&\hspace{1cm}+
\frac{\sqrt{r}}{T}
\sum_{t=j+1}^{T/2}
\Big\| H'_1 F_{t} F_{t}' H_1 - H'_0 F_{t} F_{t}' H_0 \Big\| \Big\| H'_0 F_{t-j} F_{t-j}' H_0 \Big\|.
\nonumber
\\&\hspace{1cm}+
\frac{r}{T}
\sum_{t=j+1}^{T/2}
\Big\| H'_1 F_{t-j} F_{t-j}' H_1 - H'_0 F_{t-j} F_{t-j}' H_0 \Big\|
+
\frac{r}{T}
\sum_{t=j+1}^{T/2}
\Big\| H'_1 F_{t} F_{t}' H_1 - H'_0 F_{t} F_{t}' H_0 \Big\|
\label{al:le:Gammas_al3}
\end{align}
All four summands in \eqref{al:le:Gammas_al3} can be treated similarly. For the first one, we get 
\begin{align}
&
\frac{1}{T}
\sum_{t=j+1}^{T/2}
\Big\| H'_1 F_t F_t' H_1 \Big\| \Big\| H'_1 F_{t-j} F_{t-j}' H_1 - H'_0 F_{t-j} F_{t-j}' H_0 \Big\|
\nonumber
\\&\leq
\mleft( \frac{1}{T} \sum_{t=1+j}^{T/2} \| F_t'H_1 \|^4 \mright)^{\frac{1}{2}} 
\mleft( \frac{1}{T} \sum_{t=1+j}^{T/2} \| H'_1 F_{t-j} F'_{t-j} (H_1-H_{0}) + (H_1-H_{0})'F_{t-j}F_{t-j}'H_{0} \|^2 \mright)^{\frac{1}{2}}
\nonumber
\\&\leq
\mleft( \frac{1}{T} \sum_{t=1+j}^{T/2} \| F_t'H_1 \|^4 \mright)^{\frac{1}{2}} 
\mleft( (\| H_1 \|^2 + \| H_{0} \|^2 ) \frac{1}{T} \sum_{t=1}^{T} \| F_{t} \|^4 \mright)^{\frac{1}{2}} \| H_1-H_{0}\|
 = \mathcal{O}_{\prob}\mleft( \frac{1}{\delta_{dT}} \mright),
\label{al:le:Gammas_al4}
\end{align}
where \eqref{al:le:Gammas_al4} is due to Lemmas \ref{le:fourthmoments} and \ref{le:convergenceH12toH}.
\end{proof}

The following lemma is analogous to Lemma 5 in \cite{han2015tests_WP}. 
\begin{lemma} \label{le:fourthmoments}
Under Assumptions \ref{ass:1}--\ref{ass:3}, \ref{ass:C1}--\ref{ass:C3},
\begin{align} \label{le:fourthmoments_le:1}
\frac{1}{T} \sum_{t=1}^{T/2} \| \widehat{F}_{t} - H'_1 F_{t} \|^4 = \mathcal{O}_{\prob}\mleft( \frac{1}{\delta^4_{dT}} \mright),
\hspace{0.2cm}
\frac{1}{T} \sum_{t=T/2+1}^{T} \| \widehat{F}_{t} - H'_2 F_{t} \|^4 = \mathcal{O}_{\prob}\mleft( \frac{1}{\delta^4_{dT}} \mright)
\end{align}
and
\begin{align} \label{le:fourthmoments_le:2}
\frac{1}{T} \sum_{t=1}^{T} \| \widehat{F}_{t} \|^4 = \mathcal{O}_{\prob}(1).
\end{align}
\end{lemma}

\begin{proof}
We first prove \eqref{le:fourthmoments_le:1} and then infer \eqref{le:fourthmoments_le:2}.

\textit{Proof of \eqref{le:fourthmoments_le:1}:}
Recall from \eqref{align:both_begin} the following relationship
\begin{align*}
\widehat{F} - \widetilde{F} \widetilde{H}
=
\frac{1}{dT} 
( YY' - \widetilde{F} \Lambda_{P} \widetilde{F}' ) (\widehat{F} - F ) \widehat{V}^{-1}_{r} +
\frac{1}{dT} 
( YY' - \widetilde{F} \Lambda_{P} \widetilde{F}' ) F \widehat{V}^{-1}_{r}.
\end{align*}
Then, with $e_{t}$ denoting the $t$th unit vector in $\RR^{T}$, 
\begin{align}
\frac{1}{T} \sum_{t=1}^{T/2} \| \widehat{F}_{t} - H'_1 F_{t} \|^4
&=
\frac{1}{T} \sum_{t=1}^{T/2} \| e'_t(\widehat{F} - \widetilde{F} \widetilde{H} ) \|^4
\nonumber
\\&\leq
c \frac{1}{T} \sum_{t=1}^{T/2} \| \frac{1}{dT} e'_t \mleft( YY' - \widetilde{F} \Lambda_{P} \widetilde{F}' \mright) (\widehat{F} - F) \widehat{V}^{-1}_{r} \|^4
\nonumber
\\&\hspace{1cm}+
c \frac{1}{T} \sum_{t=1}^{T/2} \| \frac{1}{dT} e'_t \mleft( YY' - \widetilde{F} \Lambda_{P} \widetilde{F}' \mright) F \widehat{V}^{-1}_{r}\|^4
\nonumber
\\&\leq
c \frac{1}{T} \frac{1}{d^4T^2} \sum_{t=1}^{T/2} \| e'_t \mleft( YY' - \widetilde{F} \Lambda_{P} \widetilde{F}' \mright) \|^4
\frac{1}{T^2} \| \widehat{F} - F\|^4 \| \widehat{V}^{-1}_{r} \|^4
\nonumber
\\&\hspace{1cm}+
c \frac{1}{T} \frac{1}{d^4T^4} \sum_{t=1}^{T/2} \| e'_t \mleft( YY' - \widetilde{F} \Lambda_{P} \widetilde{F}' \mright) F \|^4
\| \widehat{V}^{-1}_{r}\|^4
\nonumber
\\&\leq
c \frac{1}{T} \frac{1}{d^4T^2} \sum_{t=1}^{T/2} \| e'_t \mleft( YY' - \widetilde{F} \Lambda_{P} \widetilde{F}' \mright) \|^4
\mathcal{O}_{\prob}\mleft(\frac{1}{\delta_{dT}^4}\mright) \mathcal{O}_{\prob}(1)
\nonumber
\\&\hspace{1cm}+
c \frac{1}{T} \frac{1}{d^4T^4} \sum_{t=1}^{T/2} \| e'_t \mleft( YY' - \widetilde{F} \Lambda_{P} \widetilde{F}' \mright)F \|^4
\mathcal{O}_{\prob}(1),
\label{al:fourthmoments_al:0}
\end{align}
where \eqref{al:fourthmoments_al:0} is due to Proposition \ref{prop2.2} and Corollary \ref{cor:resultsforPI}\ref{cor:resultsforPI_item2} (together with Assumption \ref{ass:C1}).
For the two remaining quantities in \eqref{al:fourthmoments_al:0}, we get, using \eqref{eq:block_representation} and the fact that $e_t$ is considered with $t=1,\dots,T/2$, 
\begin{align}
&
\frac{1}{T} \frac{1}{d^4T^2} \sum_{t=1}^{T/2} \| e'_t \mleft( YY' - \widetilde{F} \Lambda_{P} \widetilde{F}' \mright) \|^4
\nonumber
\\&\leq
c \frac{1}{T} \frac{1}{d^4T^2} \sum_{t=1}^{T/2} \mleft( \| e'_t A \|^4 + \| e'_t C \|^4 \mright)
=
\mathcal{O}_{\prob}(1),
\label{al:fourthmoments_al:011}
\end{align}
with $e_t$ now being viewed as the $t$th unit vector in $\RR^{T/2}$.
We focus on the second summand in \eqref{al:fourthmoments_al:011}; the first one can be treated similarly. We have
\begin{align}
\| e'_t C \|^4 
&=
\|
e'_t F^b \Lambda_2' \widehat{P}_{1}' \widehat{P}_{2} \varepsilon^{a'}
+
e'_t \varepsilon^{b} \widehat{P}_{1}' \widehat{P}_{2} \Lambda_2 F^{a'}
+
e'_t \varepsilon^{b} \widehat{P}_{1}' \widehat{P}_{2} \varepsilon^{a'} 
\|^4.
\label{al:fourthmoments_al:0112}
\end{align}
For the summands in \eqref{al:fourthmoments_al:0112}, we get
\begin{align}
\frac{1}{T} \frac{1}{d^4T^2} \sum_{t=1}^{T/2} \| e'_t F^b \Lambda_2' \widehat{P}_{1}' \widehat{P}_{2} \varepsilon^{a'} \|^4
&=
\frac{1}{T} \sum_{t=1}^{T/2} \| e'_t F \|^4 \frac{1}{d^4 T^2} \| \Lambda_2' \widehat{P}_{1}' \widehat{P}_{2} \varepsilon^{a'} \|^4
\nonumber
\\&=
\frac{1}{T} \sum_{t=1}^{T/2} \| e'_t F \|^4 \frac{1}{d^4 T^2} \mathcal{O}_{\prob}\mleft( \frac{1}{\delta^4_{dT}} \mright)
\label{al:fourthmoments_al:012a}
\\&=
\mathcal{O}_{\prob}\mleft( \frac{1}{\delta^4_{dT}} \mright) = \mathcal{O}_{\prob}(1),
\label{al:fourthmoments_al:012b}
\end{align}
where \eqref{al:fourthmoments_al:012a} can be proved similarly to \eqref{al:poutdcvbhjhgfvdkfnek} and \eqref{al:fourthmoments_al:012b} follows by \ref{item:R1} below.

For the second summand in \eqref{al:fourthmoments_al:0112}, note that
\begin{align}
\frac{1}{T} \frac{1}{d^4T^2} \sum_{t=1}^{T/2} \| e'_t \varepsilon^{b} \widehat{P}_{1}' \widehat{P}_{2} \Lambda_2 F^{a'} \|^4
&\leq
\frac{1}{d^4T} \sum_{t=1}^{T/2} \| e'_t \varepsilon^{b} \widehat{P}_{1}' \widehat{P}_{2} \Lambda_2 \|^4 \frac{1}{T^2} \| F \|^4
\nonumber
\\&\leq
\frac{1}{d^4T} \sum_{t=1}^{T/2} \| e'_t \varepsilon^{b} \widehat{P}_{1}' \widehat{P}_{2} \Lambda_2 \|^4 
\mathcal{O}_{\prob}(1)
\label{al:fourthmoments_al:013a0000}
\\&\leq
\mathcal{O}_{\prob}\mleft( \frac{1}{\delta^4_{dT}} \mright) = \mathcal{O}_{\prob}(1),
\label{al:fourthmoments_al:013a}
\end{align}
where \eqref{al:fourthmoments_al:013a0000} is due to \ref{cond:s1} in the proof of Proposition \ref{prop2.1} and 
\eqref{al:fourthmoments_al:013a} follows since 
\begin{align}
\frac{1}{d^4T} \sum_{t=1}^{T} \| e'_t \varepsilon^{b} \widehat{P}_{1}' \widehat{P}_{2} \Lambda_2 \|^4
&\leq
\frac{8}{d^4T} \sum_{t=1}^{T} \| e'_t \varepsilon^{b} \Lambda_2 \|^4
+
\frac{8}{d^4T} \sum_{t=1}^{T} \| e'_t \varepsilon^{b}( \widehat{P}_{1}' \widehat{P}_{2} \Lambda_2 - \Lambda_2) \|^4
\nonumber
\\&=
\mathcal{O}_{\prob}\mleft(\frac{1}{d^2}\mright)
+
\frac{8}{d^4T} \sum_{t=1}^{T} \| e'_t \varepsilon^{b} \|^4 \| \widehat{P}_{1}' \widehat{P}_{2} \Lambda_2 - \Lambda_2 \|^4
\label{al:app:qq1000}
\\&=
\mathcal{O}_{\prob}\mleft(\frac{1}{d^2}\mright) + 
\frac{1}{d^4T} \mathcal{O}_{\prob}(d^2T) d^2 \mleft( \mathcal{O}_{\prob}\mleft( \frac{1}{d^4} \mright) + \mathcal{O}_{\prob}\mleft( \frac{1}{T^2} \mright) \mright) 
\label{al:app:qq1}
\\&=\mathcal{O}_{\prob}\mleft(\frac{1}{\delta_{dT}^4}\mright),
\label{al:app:qq1aaaa}
\end{align}
where \eqref{al:app:qq1000} is due to Assumption \ref{ass:3}(ii), \eqref{al:app:qq1} follows from \ref{item:R3} below and a bound similar to \eqref{eq:crucial_addedlemma} with the first factor $d$ replaced by $\sqrt{d}$.

For the third summand in \eqref{al:fourthmoments_al:0112}, we get
\begin{align}
\frac{1}{T} \frac{1}{d^4T^2} \sum_{t=1}^{T/2} \| e'_t \varepsilon^{b} \widehat{P}_{1}' \widehat{P}_{2} \varepsilon^{a'}  \|^4
&\leq
\frac{1}{Td^2} \sum_{t=1}^{T/2} \| e'_t \varepsilon \|^4 \frac{1}{d^2T^3} \|\varepsilon \|^4 \| \widehat{P}_{1}' \widehat{P}_{2} \|^4
\nonumber
\\&\leq
\mathcal{O}_{\prob}(1) \frac{1}{d^2T^3} \mleft( \mathcal{O}_{\prob}\mleft( d\sqrt{T} \mright) + \mathcal{O}(T) \mright)^2
\label{al:fourthmoments_al:013aaa}
\\&\leq
\mathcal{O}_{\prob}(1),
\label{al:fourthmoments_al:012aaa}
\end{align}
where \eqref{al:fourthmoments_al:013aaa} is due to \ref{cond:s2} in the proof of Proposition \ref{prop2.1} and \eqref{al:fourthmoments_al:012aaa} follows by \ref{item:R3}.

For the second summand in \eqref{al:fourthmoments_al:0}, by using \eqref{eq:block_representation}, 
\begin{align}
&
\frac{1}{T} \frac{1}{d^4T^4} \sum_{t=1}^{T/2} \| e'_t \mleft( YY' - \widetilde{F} \Lambda_{P} \widetilde{F}' \mright)F \|^4
\nonumber
\\&\leq
\frac{1}{T} \frac{1}{d^4T^4} \sum_{t=1}^{T/2}
\mleft\|
e'_t
\begin{pmatrix}
AF^b + CF^a \\
C'F^b + BF^a
\end{pmatrix}
\mright\|^4
\nonumber
\\&\leq
\frac{1}{T} \frac{1}{d^4T^4} \sum_{t=1}^{T/2}
\| e'_t AF^b + e'_t CF^a \|^4 
\nonumber
\\&\leq
c \frac{1}{T} \frac{1}{d^4T^4} \sum_{t=1}^{T/2}
\mleft(
\| e'_t AF^b \|^4 +\| e'_t CF^a \|^4 \mright)
=
\mathcal{O}_{\prob}\mleft(\frac{1}{\delta_{dT}^4}\mright).
\label{al:fourthmoments_al:1}
\end{align}
We consider only one of the summands in \eqref{al:fourthmoments_al:1}.
Recall from \eqref{eq:wqwqwqw12345} that
\begin{align}
AF^b
&=
F^b \Lambda_2' \widehat{P}_{1}' \widehat{P}_{1} \varepsilon^{b'} F^b 
+
\varepsilon^{b} \widehat{P}_{1}' \widehat{P}_{1} \Lambda_2 F^{b'} F^b 
+
\varepsilon^{b} \widehat{P}_{1}' \widehat{P}_{1} \varepsilon^{b'} F^b.
\label{app:eq:wqwqwqw12345}
\end{align}
We consider the three summands in \eqref{app:eq:wqwqwqw12345} separately. For the first one,
\begin{align}
\frac{1}{T} \frac{1}{d^4T^4} \sum_{t=1}^{T} \| e'_t F^b \Lambda_2' \widehat{P}_{1}' \widehat{P}_{1} \varepsilon^{b'} F^b \|^4
&\leq
\frac{1}{T} \sum_{t=1}^{T} \| e'_t F \|^4 \| \frac{1}{d^4T^4} \| \Lambda_2' \widehat{P}_{1}' \widehat{P}_{1} \varepsilon^{b'} F^b \|^4
\nonumber
\\&=
\frac{1}{T} \sum_{t=1}^{T} \| e'_t F \|^4  \mathcal{O}_{\prob}\mleft(\frac{1}{\delta_{dT}^8}\mright)
\label{al:fourthmoments_al:2}
\\&=
\mathcal{O}_{\prob}(1)  \mathcal{O}_{\prob}\mleft(\frac{1}{\delta_{dT}^8}\mright)
=
\mathcal{O}_{\prob}\mleft(\frac{1}{\delta_{dT}^4}\mright),
\label{al:fourthmoments_al:3}
\end{align}
where \eqref{al:fourthmoments_al:2} follows by \eqref{eq:popoqwqwqpopqwqwAA} and \eqref{al:fourthmoments_al:3} is due to \ref{item:R1} below.

For the second summand in \eqref{app:eq:wqwqwqw12345}, note that
\begin{align}
\frac{1}{T} \frac{1}{d^4T^4} \sum_{t=1}^{T} \| e'_t \varepsilon^{b} \widehat{P}_{1}' \widehat{P}_{1} \Lambda_2 F^{b'} F^{b} \|^4
&\leq
\frac{1}{d^4T} \sum_{t=1}^{T} \| e'_t \varepsilon^{b} \widehat{P}_{1}' \widehat{P}_{1} \Lambda_2 \|^4 \frac{1}{T^4} \| F\|^8
\nonumber
\\&=
\frac{1}{d^4T} \sum_{t=1}^{T} \| e'_t \varepsilon^{b} \widehat{P}_{1}' \widehat{P}_{1} \Lambda_2 \|^4 \mathcal{O}_{\prob}(1)
\label{al:fourthmoments_al:22}
\\&=
\mathcal{O}_{\prob}\mleft(\frac{1}{d^2}\mright)=
\mathcal{O}_{\prob}\mleft(\frac{1}{\delta_{dT}^4}\mright),
\label{al:fourthmoments_al:33} 
\end{align}
where \eqref{al:fourthmoments_al:22} is due to \ref{cond:s1} in the proof of Proposition \ref{prop2.1} and \eqref{al:fourthmoments_al:33} follows similarly to  \eqref{al:app:qq1aaaa} above.

Finally, the third summand in \eqref{app:eq:wqwqwqw12345} can be bounded as
\begin{align}
\frac{1}{T} \frac{1}{d^4T^4} \sum_{t=1}^{T} \| e'_t  \varepsilon^{b} \widehat{P}_{1}' \widehat{P}_{1} \varepsilon^{b'} F^b \|^4
&\leq
\frac{1}{d^2 T^3} \sum_{t=1}^{T} \| e'_t  \varepsilon \|^4 \frac{1}{d^2T^2} \| \varepsilon^{b'} F^{b} \|^4 \| \widehat{P}_{1}' \widehat{P}_{1} \|^4
\nonumber
\\&\leq
\frac{1}{d^2T^3} \sum_{t=1}^{T} \| e'_t  \varepsilon \|^4 \mathcal{O}_{\prob}(1)
\label{al:fourthmoments_al:222}
\\&=
\mathcal{O}_{\prob}\mleft(\frac{1}{T^2}\mright)
=\mathcal{O}_{\prob}\mleft(\frac{1}{\delta_{dT}^4}\mright),
\label{al:fourthmoments_al:333}
\end{align}
where the asymptotics of $(\varepsilon^{b'} F^b)/T$ in \eqref{al:fourthmoments_al:222} follow by the same calculations as done on p.\ 199 in \cite{doz2011two} as part of the proof of Lemma 2(i) under Assumptions  \ref{ass:2} and \ref{ass:C4}. Furthermore, \eqref{al:fourthmoments_al:333} follows by \ref{item:R3} below under Assumption \ref{ass:3}(iii).

\begin{enumerate}[label=\textbf{R.\arabic*.},ref=R.\arabic*, align=left]
\item \label{item:R1} 
Under Assumption \ref{ass:2}, $\frac{1}{T} \sum_{t=1}^{T} \| e'_t F \|^4 = \mathcal{O}_{\prob}(1)$ since
\begin{align}
&
\frac{1}{T} \sum_{t=1}^{T} \E \| e'_t F \|^4
\nonumber
\\&=
\frac{1}{T} \sum_{t=1}^{T} \E \mleft\| \sum_{k=0}^{\infty}
C_{k} a_{t-k} \mright\|^4
\nonumber
\\&\leq
\frac{1}{T} \sum_{t=1}^{T} \E \mleft( \sum_{k=0}^{\infty}
\| C_{k} \| \| a_{t-k} \| \mright)^4
\nonumber
\\&\leq
\frac{1}{T} \sum_{t=1}^{T} \sum_{k_1,k_2,k_3,k_4=0}^{\infty}
\| C_{k_1} \| \| C_{k_2} \| \| C_{k_3} \| \| C_{k_4} \| \E \mleft( \| a_{t-k_1} \| \| a_{t-k_2} \| \| a_{t-k_3} \| \| a_{t-k_4} \| \mright)
\nonumber
\\&\leq
\frac{1}{T} \sum_{t=1}^{T} \sum_{k_1,k_2,k_3,k_4=0}^{\infty}
\| C_{k_1} \| \| C_{k_2} \| \| C_{k_3} \| \| C_{k_4} \| \mleft( \E \| a_{t-k_1} \|^4 \E \| a_{t-k_2} \|^4 \E \| a_{t-k_3} \|^4 \E \| a_{t-k_4} \|^4 \mright)^{\frac{1}{4}}
\nonumber
\\&\leq 
c \mleft(\sum_{k=0}^{\infty} \| C_{k} \|\mright)^4
\label{al:dfdfdfdfdf}
\end{align}
for some constant $c$, where \eqref{al:dfdfdfdfdf} follows using Cauchy-Schwarz inequality
\begin{align*}
\E \| a_{t-k} \|^4 &= \E \mleft( \sum_{i=1}^{r} |a_{i,t-k}|^2 \mright)^2= \sum_{i_1,i_2=1}^{r} \E ( |a_{i_1,t-k}|^2 |a_{i_2,t-k}|^2) 
\\&\leq 
\sum_{i_1,i_2=1}^{r} ( \E ( |a_{i_1,t-k}|^4 ) \E ( |a_{i_2,t-k}|^4 ))^{\frac{1}{2}} < \infty 
\end{align*}
by Assumption \ref{ass:2}(i) and stationarity of $\{a_t\}$.
\item \label{item:R3}
Under Assumption \ref{ass:3}(iii), it holds that $\sum_{t=1}^{T} \| e'_t \varepsilon \|^4 = \mathcal{O}_{\prob}(d^2T)$ since
\begin{align}
\sum_{t=1}^{T} \E \| e'_t \varepsilon \|^4
&=
\E \sum_{t=1}^{T} \mleft( \sum_{u=1}^{d} \varepsilon^2_{t,u} \mright)^2
\nonumber
\\&=
\sum_{t=1}^{T} \sum_{u_1,u_2=1}^{d} \E ( \varepsilon^2_{t,u_1} \varepsilon^2_{t,u_2} )
\nonumber
\\&\leq
\sum_{t=1}^{T} \sum_{u_1,u_2=1}^{d} (\E ( \varepsilon^4_{t,u_1} ) \E( \varepsilon^4_{t,u_2} ))^{\frac{1}{2}}
\label{al:fourthmoments_al:R_3_al1} 
\\&\leq
d^2 T M,
\label{al:fourthmoments_al:R_3_al2} 
\end{align}
where \eqref{al:fourthmoments_al:R_3_al1} uses Cauchy-Schwarz inequality and \eqref{al:fourthmoments_al:R_3_al2} is due to Assumption \ref{ass:3}(iii).
\end{enumerate}

\textit{Proof of \eqref{le:fourthmoments_le:2}:}
Note that
\begin{align}
\frac{1}{T} \sum_{t=1}^{T} \| \widehat{F}_{t} \|^4
&=
\frac{1}{T} \sum_{t=1}^{T/2} \| \widehat{F}_{t} \|^4 + \frac{1}{T} \sum_{t=T/2 + 1}^{T} \| \widehat{F}_{t} \|^4
\nonumber
\\&=
\frac{1}{T} \sum_{t=1}^{T/2} \| \widehat{F}_{t} - H'_1 F_{t} + H'_1 F_{t} \|^4 + 
\frac{1}{T} \sum_{t=T/2 + 1}^{T} \| \widehat{F}_{t} - H'_2 F_{t} + H'_2 F_{t} \|^4
\nonumber
\\&\leq
\frac{8}{T} \sum_{t=1}^{T/2} \| \widehat{F}_{t} - H'_1 F_{t} \|^4 + \frac{8}{T} \sum_{t=1}^{T/2} \| H'_1 F_{t} \|^4
\nonumber
\\&\hspace{1cm}+
\frac{8}{T} \sum_{t=T/2 + 1}^{T} \| \widehat{F}_{t} - H'_2 F_{t} \|^4 + \frac{8}{T} \sum_{t=T/2 + 1}^{T} \| H'_2 F_{t} \|^4
=
o_{\prob}(1) + \mathcal{O}_{\prob}(1),
\label{al:fourthmoments_al:second_al1} 
\end{align}
where \eqref{al:fourthmoments_al:second_al1} is due to \eqref{le:fourthmoments_le:1}, Lemma \ref{le:convergenceHtoH0} and \ref{item:R1}.
\end{proof}

\section{Analysis under the alternative hypothesis} \label{app:C}
This appendix provides proofs of the results under the alternative hypothesis as stated in Section \ref{sec:alternative}.

Recall from \eqref{eq:with_change} the transformed series $Y$ and write it more compactly by concatenating over time as
\begin{align*}
\begin{pmatrix}
Y^{b} \\
Y^{a}
\end{pmatrix}
= 
\begin{pmatrix}
F^{b} \Lambda_{2}' \widehat{P}_{1}' \\
F^{a} \Lambda_{2}' \widehat{P}_{2}'
\end{pmatrix}
+
\begin{pmatrix}
\widetilde{\varepsilon}^{b} \\
\widetilde{\varepsilon}^{a}
\end{pmatrix}
:= 
\begin{pmatrix}
F^{b} \widetilde{\Theta}_{1}' \\
F^{a} \widetilde{\Theta}_{2}'
\end{pmatrix}
+
\begin{pmatrix}
\widetilde{\varepsilon}^{b} \\
\widetilde{\varepsilon}^{a}
\end{pmatrix}.
\end{align*}
Recall also the population model from \eqref{eq:with_change_popu} and write it similarly as
\begin{align} \label{al:true_model_A}
\begin{pmatrix}
Z^{b} \\
Z^{a}
 \end{pmatrix}
 &= 
 \begin{pmatrix}
F^{b} \Lambda_{2}' P_{1} \\
F^{a} \Lambda_{2}'
 \end{pmatrix}
 +
\begin{pmatrix}
\eta^{b} \\
\eta^{a}
 \end{pmatrix}
= 
 \begin{pmatrix}
F^{b} \Theta_{1}' \\
F^{a} \Theta_{2}'
 \end{pmatrix}
 +
\begin{pmatrix}
\eta^{b} \\
\eta^{a}
 \end{pmatrix}
=  
 \begin{pmatrix}
G^{b} B' \\
G^{a} C'
 \end{pmatrix}
\Theta'
  +
\begin{pmatrix}
\eta^{b} \\
\eta^{a}
 \end{pmatrix}
=: 
G \Theta' + \eta.
 \end{align}
To specify $\Theta$ in the population model further, we need to distinguish between the two types of breaks as introduced in \eqref{eq:type1BC} and \eqref{eq:type2BC} and briefly recalled next.

\textit{TA 1:} 
Suppose $\col(\Lambda_1) \cap \col(\Lambda_2) = \{0\}$ and $\rank( \Lambda_1' \Lambda_2)=r$. Then, according to Case \ref{item1a}, the alternative model admits a change of Type \ref{item:Type1} since $\col( P_{1} \Lambda_2) \cap \col(\Lambda_2) = \{0\}$ and we get
\begin{equation*}
\Theta = \left( \Theta_1, \Theta_2 \right) = \frac{1}{\sqrt{2}} \left( P_{1} \Lambda_2, \Lambda_2 \right).
\end{equation*}

\textit{TA 2:} 
Suppose $\col(\Lambda_1) \cap \col(\Lambda_2) \neq \{0\}$, $\col(\Lambda_1) \neq \col(\Lambda_2)$ and $\rank( \Lambda_1' \Lambda_2)=r$. Then, according to Case \ref{item2a}, the alternative model admits a change of Type \ref{item:Type3} since $\col( P_{1} \Lambda_2) \cap \col(\Lambda_2) \neq \{0\}$ and $\col( P_{1} \Lambda_2) \neq \col(\Lambda_2)$.
Suppose further that $\Lambda_{2} = (\Lambda_1 \Phi , \Pi_2)$ with $\Phi \in \RR^{r \times r_1}$, and $\Pi_2 \in \RR^{d \times r_2}$ with $r_2 = r-r_1$ being linearly independent but not orthogonal to $\Lambda_1$. Then, 
$P_{1} \Lambda_{2} = (\Lambda_1 \Phi , P_{1} \Pi_2)$ and some (but not all) columns of $\Theta_1 = P_{1} \Lambda_{2}$ and $\Theta_2 = \Lambda_{2}$ are linearly independent. Then, 
\begin{equation*}
\Theta = \mleft( \Lambda_1, \frac{1}{\sqrt{2}} \Pi_2, \frac{1}{\sqrt{2}} P_{1} \Pi_2 \mright).
\end{equation*}

The PCA estimators for the alternative models can be derived as in \eqref{eq:PCA1}, that is
\begin{align} \label{eq:PCA_A}
\widehat{G} = \sqrt{T} \widehat{Q}_{\npf},
\hspace{0.2cm} 
\widehat{\Theta}' = \frac{1}{T} \widehat{G}' Y,
\end{align}
where $\widehat{Q}_{\npf} $ are the eigenvectors corresponding to the $\npf$ largest eigenvalues of the $T \times T$ matrix $YY'$.
Furthermore, let $\widehat{\Pi}_{\npf}$ and $\widehat{V}_{\npf}$ denote respectively the diagonal matrices with the $\npf$ largest eigenvalues of $\frac{1}{T} YY'$ and $\frac{1}{dT} YY'$. We can also define
\begin{equation} \label{eq:PCA2_A}
\widebar{G} = \frac{1}{d} Y \widebar{\Theta}, \hspace{0.2cm} \widebar{\Theta} = \sqrt{d} \widebar{Q}_{\npf}
\end{equation}
with $\widebar{Q}_{\npf}$ being the $\npf$ eigenvectors corresponding to the $\npf$ largest eigenvalues of the $d \times d$ matrix $Y'Y$. 

On the population level as around \eqref{eq:def_Q}, we consider the eigendecomposition 
\begin{equation*}
\Theta' \Theta = R \Pi_{\Theta} R',
\end{equation*}
where the diagonal matrix $\Pi_\Theta$ consists of the eigenvalues in decreasing order and the orthogonal matrix $R$ consists of the corresponding eigenvectors. Choose the $d \times \npf$ matrix 
\begin{equation} \label{eq:def_Q_A}
Q = \Theta R \Pi_{\Theta}^{-\frac{1}{2}} 
\end{equation}
such that $\Theta \Theta' = Q \Pi_\Theta Q'$ and $Q'Q = I_{\npf}$. 

We follow the same line of proofs as under the null hypothesis. We first provide the proof for the behavior of our test statistic under the alternative hypothesis in Appendix \ref{se:mainresultsproofs_A}. We then prove auxiliary results in Appendix \ref{se:proofs_auxiliary_results_A}. Finally, Appendix \ref{se:proofs_consistency_PCA_A} states and proves consistency results for the PCA estimators under the alternative models \eqref{eq:type1BC} and \eqref{eq:type2BC}.

\subsection{Proof of result in Section \ref{sec:alternative}} \label{se:mainresultsproofs_A}

\begin{proof}[Proof of Proposition  \ref{propAlternative0}]

\textit{Proof of (i):} The result follows by Propositions \ref{propAlternative1} and \ref{propAlternative2} below and the same arguments as in the proof of Theorem 4 in \cite{han2015tests}.
Write 
\begin{align}
\frac{1}{T} \sum_{t=1}^{T/2} \widehat{G}_{t} \widehat{G}_{t}' - \frac{1}{T} \sum_{t=T/2 + 1}^{T} \widehat{G}_{t} \widehat{G}_{t}'
\nonumber
&=
\mleft( \frac{1}{T} \sum_{t=1}^{T/2} J' G_{t} G_{t}' J - \frac{1}{T} \sum_{t=T/2 + 1}^{T} J' G_{t} G_{t}' J \mright)
\nonumber \\
&\hspace{0.2cm}+
\frac{1}{T} \sum_{t=1}^{T/2} \mleft( \widehat{G}_{t} \widehat{G}_{t}' - J' G_{t} G_{t}' J \mright)
-
\frac{1}{T} \sum_{t=T/2 + 1}^{T} \mleft( \widehat{G}_{t} \widehat{G}_{t}' - J' G_{t} G_{t}' J \mright).
\label{eq:prop41_al1}
\end{align}
The second and third summands in \eqref{eq:prop41_al1} can be treated similarly, for instance,
\begin{align}
&
\frac{1}{T} \sum_{t=1}^{T/2} \mleft( \widehat{G}_{t} \widehat{G}_{t}' - J' G_{t} G_{t}' J \mright)
\nonumber
\\&=
\frac{1}{T} \sum_{t=1}^{T/2} \mleft( 
(\widehat{G}_{t} - J' G_{t}) (\widehat{G}'_{t} - G'_{t}J) + 
( \widehat{G}_{t} - J'G_{t}) G_{t}'J + 
J G_{t}' (\widehat{G}_{t}' - G'_{t} J )
\mright)
=
\mathcal{O}_{\prob}\mleft(\frac{1}{\delta_{dT}^2}\mright)
\nonumber
\end{align}
by Propositions \ref{propAlternative1} and \ref{propAlternative2}.

The asymptotic behavior of the first summand in \eqref{eq:prop41_al1} depends on the type of break. We get
\begin{align*}
\frac{1}{T} \sum_{t=1}^{T/2} J' G_{t} G_{t}' J - \frac{1}{T} \sum_{t=T/2 + 1}^{T} J' G_{t} G_{t}' J
\overset{\prob}{\to}
J_0'(D_1 - D_2)J_0
\end{align*}
with $D_1$, $D_2$ as in \eqref{eq:type1_pop_cov} and \eqref{eq:type2_pop_cov} and $J_0$ as the limit of $J$; see Lemma \ref{le:convergenceJtoJ0} below.

\textit{Proof of (ii):} We refer to the proof of Proposition \ref{propAlternative1} for a discussion on the assumptions. The bottom line is that all assumptions are satisfied for the alternative model and we can adopt the proof of Proposition \ref{prop3}.

Given all assumptions of Proposition \ref{prop3} are satisfied for the alternative model, we can infer
\begin{equation} \label{eq:am_omega}
\| \Omega(\widehat{G}) - \Omega(GJ_{0}) \| \overset{\prob}{\to} 0.
\end{equation}
Then,
\begin{align}
\| \Omega^{-1}(\widehat{G}) - \Omega^{-1}(GJ_{0}) \| 
&\leq
\| \Omega^{-1}(\widehat{G}) \| \mleft\| \Omega(\widehat{G}) - \Omega(GJ_{0}) \mright\| \| \Omega^{-1}(GJ_{0}) \|
\overset{\prob}{\to} 0
\label{eq:am_omega_inv}
\end{align}
by \eqref{eq:am_omega} and the same arguments as in \ref{item:Q3} (see \eqref{al:Q3_2}).

For the test statistic as defined in \eqref{eq:def:Wald}, we get
\begin{align}
W(\widehat{G})
&=
V(\widehat{G})' \Omega^{-1}(\widehat{G}) V(\widehat{G}) \nonumber
\\&=
\frac{T}{b_{T/2}} \frac{1}{\sqrt{T}} V(\widehat{G})' b_{T/2} \Omega^{-1}(\widehat{G}) \frac{1}{\sqrt{T}} V(\widehat{G}) \nonumber
\\&=
\frac{T}{b_{T/2}} \mleft( \vechop(\mathcal{C}) + o_{\prob}(1) \mright)' b_{T/2} \mleft( \Omega^{-1}(G J_0) + o_{\prob}(1) \mright) \mleft( \vechop(\mathcal{C}) + o_{\prob}(1) \mright)
\to \infty,
\label{eq:last_test_am}
\end{align}
where \eqref{eq:last_test_am} is due to \eqref{eq:am_omega_inv} and Assumption \ref{ass:P2_main}(i).
\end{proof}

\subsection{Proofs of auxiliary results under alternative hypothesis} \label{se:proofs_auxiliary_results_A}

We state below Propositions \ref{propAlternative1} and \ref{propAlternative2} which are analogous to Propositions \ref{prop2.1} and \ref{prop2.1.0} under the null hypothesis. For the statements, we need to introduce further notation including the crucial matrix $\widetilde{J}$ which is analogous to $\widetilde{H}$ in \eqref{eq:rep_FtildeHtilde} under the null hypothesis.

Similarly to $\widehat{F}^b, \widehat{F}^a$, set $\widehat{G}^b = \widehat{G}'_{1:T/2}$ and $\widehat{G}^b = \widehat{G}'_{(T/2+1):T}$ with $\widehat{G}$ as in \eqref{eq:PCA_A}. Recall, in particular, that $\widehat{G}$ is a $T \times \npf$ matrix and $\npf$ differs depending on the type of alternative.
We introduce the matrices
\begin{equation*} \label{def:J1J2}
\begin{aligned}
J_1 &:= \frac{1}{dT} (\Lambda_2' \widehat{P}_{1}' \widehat{P}_{1} \Lambda_2 F^{b'} \widehat{G}^{b} 
+ 
\Lambda_2' \widehat{P}_{1}' \widehat{P}_{2} \Lambda_2 F^{a'} \widehat{G}^{a})
\widehat{V}^{-1}_{\npf},
\\
J_2 &:= \frac{1}{dT} (\Lambda_2' \widehat{P}_{2}' \widehat{P}_{2} \Lambda_2 F^{b'} \widehat{G}^{b} 
+ 
\Lambda_2' \widehat{P}_{2}' \widehat{P}_{1} \Lambda_2 F^{a'} \widehat{G}^{a})
\widehat{V}^{-1}_{\npf}
\end{aligned}
\end{equation*}
such that
\begin{align}
\widetilde{G} \widetilde{J}
&:=
\begin{pmatrix}
F^b & 0
\\
0 & F^a
\end{pmatrix}
\begin{pmatrix}
J_1
\\
J_2
\end{pmatrix}
=
\begin{pmatrix}
F^b J_1
\\
F^a J_2
\end{pmatrix}
\nonumber
\\&=
\frac{1}{dT}
\widetilde{G}
\begin{pmatrix}
\Lambda_2' \widehat{P}_{1}' \widehat{P}_{1} \Lambda_2
&
\Lambda_2' \widehat{P}_{1}' \widehat{P}_{2} \Lambda_2
\\
\Lambda_2' \widehat{P}_{2}' \widehat{P}_{1} \Lambda_2
&
\Lambda_2' \widehat{P}_{2}' \widehat{P}_{2} \Lambda_2
\end{pmatrix}
\widetilde{G}'
\widehat{G} \widehat{V}^{-1}_{\npf}
\nonumber
\\&=:
\frac{1}{dT}
\widetilde{G} \Lambda_{P} \widetilde{G}'\widehat{G} \widehat{V}^{-1}_{\npf}.
\label{eq:rep_FtildeHtilde_A}
\end{align}
We write further 
\begin{equation*}
J = (\Theta' \Theta/d) (G' \widehat{G} /T) \widehat{V}^{-1}_{\npf}
\end{equation*}
with $\Theta$ as in \eqref{al:true_model_A}. Again, $\Theta$ differs depending on the respective type of alternative. Note also that $\widetilde{G}$ is the same as $\widetilde{F}$ in \eqref{eq:rep_FtildeHtilde}, and so is $\Lambda_{P}$.

The following propositions are analogous to Lemma 10(i) and (ii) in \cite{han2015tests}. 
 \begin{proposition}\label{propAlternative1}
Under Assumptions \ref{ass:1}--\ref{ass:2}, \ref{ass:C1}--\ref{ass:C3} and \ref{ass:P1_main},
\begin{equation*}
\frac{1}{T} \| \widehat{G} - G J \|_{F}^2 = 
\mathcal{O}_{\prob}\mleft(\frac{1}{\delta_{dT}^2}\mright), \text{ as } d,T\to\infty.
\end{equation*}
\end{proposition}

\begin{proof}[Proof of Proposition \ref{propAlternative1}]
The proof is analogous to the proof of Proposition \ref{prop2.1} but for the model \eqref{al:true_model_A} under the alternative. In the proof of Proposition \ref{prop2.1}, it is sufficient to impose assumptions on the original models \eqref{eq:model_unified}. The main difference between the proofs under the null hypothesis and the alternative is the following. Under the null hypothesis, the true model we estimate is the same as the one estimated through the transformed series. In contrast, under the alternative, we estimate the parameters in \eqref{al:true_model_A}, which are different from those in the original series $X^2$, the one we chose to transform. 
For this reason, the only additional assumption needed is Assumption \ref{ass:P1_main}. Assumption \ref{ass:P1_main} imposes conditions directly on the loadings of the population model \eqref{al:true_model_A}. 

We provide here the beginning of the proof and will then refer to the proof of Proposition \ref{prop2.1}. Note first that
\begin{align} \label{al:triangleGGtildeJJtilde}
\frac{1}{T} \| \widehat{G} - G J \|_{F}^2
\leq 
\frac{2}{T} \| \widehat{G} - \widetilde{G} \widetilde{J} \|_{F}^2 + \frac{2}{T} \| \widetilde{G} \widetilde{J} - GJ \|_{F}^2.
\end{align}
We consider the two summands in \eqref{al:triangleGGtildeJJtilde} separately. 
For the first summand in \eqref{al:triangleGGtildeJJtilde}, recall the transformed series $Y_{1:T} = \mleft( \widehat{P}_{1} X^{2}_{1:T/2}, \widehat{P}_{2} X^{2}_{ (T/2 +1) :T} \mright)$ from \eqref{eq:YT2def} as well as $Y':=Y_{1:T}$ from \eqref{eq:defofY}.
Due to the PCA estimators \eqref{eq:PCA_A}, we get $\frac{1}{dT} YY' \widehat{G} = \widehat{G} \widehat{V}_{\npf}$ such that 
$\frac{1}{dT} YY' \widehat{G} \widehat{V}^{-1}_{\npf} = \widehat{G} $.
Then, given \eqref{eq:rep_FtildeHtilde_A},
\begin{align}
\widehat{G} - \widetilde{G} \widetilde{J}
&=
\frac{1}{dT} YY' \widehat{G} \widehat{V}^{-1}_{\npf} - \frac{1}{dT} \widetilde{G} \Lambda_{P} \widetilde{G}' \widehat{G} \widehat{V}^{-1}_{\npf}
\nonumber
\\&=
\frac{1}{dT} 
( YY' - \widetilde{G} \Lambda_{P} \widetilde{G}' ) (\widehat{G} - G ) \widehat{V}^{-1}_{\npf} +
\frac{1}{dT} 
( YY' - \widetilde{G} \Lambda_{P} \widetilde{G}' ) G \widehat{V}^{-1}_{\npf}.
\label{align:both_begin_A}
\end{align}
As noted above, with $\widehat{F}$ in \eqref{eq:rep_FtildeHtilde},
\begin{align} \label{eq:block_representation_A}
YY' - \widetilde{G} \Lambda_{P} \widetilde{G}'
=
YY' - \widetilde{F} \Lambda_{P} \widetilde{F}'.
\end{align}
It was argued in the proof of Proposition \ref{prop2.1} that 
\begin{equation} \label{eq:G120.5}
\| \frac{1}{dT} (YY' - \widetilde{F} \Lambda_{P} \widetilde{F}') \| = \mathcal{O}_{\prob}(1).
\end{equation}
After applying the Frobenius norm and triangle inequality, we consider the two summands in \eqref{align:both_begin_A}
separately. 
For the first summand in \eqref{align:both_begin_A}, we get by \eqref{eq:block_representation_A}, 
\begin{align}
&\frac{1}{T} \| 
\frac{1}{dT} ( YY' - \widetilde{G} \Lambda_{P} \widetilde{G}' ) (\widehat{G} - G) \widehat{V}^{-1}_{\npf} \|_{F}^2
\nonumber
\\&\leq
\| \frac{1}{dT} ( YY' - \widetilde{F} \Lambda_{P} \widetilde{F}') \|^2 
\| \widehat{V}^{-1}_{\npf} \|^2
\frac{1}{T} \| \widehat{G} - G \|_{F}^2
\label{align:XX1_A}
\\&=
\mathcal{O}_{\prob}(1)
\mathcal{O}_{\prob}(1)
\mathcal{O}_{\prob}\mleft(\frac{1}{\delta^2_{dT}}\mright)
=
\mathcal{O}_{\prob}\mleft(\frac{1}{\delta^2_{dT}}\mright),
\label{align:XX4_A}
\end{align}
where \eqref{align:XX1_A} follows by \eqref{eq:block_representation_A} and \ref{item:M0} in Appendix \ref{se:matrixnorminequalities}, and    
\eqref{align:XX4_A} follows by \eqref{eq:G120.5}, Proposition \ref{prop2.2_A} below and Corollary \ref{cor:resultsforPI_A}.
We also omit the proof for the second summand in \eqref{align:both_begin_A} and refer to the proof of Proposition \ref{prop2.1}. 

It is left to consider the second summand in \eqref{al:triangleGGtildeJJtilde}.
Recall that $J = (\Theta' \Theta/d) (G' \widehat{G} /T) \widehat{V}^{-1}_{\npf}$. Then, with further explanations given below, we get
\begin{align}
&
\frac{1}{T} \| \widetilde{G} \widetilde{J} - GJ \|_{F}^2
\nonumber
\\&=
\frac{1}{T} \frac{1}{(dT)^2}
\mleft\| 
\widetilde{G} \Lambda_{P} \widetilde{G}' \widehat{G} \widehat{V}^{-1}_{\npf} -
G \Theta' \Theta G' \widehat{G} \widehat{V}^{-1}_{\npf}
\mright\|_{F}^2
\nonumber
\\&=
\frac{1}{T} \frac{1}{(dT)^2}
\mleft\| 
\widetilde{G}
\begin{pmatrix}
\Lambda_2' \widehat{P}_{1}' \widehat{P}_{1} \Lambda_2
&
\Lambda_2' \widehat{P}_{1}' \widehat{P}_{2} \Lambda_2
\\
\Lambda_2' \widehat{P}_{2}' \widehat{P}_{1} \Lambda_2
&
\Lambda_2' \widehat{P}_{2}' \widehat{P}_{2} \Lambda_2
\end{pmatrix}
\widetilde{G}'
\widehat{G} \widehat{V}^{-1}_{\npf}
 -
G \Theta' \Theta G' \widehat{G} \widehat{V}^{-1}_{\npf}
\mright\|_F^2
\nonumber
\\&\leq
\frac{1}{(dT)^2}
\mleft\| 
\begin{pmatrix}
F^b \Lambda_2' \widehat{P}_{1}' \widehat{P}_{1} \Lambda_2 F^{a'}
&
F^b \Lambda_2' \widehat{P}_{1}' \widehat{P}_{2} \Lambda_2 F^{b'}
\\
F^a \Lambda_2' \widehat{P}_{2}' \widehat{P}_{1} \Lambda_2 F^{a'}
&
F^a \Lambda_2' \widehat{P}_{2}' \widehat{P}_{2} \Lambda_2 F^{b'}
\end{pmatrix}
-
G \Theta' \Theta G' 
\mright\|^2
\frac{1}{T} \| \widehat{G} \|_{F}^2 \| \widehat{V}^{-1}_{\npf} \|^2
\label{al:proofH1_al2_A00}
\\&= 
\mathcal{O}_{\prob}\mleft(\frac{1}{\delta^2_{dT}}\mright).
\label{al:proofH1_al2_A}
\end{align}
For \eqref{al:proofH1_al2_A}, note first that
\begin{align} \label{al:proofH1_al3_A}
\begin{pmatrix}
F^b \Lambda_2' P_{1}' P_{1} \Lambda_2 F^{a'}
&
F^b \Lambda_2' P_{1}' P_{2} \Lambda_2 F^{b'}
\\
F^a \Lambda_2' P_{2}' P_{1} \Lambda_2 F^{a'}
&
F^a \Lambda_2' P_{2}' P_{2} \Lambda_2 F^{b'}
\end{pmatrix}
=
\begin{pmatrix}
F^b \Lambda_2' P_{1}' 
\\
F^a \Lambda_2'
\end{pmatrix}
\begin{pmatrix}
P_{1} \Lambda_2 F^{a'} & \Lambda_2 F^{b'}
\end{pmatrix}
= G \Theta' \Theta G',
\end{align}
where we used \eqref{al:true_model_A}. Combining \eqref{al:proofH1_al2_A00}, \eqref{al:proofH1_al3_A} with \ref{item:M1}, it is left to show the asymptotic behavior for three different summands.
We focus on one of them, for example, 
\begin{align}
&
\frac{1}{dT}
\mleft\| 
F^b \Lambda_2' \widehat{P}_{1}' \widehat{P}_{1} \Lambda_2 F^{a'}
-
F^b \Lambda_2' P_{1}' P_{1} \Lambda_2 F^{a'}
\mright\|
\nonumber
\\&\leq
\frac{1}{d}
\mleft\| 
\Lambda_2' \widehat{P}_{1}' \widehat{P}_{1} \Lambda_2
-
\Lambda_2' P_{1}' P_{1} \Lambda_2 \mright\| 
\frac{1}{\sqrt{T}} \| F^{a'} \| \frac{1}{\sqrt{T}} \| F^{b'} \|
\nonumber
\\&=
\mathcal{O}_{\prob} \mleft(\frac{1}{d}\mright) + \mathcal{O}_{\prob}\mleft(\frac{1}{\sqrt{T}}\mright),
\label{al:proofH1_al3}
\end{align}
where the asymptotic of $\frac{1}{\sqrt{T}} \| F^{a'} \|$ is due to \ref{cond:s1} in the proof of Proposition \ref{prop2.1}. The behavior of the estimated projection matrices in \eqref{al:proofH1_al3} follows under similar considerations as in \eqref{eq:crucial_addedlemma}.

The asymptotics of $\frac{1}{\sqrt{T}} \| \widehat{G} \|_F$ in \eqref{al:proofH1_al2_A00} is due to Proposition \ref{prop2.2_A} and \ref{cond:s1} in the proof of Proposition \ref{prop2.1}.
\end{proof}

\begin{proposition}\label{propAlternative2}
Under \ref{ass:1}--\ref{ass:3}(i), \ref{ass:C1}--\ref{ass:C3} and \ref{ass:P1_main},
\begin{equation*}
\frac{1}{T^2} \| ( \widehat{G} - GJ)'G \|^2_{F} = \mathcal{O}_{\prob}\mleft(\frac{1}{\delta_{dT}^4}\mright), \text{ as } d,T\to\infty.
\end{equation*}
\end{proposition}

\begin{proof}[Proof of Proposition \ref{propAlternative2}]
The proof is analogous to the proof of Proposition \ref{prop2.1.0} but for the models under the alternative hypothesis in Section \ref{sec:alternative}.
We refer to the proof of Proposition \ref{propAlternative1} above for the necessary adjustments of the proof of Proposition \ref{prop2.1.0}.
\end{proof}

The following is analogous to Lemma \ref{le:convergenceHtoH0}. We omit the proof since it follows similar lines as the proof of Lemma \ref{le:convergenceHtoH0}.
 
\begin{lemma} \label{le:convergenceJtoJ0}
Under Assumptions \ref{ass:1}--\ref{ass:2}, \ref{ass:C1}--\ref{ass:C4} and \ref{ass:P1_main},
\begin{equation*}
J - J_0 = \mathcal{O}_{\prob}\mleft(\frac{1}{\delta_{dT}}\mright)
\hspace{0.2cm}
\text{ and }
\hspace{0.2cm}
\| J_0 \| = \mathcal{O}(1),
\text{ as } d,T\to\infty,
\end{equation*}
where, with $\Theta \Theta' = Q \Pi_{\Theta} Q'$ and $ \widetilde{\Pi}$ as in Assumption \ref{ass:P1_main},  
\begin{equation*} \label{eq:def:J0}
J_0 = Q \widetilde{\Pi} Q' \Sigma_{G} \widetilde{\Pi}^{-1}.
\end{equation*}
\end{lemma}

\subsection{Results for consistency of PCA estimators under alternative hypothesis} \label{se:proofs_consistency_PCA_A}

This section is analogous to Appendix \ref{se:consistency_PCA_NH} but under the alternative hypothesis. We focus on presenting analogues of Lemma \ref{lem2.1} and Proposition \ref{prop2.2}. The remaining statements in Appendix \ref{se:consistency_PCA_NH} (Lemma \ref{lem2.2} and Corollaries \ref{cor:resultsforPI}--\ref{cor2.3}) are also stated for the alternative models but their proofs are omitted since they follow along the same lines as under the null hypothesis.

\begin{lemma}\label{lem2.1_A}
Set $\widehat{\Sigma}_{Y} = \frac{1}{T} Y'Y$. Then, under Assumptions \ref{ass:1}--\ref{ass:2}, \ref{ass:C1}--\ref{ass:C4} and \ref{ass:P1_main}, 
\begin{equation*}
\frac{1}{d} \| \widehat{\Sigma}_{Y} - \Theta \Theta' \| = \mathcal{O}_{\prob}\mleft( \frac{1}{d} \mright) + \mathcal{O}_{\prob} \mleft( \frac{1}{\sqrt{T}} \mright).
\end{equation*}
\end{lemma}

\begin{proof}
We follow again the proof of Lemma 2 in \cite{doz2011two}. Consider
\begin{equation} \label{pr:eq1_A}
\frac{1}{d} \| \widehat{\Sigma}_{Y} - \Theta \Theta' \|
\leq
\frac{1}{d} \| \widehat{\Sigma}_{Y} - \Sigma \|
+
\frac{1}{d} \| \Sigma - \Theta \Theta' \|
\end{equation}
with
\begin{align} \label{eq:Sigma-sep_A}
\Sigma 
=
\begin{cases}
\frac{1}{2} \Theta_1 \Theta_1' + \frac{1}{2} \Theta_2 \Theta_2' + 
\Sigma_{\eta},
& \hspace{0.2cm}
\text{if \textit{TA 1}}, \\
\Lambda_1 \Phi \Phi' \Lambda'_1 + \frac{1}{2} P_{1} \Pi_2 \Pi_2' P_{1}' + \frac{1}{2} \Pi_2 \Pi'_2 + 
\Sigma_{\eta},
& \hspace{0.2cm}
\text{if \textit{TA 2}}.
\end{cases}
\end{align}
where 
\begin{equation*}
\Sigma_{\eta} = 
P_{1}
\Sigma_{\varepsilon}
P_{1}'
+
P_{2}
\Sigma_{\varepsilon}
P_{2}'.
\end{equation*}
For the rest of the proof we focus on \textit{TA 2} since it is more involved. \textit{TA 1} follows by similar arguments. 

We consider the two summands in \eqref{pr:eq1_A} separately. For the second summand, we get
\begin{align} \label{eq:sigmaetaineq}
\frac{1}{d} \| \Sigma - \Theta \Theta' \|
=
\frac{1}{d} \mleft\| \Sigma_{\eta} \mright\|
\leq
\frac{1}{d} \| P_{1} \|^2 \| \Sigma_{\varepsilon} \| + \frac{1}{d} \| P_{2} \|^2 \| \Sigma_{\varepsilon} \|
\leq
\frac{2}{d} \| \Sigma_{\varepsilon} \|
\leq
O\mleft(\frac{1}{d} \mright)
\end{align}
due to Assumption \ref{ass:C4} and our observation \eqref{eq:operatornormprojection}.
For the first summand in \eqref{pr:eq1_A}, we separate the series according to the transformation \eqref{eq:YT2def} as
\begin{align}
\| \widehat{\Sigma}_{Y} - \Sigma \| \nonumber
&=
\mleft\| \frac{1}{2} \widehat{P}_{1} \frac{1}{T/2} \sum_{t=1}^{T/2} X^2_{t} X^{2'}_{t} \widehat{P}_{1}'
+ 
\frac{1}{2} \widehat{P}_{2} \frac{1}{T/2} \sum_{t=T/2+1}^{T} X^2_{t} X^{2'}_{t} \widehat{P}_{2}'
- 
\Sigma \mright\| \nonumber
\\&\leq
\frac{1}{2} 
\mleft\| \widehat{P}_{1} \frac{1}{T/2} \sum_{t=1}^{T/2} X^2_{t} X_{t}^{2'} \widehat{P}_{1}' - 
\Big( \Lambda_1 \Phi \Phi' \Lambda'_1 + P_{1} \Pi_2 \Pi_2' P_{1}' + \Sigma_{\eta} \Big)
\mright\|
\nonumber
\\&\hspace{1cm}+
\frac{1}{2} \mleft\| 
\widehat{P}_{2} \frac{1}{T/2} \sum_{t=T/2+1}^{T} X^2_{t} X^{2'}_{t} \widehat{P}_{2}' - 
\Big( \Lambda_1 \Phi \Phi' \Lambda'_1 + \Pi_2 \Pi'_2 + \Sigma_{\eta} \Big) 
\mright\|, \label{eq:sec:proofs:al1}
\end{align}
where we used \eqref{eq:Sigma-sep_A}.
We consider the two summands in \eqref{eq:sec:proofs:al1} separately. For the first summand in \eqref{eq:sec:proofs:al1}, with explanations given below,
\begin{align}
&
\mleft\| \widehat{P}_{1} \frac{1}{T/2} \sum_{t=1}^{T/2} X^2_{t} X_{t}^{2'} \widehat{P}_{1}' - 
\Big( \Lambda_1 \Phi \Phi' \Lambda'_1 + P_{1} \Pi_2 \Pi_2' P_{1}' + \Sigma_{\eta} \Big)
\mright\|
\nonumber
\\&\leq
\mleft\| \widehat{P}_{1} \Big( \widehat{\Sigma}_{2} - \Lambda_{2} \Lambda'_{2} \Big) \widehat{P}_{1}' \mright\|
+
\mleft\| \widehat{P}_{1}
\Lambda_{2} \Lambda'_{2} \widehat{P}_{1}' - 
\Big( \Lambda_1 \Phi \Phi' \Lambda'_1 + P_{1} \Pi_2 \Pi_2' P_{1}' + \Sigma_{\eta} \Big)
\mright\|
\nonumber
\\&\leq
\mleft\| \widehat{P}_{1} \Big( \widehat{\Sigma}_{2} - \Lambda_{2} \Lambda'_{2} \Big) \widehat{P}_{1}' \mright\|
+
\mleft\| \widehat{P}_{1} \Lambda_{2} \Lambda_{2}' \widehat{P}_{1}' - P_{1} \Lambda_{2} \Lambda_{2}' P_{1}' 
\mright\|
\nonumber
\\&\hspace{1cm}+
\mleft\| P_{1} \Lambda_{2} \Lambda_{2}' P_{1}'  - (\Lambda_1 \Phi \Phi' \Lambda'_1 + P_{1} \Pi_2 \Pi_2' P_{1}')
\mright\|
+ 
\| \Sigma_{\eta} \| 
\nonumber
\\&\leq
\| \widehat{P}_{1} \|^2 \| \widehat{\Sigma}_{2} - \Lambda_{2} \Lambda'_{2} \|
+
\mleft\| \widehat{P}_{1} \Lambda_{2} \Lambda_{2}' \widehat{P}_{1}' - P_{1} \Lambda_{2} \Lambda_{2}' P_{1}' 
\mright\|
+
\| \Sigma_{\eta} \| 
\label{al:SigmaY:eq1}
\\&\leq
\| \widehat{\Sigma}_{2} - \Lambda_{2} \Lambda'_{2} \|
+
\mleft\| \widehat{P}_{1} \Lambda_{2} \Lambda_{2}' \widehat{P}_{1}'- P_{1} \Lambda_{2} \Lambda_{2}' P_{1}' 
\mright\|
+
\| \Sigma_{\eta} \| 
\label{al:SigmaY:eq2}
\\&
=
d \mathcal{O}_{\prob}\Big( \frac{1}{d} \Big) + d \mathcal{O}_{\prob}\Big( \frac{1}{\sqrt{T}} \Big)
+\mathcal{O}(1),
\label{al:SigmaY:eq3}
\end{align}
where \eqref{al:SigmaY:eq1} uses the submultiplicativity of the spectral norm and $P_{1} \Lambda_{2} \Lambda_{2}' P_{1}'  = \Lambda_1 \Phi \Phi' \Lambda'_1 + P_{1} \Pi_2 \Pi_2' P_{1}'$, since $P_{1} \Lambda_{2} = (\Lambda_1\Phi, P_{1} \Pi_2)$.
Furthermore, 
\eqref{al:SigmaY:eq2} is due to \eqref{eq:operatornormprojection} and similarly for their estimated counterparts $\| \widehat{P}_{k} \| = 1 $, $k=1,2$.
Finally, \eqref{al:SigmaY:eq3} is due to \ref{Doz1} in Appendix \ref{se:appB}, similar arguments as in Lemma \ref{le:difference_projections} and $\| \Sigma_{\eta} \| = \mathcal{O}(1)$ by the same arguments as in \eqref{eq:sigmaetaineq}. The second summand in \eqref{eq:sec:proofs:al1} can be handled by analogous considerations.
\end{proof}

\begin{corollary} \label{cor:resultsforPI_A}
Set $\widehat{\Pi}_{\npf} = d\widehat{V}_{\npf}$ and recall that $\widehat{V}_{\npf}$ is a diagonal matrix consisting of the $\npf$ largest eigenvalues of $\frac{1}{dT} YY'$. Recall also the diagonal matrix $\Pi_{\Theta}$ from Assumption \ref{ass:P1_main}.
Under Assumptions \ref{ass:1}--\ref{ass:2}, \ref{ass:C1}--\ref{ass:C3} and \ref{ass:P1_main},
\begin{enumerate}[label=(\roman*)]
\item $\frac{1}{d} \|\widehat{\Pi}_{\npf} - \Pi_\Theta \| = 
\mathcal{O}_{\prob}\mleft(\frac{1}{d}\mright) + \mathcal{O}_{\prob}\mleft(\frac{1}{\sqrt{T}}\mright)$,
\item $d \|\widehat{\Pi}_{\npf}^{-1} - \Pi^{-1}_\Theta \| = 
\mathcal{O}_{\prob} \mleft(\frac{1}{d}\mright) + \mathcal{O}_{\prob}\mleft(\frac{1}{\sqrt{T}}\mright)$,
\item $\Pi_\Theta \widehat{\Pi}_{\npf}^{-1}-I_{\npf} = \mathcal{O}_{\prob} \mleft(\frac{1}{d}\mright) + \mathcal{O}_{\prob}\mleft(\frac{1}{\sqrt{T}}\mright)$.
\end{enumerate}
\end{corollary}

\begin{lemma}\label{lem2.2_A}
Let $\widehat{A}=(\widehat{a}_{ij})_{i,j=1,\ldots,\npf}=\widebar{Q}_{\npf}'Q$ with $Q$ defined in \eqref{eq:def_Q_A} and $\widebar{Q}_{\npf}$ defined in \eqref{eq:PCA_A}. Under Assumptions \ref{ass:1}--\ref{ass:2}, \ref{ass:C1}--\ref{ass:C4} and \ref{ass:P1_main},
$\widehat{a}_{ij} = \mathcal{O}_{\prob}\mleft(\frac{1}{d}\mright) + \mathcal{O}_{\prob}\mleft(\frac{1}{\sqrt{T}}\mright)$, $i\neq j$,
$\widehat{a}_{ii}^2 = 1 +\mathcal{O}_{\prob}\mleft(\frac{1}{d}\mright)+\mathcal{O}_{\prob}\mleft(\frac{1}{\sqrt{T}}\mright)$, $i=1,\ldots,\npf$.
\end{lemma}

\begin{corollary}\label{cor2.2_A}
Under Assumptions \ref{ass:1}--\ref{ass:2}, \ref{ass:C1}--\ref{ass:C4} and \ref{ass:P1_main}, one can take $\widebar{Q}_{\npf}$, such that
\begin{displaymath}
\widebar{Q}_{\npf}'Q = I_{\npf} + \mathcal{O}_{\prob}\mleft(\frac{1}{d}\mright) + \mathcal{O}_{\prob}\mleft(\frac{1}{\sqrt{T}}\mright),\quad 
\|\widebar{Q}_{\npf}-Q\|^2 = \mathcal{O}_{\prob}\mleft(\frac{1}{d}\mright) + \mathcal{O}_{\prob}\mleft(\frac{1}{\sqrt{T}}\mright).
\end{displaymath}
\end{corollary}

\begin{proposition}\label{prop2.2_A}
Under Assumptions \ref{ass:1}--\ref{ass:2}, \ref{ass:C1}--\ref{ass:C4} and \ref{ass:P1_main},
\begin{equation*}
\frac{1}{T} \| \widehat{G} - G \|_{F}^2 
= \frac{1}{T} \sum_{t=1}^{T} \| \widehat{G}_t-G_{t} \|^2_{F}
= \mathcal{O}_{\prob}\mleft(\frac{1}{d}\mright) + \mathcal{O}_{\prob}\mleft(\frac{1}{T}\mright), 
\text{ as } d,T\to\infty.
\end{equation*}
\end{proposition}

\begin{proof}
Using the representations \eqref{eq:PCA_A}, \eqref{eq:PCA2_A} and the definition \eqref{eq:YT2def},
\begin{equation*}
\widehat{G}_t = 
\widehat{V}_{\npf}^{-\frac{1}{2}} \widebar{G}_t = 
\widehat{V}_{\npf}^{-\frac{1}{2}} \frac{1}{d} \widebar{\Lambda}' Y_t = 
\widehat{V}_{\npf}^{-\frac{1}{2}} \frac{1}{\sqrt{d}} \widebar{Q}'_{\npf} Y_t = 
\begin{cases}
\widehat{\Pi}_{\npf}^{-1/2} \widebar{Q}_{\npf}' \widehat{P}_{1} X^2_{t},
\hspace{0.2cm}
&\text{ for }
\hspace{0.2cm}
t = 1, \dots, T/2,
\\
\widehat{\Pi}_{\npf}^{-1/2} \widebar{Q}_{\npf}' \widehat{P}_{2} X^2_{t},
\hspace{0.2cm}
&\text{ for }
\hspace{0.2cm}
t = T/2+1, \dots, T.
\end{cases}
\end{equation*}
We consider only the case $t = 1, \dots, T/2$. Then,
\begin{align}
&
\widehat{G}_t-G_{t} 
\nonumber
\\&= 
\widehat{\Pi}_{\npf}^{-1/2} \widebar{Q}_{\npf}' \widehat{P}_{1} X^2_{t} - G_{t}
\nonumber
\\&= 
\widehat{\Pi}_{\npf}^{-1/2}\widebar{Q}_{\npf}' \widehat{P}_{1} ( \Lambda_2 F_{t} + \varepsilon_{t}) - G_{t}
\nonumber
\\&= 
\widehat{\Pi}_{\npf}^{-1/2}\widebar{Q}_{\npf}' ( \widehat{P}_{1} \Lambda_2 - P_{1} \Lambda_2)F_{t} + 
\widehat{\Pi}_{\npf}^{-1/2}\widebar{Q}_{\npf}' P_{1} \Lambda_2 F_{t} - G_t +
\widehat{\Pi}_{\npf}^{-1/2}\widebar{Q}_{\npf}' \widehat{P}_{1} \varepsilon_{t} 
\nonumber
\\&= 
\widehat{\Pi}_{\npf}^{-1/2}\widebar{Q}_{\npf}' ( \widehat{P}_{1} \Lambda_2 - P_{1} \Lambda_2)F_{t} + 
(\widehat{\Pi}_{\npf}^{-1/2}\widebar{Q}_{\npf}' \Theta - I_{\npf} ) G_t +
\widehat{\Pi}_{\npf}^{-1/2}\widebar{Q}_{\npf}' \widehat{P}_{1} \varepsilon_{t} 
\label{align_xsxsxsx0_A}
\\&= 
\widehat{\Pi}_{\npf}^{-1/2}\widebar{Q}_{r}' (\widehat{P}_{1} \Lambda_2 - P_{1} \Lambda_2)F_{t} 
+ \widehat{\Pi}_{\npf}^{-1/2} ( \widebar{Q}'_{\npf} Q - \widehat{\Pi}_{\npf}^{1/2} \Pi_{\Theta}^{-1/2} ) \Pi_{\Theta}^{1/2} G_{t} 
+ \widehat{\Pi}_{\npf}^{-1/2}\widebar{Q}_{\npf}' \widehat{P}_{1} \varepsilon_{t},
\label{align_xsxsxsx1_A}
\end{align}
where \eqref{align_xsxsxsx0_A} follows by \eqref{al:true_model_A}.
We consider the three summands in \eqref{align_xsxsxsx1_A} separately.
For the first one, note that
\begin{align*}
\frac{1}{T} \sum_{t=1}^{T/2} \| \widehat{\Pi}_{\npf}^{-1/2}\widebar{Q}_{\npf}'  (\widehat{P}_{1} \Lambda_2 - P_{1} \Lambda_2)F_{t} \|_{F}^2
&\leq
\| \widehat{\Pi}_{\npf}^{-1/2}\widebar{Q}_{\npf}' (\widehat{P}_{1} \Lambda_2 - P_{1} \Lambda_2)\|^2 \frac{1}{T} \| F \|^2_{F}
\\&\leq
d \| \widehat{\Pi}_{\npf}^{-1/2} \|^2 \| \widebar{Q}_{\npf}' \|^2 
\frac{1}{d} \| (\widehat{P}_{1} \Lambda_2 - P_{1} \Lambda_2)\|^2 \frac{1}{T} \| F \|^2_{F}
\\&=
\mathcal{O}_{\prob} \mleft(\frac{1}{d^2}\mright) + \mathcal{O}_{\prob}\mleft(\frac{1}{T}\mright),
\end{align*}
since $\frac{1}{T} \| F \|^2_{F} \leq  \frac{r}{T} \| F \|^2 = \mathcal{O}_{\prob}(1)$ by \ref{cond:s1} in the proof of Proposition \ref{prop2.1}, $\widehat{\Pi}_{\npf}^{-1/2} = \frac{1}{\sqrt{d}} \mleft( \frac{1}{d} \widehat{\Pi}_{\npf} \mright)^{-1/2} = \mathcal{O}_{\prob}\mleft(\frac{1}{\sqrt{d}}\mright)$ by Corollary \ref{cor:resultsforPI_A} and the asymptotic of $\frac{1}{d} \| \widehat{P}_{1} \Lambda_2 - P_{1} \Lambda_2 \|^2 $ can be proved as in \eqref{al:qqppqqpp01}. 

We omit the remainder of the proof concerning the other two terms in \eqref{align_xsxsxsx1_A} since it follows the exact same arguments as the proof of Proposition \ref{prop2.2}.
\end{proof}

\section{Matrix norm inequalities} \label{se:matrixnorminequalities}
We collect here some results on matrix norms that are used throughout the paper.
\begin{enumerate}[label=\textbf{M.\arabic*.},ref=M.\arabic*, align=left]
\item 
\label{item:M0}
$\|AB\|_{F} \leq \| A \| \| B \|_{F}$ for matrices $A,B$; see Theorem 1 in \cite{Fangmatrices1994}. 
\item 
\label{item:M1}
For a general block matrix (not necessarily positive semidefinite), we know that
\begin{align}
\mleft\|
\begin{pmatrix}
A & C \\
C' & B
\end{pmatrix}
\mright\|^2
&=
\lambda_{\max}
\begin{pmatrix}
AA' + CC' & AC + CB' \\
C'A' + BC' & BB' + CC'
\end{pmatrix}
\nonumber
\\&=
\mleft\|
\begin{pmatrix}
AA' + CC' & AC + CB' \\
C'A' + BC' & BB' + CC'
\end{pmatrix}
\mright\|
\nonumber
\\&\leq 
\| AA' + CC' \| + \| BB' + CC' \|
\label{eq:matrix0011}
\\&\leq 
\|A\|^2 + \| B \|^2 + 2 \|C\|^2,
\label{eq:block_matrix_inequality}
\end{align}
where \eqref{eq:matrix0011} is due to the unitary invariance of the operator norm and Lemma 1.1. in \cite{bourin2012decomposition}. For \eqref{eq:block_matrix_inequality}, see Weyl's Theorem (Theorem 4.3.1 in \cite{horn2012matrix}). 
\item 
\label{item:M3}
Let $A,B$ be Hermitian and suppose that $B$ is positive semidefinite. Then,
$\lambda_{\max}(A) \leq \lambda_{\max}(A + B)$ by Corollary 4.3.12 in \cite{horn2012matrix}.
\item 
\label{item:M4}
For matrices $A,B$,
\begin{align*}
\mleft\|
\begin{pmatrix}
A \\
B
\end{pmatrix}
\mright\|^2
&=
\lambda_{\max}
\begin{pmatrix}
A'A + B'B 
\end{pmatrix}
= \| A'A + B'B \|
\leq \|A\|^2 + \| B \|^2
\end{align*}
again by Weyl's Theorem (Theorem 4.3.1 in \cite{horn2012matrix}). 
\end{enumerate}

